\documentclass{llncs}
\usepackage{amssymb}
\usepackage{amsmath}
\usepackage{multirow}
\usepackage{wrapfig}
\usepackage{cite}
\usepackage{hyperref}
\usepackage{xcolor}
\usepackage{xspace}
\usepackage{tikz}
\usetikzlibrary{shapes,snakes}
\usepackage{longtable}
\usepackage{atbegshi}
\AtBeginDocument{\AtBeginShipoutNext{\AtBeginShipoutDiscard}}

\begin{document}


￼
￼
\title{Specification and Reactive Synthesis of Robust Controllers}

\author{{Paritosh K. Pandya$^1$} \and {Amol Wakankar$^2$}}
 
\institute{
Tata Institute of Fundamental Research, 
Mumbai 400005, India.\\
\email{pandya@tifr.res.in}
\and
Homi Bhabha National Institute, 
Mumbai 400094, India.\\
\email{amolk@barc.gov.in}
}
\maketitle


\newtheorem{fact}[theorem]{Fact}
\newtheorem{obs}{Observation}[theorem]
\newtheorem{sublemma}[obs]{Lemma}
\newtheorem{algorithm} {Algorithm}

\newcommand{\seq}[1]{\langle #1 \rangle}
\newcommand{\Nat}{\mathbb N}
\newcommand{\Real}{\mathbb R}
\newcommand{\Bb}{\mathbb B}
\newcommand{\Ff}{\mathcal F}
\newcommand{\Vv}{\mathcal V}
\newcommand{\px}{\mathbf{x}}
\newcommand{\py}{\mathbf{y}}
\newcommand{\Int}{\mathbb{Z}}
\newcommand{\set}[1]{\left\{ #1 \right\}}
\newcommand{\Set}[1]{\big\{ #1 \big\}}
\newcommand{\subst}[2]{#1/#2} 
\let\longimplies=\implies
\let\longiff=\iff
\renewcommand{\iff}{\Leftrightarrow}

\newcommand{\fif}{\mbox{\bf iff}}

\newcommand{\comment}[1]{\tcp{#1}}
\newcommand{\lcomment}[1]{\tcp{#1}}
\newcommand{\mcomment}[1]{}

\newcommand{\commentsp}[1]{{\color{black!40}  #1}}

\newcommand{\sem}[1]{ [ \! [ {#1}  ]  \! ]} 
\def\rmdef{\stackrel{\mbox{\rm {\tiny def}}}{=}} 

\newcommand{\Up}[1]{{\mathsf{Up}(#1)}}
\newcommand{\Ups}[1]{{\mathrm{Up}^+(#1)}}
\newcommand{\Sline}{\mathrm{Skyline}}
\newcommand{\Sup}{\mathsf{Supp}}
\newcommand{\Nbd}{\mathsf{Nbd}}
\newcommand{\CONF}{\textsc{Conf}}

\newcommand{\tab}{}

\newcommand{\cb}[2]{\ensuremath{{\mathtt{Cb{#1}}[{#2}]}}}
\newcommand{\cba}[2]{\ensuremath{{\mathtt{N{#1}}}[{#2}]}}
\newcommand{\cbr}[2]{\ensuremath{{\mathtt{r{#1}}}[{#2}]}}

\newcommand{\bad}[1]{{\mathrm{bad}_{#1}}}
\newcommand{\false}{\textsf{false}}
\newcommand{\mono}{\textsc{MonoSkolem}}
\newcommand{\cegar}{\textsc{CegarSkolem}}
\newcommand{\bloqqer}{\textsf{Bloqqer}}
\newcommand{\qrat}{\textsf{qrat-trim}}
\newcommand{\type}{\textsc{TYPE-}}
\newcommand{\dontprintsemicolon}{} 
\newcommand{\printsemicolon}{} 

\newcommand {\buchi} { B\"{u}chi }
\newcommand {\winf} { $\Sigma^\omega$ }
\newcommand {\wfin} { $\Sigma^*$ }
\newcommand {\transition} {$\rightarrow$}


\newcommand{\0}{\ensuremath{\textsf 0}}
\newcommand{\1}{\ensuremath{\textsf 1}}
\newcommand{\2}{\ensuremath{\textsf 2}}
\newcommand{\3}{\ensuremath{\textsf 3}}
\newcommand{\x}{\ensuremath{\textsf x}}

\newcommand{\U}{\mathcal{U}~}
\newcommand{\G}{\mathcal{G}~}
\newcommand{\F}{\mathcal{F}~}
\newcommand{\X}{\mathcal{X}~}

\newcommand{\A}{\ensuremath{\mathcal{A}}\xspace}
\newcommand{\B}{\ensuremath{\mathcal{B}}~}
\newcommand{\D}{\ensuremath{\mathcal{D}}\xspace}
\newcommand{\M}{\ensuremath{\mathcal{M}}~}
\newcommand{\W}{\ensuremath{\mathcal{W}} }
\newcommand{\MS}{\ensuremath{\mathcal{S}}~}
\newcommand{\R}{\ensuremath{\mathcal{R}}\xspace}
\newcommand{\Shield}{\ensuremath{\mathcal{S}}\xspace}
\newcommand{\wn}[1] {\textsf{#1}~}
\newcommand{\wnarg}[1] {\textsf{#1}}

\newcommand{\dcsynth}{DCSynth\xspace}
\newcommand{\Areg}{\ensuremath{\mathcal{A}_{reg}}~}
\newcommand{\bigo}[1] {\ensuremath{O(#1)~}}
\newcommand{\focc}[1] {\ensuremath{\Xi(#1) }}

\newcommand{\init} {\ensuremath{\mbox{\textbf{init}}(D_1:\Theta_1~/~ D_2:\Theta_2)}}
\newcommand{\triggers} {\ensuremath{\mbox{\textbf{triggers}}(D_1:\Theta_1\leadsto D_2:\Theta_2/D_3:\Theta_3)}}
\newcommand{\imp} {\ensuremath{\mbox{\textbf{implies}}(D_1:\Theta_1\leadsto D_2:\Theta_2)}}
\newcommand{\follows} {\ensuremath{\mbox{\textbf{follows}}(D_1:\Theta_1\leadsto D_2:\Theta_2/D_3:\Theta_3)}}
\newcommand{\pref}[1] {\ensuremath{\mbox{\textbf{pref}}(#1)~}}
\newcommand{\ppref}[1] {\ensuremath{\mbox{\textbf{ppref}}(#1)~}}
\newcommand{\anti}[1] {\ensuremath{\mbox{\textbf{anti}}(#1)}}
\newcommand{\spec}{\ensuremath{spec~}}

\newcommand{\lags}[1] {\ensuremath{\textsf{lags(#1)}~}}
\newcommand{\tracks}[1] {\ensuremath{\textsf{tracks(#1)}~}}
\newcommand{\latency}[1] {\ensuremath{\textsf{latency(#1)}~}}
\newcommand{\deadtime}[1] {\ensuremath{\textsf{deadtime(#1)}~}}

\newcommand{\wf} {\ensuremath{\textsf{Wf~}}}
\newcommand{\modelsv} {\ensuremath{\models_\nu}}

\newcommand {\Or} {\ensuremath{\vee}}
\newcommand {\union} {\ensuremath{\cup}}
\newcommand {\intersection} {\ensuremath{\cap}}

\renewcommand {\iff} {\ensuremath{\Leftrightarrow}}

\newcommand {\qddc} {QDDC\xspace}
\newcommand {\qddcn} {QDDC+N}
\newcommand {\qddcnl} {QDDC+NL}
\newcommand {\lhrs} {LHRS}

\newcommand {\ltl} {LTL}
\newcommand {\itl} {ITL}
\newcommand {\pltl} {PLTL}
\newcommand {\nnltl} {N\ensuremath{\tilde{\mathrm{N}}}LTL}

\newcommand {\ang}[1] {\ensuremath{\langle#1\rangle}}
\newcommand {\sq}[1] {\ensuremath{[#1]}}
\newcommand {\dsq}[1] {\ensuremath{[[#1]]}}
\newcommand {\dcurly}[1] {\ensuremath{\{\{#1\}\}}}
\newcommand {\defeq} {\ensuremath{\stackrel{\mathrm{def}}{\equiv}}}

\newcommand {\len}[1] {\ensuremath{len(#1)}}
\newcommand {\lenminus}[1] {\ensuremath{len^-(#1)}}
\newcommand {\intv}[1] {\ensuremath{Intv(#1)}}
\newcommand {\slen} {\ensuremath{slen} }
\newcommand {\scount} {\ensuremath{scount} }
\newcommand {\sdur} {\ensuremath{sdur} }
\newcommand {\pt} {\ensuremath{pt}~}
\newcommand {\ext} {\ensuremath{ext}~}
\newcommand {\true} {\ensuremath{true}~}

\newcommand {\dom}[1] {\ensuremath{dom(#1)}}
\newcommand {\doml}[1] {\ensuremath{dom\lambda(#1)}}

\newcommand {\nat} {\ensuremath{\mathbb{N}}}
\newcommand {\expected} {\ensuremath{\mathbb{E}}}

\newcommand {\eqv}[1] {\ensuremath{[[#1]]}}

\newcommand{\oomit}[1]{}
\newcommand{\nom} {\ensuremath{nom}}
\newcommand{\chop}{\^}

\newcommand{\df}{=}

\renewcommand*{\thefootnote}{\arabic{footnote}}

\newcommand{\DCSYNTH}{DCSynth\xspace}
\newcommand{\invariant}{\mbox{\bf inv\/}}
\newcommand{\MPNC}{MPS\xspace}
\newcommand{\GODSC}{MPHOS\xspace}
\newcommand{\detMPNC}{detMPS\xspace}
\newcommand{\detGODSC}{detMPHOS\xspace}
\newcommand{\TYPEZERO}{TYPE0\xspace}
\newcommand{\TYPEONE}{MPS\xspace}
\newcommand{\TYPETWO}{TYPE2\xspace}
\newcommand{\TYPETHREE}{MPHOS\xspace}
\newcommand{\Rb}{Rb}
\newcommand{\RbSpec}{\ensuremath{Rb_{Spec}}}

\newcommand{\st}{.~}

\begin{abstract}
This paper investigates the synthesis of robust controllers from
logical specification of regular properties given in an interval temporal logic \qddc. Our specification encompasses both \emph{hard robustness} and 
\emph{soft robustness}. Here, {\em hard robustness} guarantees invariance of commitment under user specified \emph{relaxed (weakened) assumptions}. A systematic framework for logically specifying the assumption weakening by means of a formula, called Robustness criterion, is presented. The {\em soft robustness} pertains to the ability of the controller to maintain the commitment for as many inputs as possible, irrespective of any assumption.
We present a uniform method for the synthesis of a robust controller which guarantees the  specified hard robustness and it optimizes the specified soft robustness. The method is implemented using a tool \DCSYNTH, which provides soft requirement optimized controller synthesis. Through the case study of a synchronous bus arbiter, we experimentally show the impact of hard robustness criteria as well as soft robustness on the ability of the synthesized controllers to meet the commitment ``as much as possible''. Both, the worst-case and the expected case behaviors are  analyzed.
\end{abstract}

\oomit{begin{abstract}
This paper proposes an interval temporal logic based framework to encode and synthesize robust controllers from regular specification given in logic \qddc.
A controller specification consists of a pair of \qddc formulas $(D_A,D_C)$ giving the regular \emph{assumption} and \emph{commitment}, respectively. Standard correctness  mandates
that $D_C$ should hold invariantly provided $D_A$ holds invariantly.
A  relaxed assumption $Rb(D_A)$ specifies a weaker condition  than the invariance of $D_A$ in past, and  a robust controller should satisfy $D_C$ whenever $Rb(D_A)$ holds. We term this as {\em hard robustness}. 
We formulate various robustness criteria (spanning existing and new notions) in logic \qddc, and  show their effect in improving hard robustness.
We also propose a notion of \emph{Soft robustness}, which pertains to the ability of the controller to
meet $D_C$ even when $Rb(D_A)$ does not hold. It tries to satisfy $D_C$ ``as much as possible'' irrespective of the assumption. A uniform synthesis method from various relaxed assumptions using the framework of soft requirement guided synthesis is given. This is implemented in a tool \DCSYNTH. We also give a formal comparison of various robustness notions formulated in \qddc. 
We experimentally show the impact of hard and soft robustness on the ability of the synthesized controllers to meet the commitments.
\end{abstract}
}


\section{Introduction}
\label{section:intro}

In this paper we consider automatic synthesis of robust controllers from regular specification given using an interval temporal logic \qddc (Quantified Discrete Duration Calculus)\cite{Pan01a}.
A regular property $D$ conceptually specifies a deterministic finite automaton $A(D)$ over finite words. The property holds at a point in behavior
if the past of the point is accepted by $A(D)$. 
Such properties hold intermittently in the behaviour. 
\qddc is a highly succinct and powerful
logic for specifying regular properties \cite{MPW17}. Formally, \qddc has  the expressive power of regular languages.

A controller specification consists of a pair of \qddc formulas $(D_A,D_C)$  giving the specification of \emph{assumption} and \emph{commitment}, respectively.
A standard correctness criterion, termed BeCorrect \cite{BCGHHJKK14},  mandates
that in all the behaviour of the synthesized controller, the commitment $D_C$ should hold invariantly provided the assumption $D_A$ holds invariantly. This can be denoted as $(\mathit{pref(D_A) \Rightarrow pref(D_C)})$. Here
$\mathit{pref(D)}$ holds at a point in behaviour if the regular property $D$ has been invariantly true in the past.

{\em Robustness} pertains to the ability of $D_C$ holding even when $D_A$ does not hold invariantly in the past \cite{BCGHHJKK14}. 
A {\em relaxed assumption} $\mathit{Rb(D_A)}$ specifies weaker condition than $\mathit{pref(D_A)}$ and a robust specification can be given by the formula 
$(\mathit{Rb(D_A) \Rightarrow D_C})$ that should be satisfied invariantly. Thus, $D_C$ should hold whenever the relaxed assumption $Rb(D_A)$ holds. We term this as {\bf hard robustness}.
For example, given integer parameters $k$ and $b$, the relaxed assumption $LenCntInt(D_A,k,b)$ holds at a point $i$ if the assumption $D_A$ is violated at most $k$ times in the interval $[i-(b+1),i]$ spanning $b$ previous cycles from the current point $i$. The controller synthesized under the relaxed assumption $LenCntInt(D_A,k,b)$ would be more robust as it will tolerate upto $k$ assumption violations in recent past. 

As a complementary notion of robustness, we propose  {\textbf{soft robustness}}, which  pertains to the ability of a controller to
meet $D_C$ even when the relaxed assumption $Rb(D_A)$ does not hold. The controller synthesis technique should try to satisfy $D_C$ ``as much as possible'' irrespective of the assumption \cite{BEJK14}. 
Note that soft robustness is a global optimization problem \cite{Bel57}. For example, the commitment $D_C$ may be such that satisfying $D_C$ at a current point may prevent it from holding at most points in future. Hence the controller must choose not to satisfy it at the current point. 

We structure the formulation of relaxed assumption $Rb(D_A)$ (used to provide hard robustness) as a pair $(Rb(A),D_A)$ where $D_A$ is the concrete assumption formulated by the user. {\bf Robustness criterion} $Rb(A)$, which is a \qddc\/ formula over a fresh proposition $A$, 
specifies a generic method of relaxing any assumption $D_A$. A notion of cascade composition $Rb(A) \ll Ind(D_A,A)$ (formalized in this paper using logic \qddc), 
gives us the
desired relaxed assumption formula $Rb(D_A)$. We develop a theory of robustness criteria and its impact on the robustness of resulting controllers by defining \emph{dominance} relation between them.

We show that logic \qddc can be used to conveniently and systematically formulate a wide variety of robustness criteria $Rb(A)$. The interval logic modalities, bounded counting constraints and second order quantification features of logic \qddc\/ are particularly helpful in the formulation of various robustness criteria. We provide a list of several useful robustness  criteria, spanning some existing as well as novel notions. 

Thus, we have a logic based specification of hard robustness. 
We can order these robustness criterion by their weakness (in logical implication order). We show that a weaker robustness criterion specifies a more robust controller.
Hence, to get improved hard robustness, the designer must relax the concrete assumption $D_A$ with the weakest robustness criterion for which the controller remains realizable. We will illustrate this design aspect in our case studies.

Using the framework of soft requirement guided synthesis implemented in a tool \DCSYNTH \cite{WPM19}, we give a method to synthesize robust controller which guarantees the hard robustness and it optimizes the soft robustness. Tool \DCSYNTH\/ uses techniques from optimal control of Markov Decision Processes \cite{Put94,Bel57} for optimizing the expected value of soft robustness.  

We show the impact of hard and soft robustness on the ability of the synthesized controller to maintain the desired commitment $D_C$ for as many inputs as possible.
We give the case studies of a Synchronous bus arbiter (and a Mine pump controller specification in Appendix \ref{section:minepumpcasestudy}). Our experiments show that soft robustness greatly improves the \textit{expected value} of the commitment holding on random independent and identically distributed (iid) inputs. Moreover, hard robustness gives the worst case guarantees on meeting the commitments. Indeed we get multiple controllers with the same expected value of the commitment $D_C$ over random (iid) inputs. But these controllers satisfy $D_C$ for incomparable sets of inputs (under the subset ordering). These controllers guarantee $D_C$ under
diverse relaxed assumptions, thus providing distinct hard guarantees. Thus, our synthesis technique,  combining the hard and the soft robustness, seems useful.

The rest of the paper is arranged as follows. Section \ref{section:qddc} describes syntax and semantics of the Logic \qddc and the necessary definitions for robust controller synthesis. Section \ref{section:dcsynth-spec}  gives the syntax of \DCSYNTH specification and brief outline of synthesis method. 
Section \ref{section:robustness} introduces the  theory of robust specifications and its controller synthesis. 
Section \ref{section:motivation} describes
Arbiter case study and the corresponding experimental results. 
In Section \ref{section:discussion}, we conclude the paper with major contribution and related work.

\section{Quantified Discrete Duration Calculus (\qddc) Logic}
\label{section:qddc}
Let $PV$ be a finite non-empty set of propositional variables. 
Let $\sigma$ a non-empty finite word over the alphabet
$2^{PV}$. It has the form 
$\sigma=P_0\cdots P_n$ where $P_i\subseteq PV$ for each $i\in\{0,\ldots,n\}$. 
Let $\len{\sigma}=n+1$, 
$\dom{\sigma}=\{0,\ldots,n\}$, 
$\sigma[i,j]=P_i \cdots P_j$ and $\sigma[i]=P_i$.

The syntax of a \emph{propositional formula} over variables $PV$ is given by:
\[
\varphi := false\ |\ \true\ |\ p\in PV |\ !\varphi\ |\ \varphi~\&\&~\varphi\ |\ \varphi~||~\varphi
\]
with $\&\&, ||, !$ denoting conjunction, dis-junction and negation, respectively. Operators
such as $\Rightarrow$  and $\Leftrightarrow$ are defined as usual. 
Let $\Omega(PV)$ be the set of all propositional formulas over variables $PV$. 
%
Let $i\in\dom{\sigma}$. 
Then the satisfaction of propositional formula $\varphi$ at point $i$, denoted $\sigma,i\models\varphi$ is defined as usual and omitted here for brevity.
\oomit{
$\forall i\in\dom{\sigma}: \sigma, i \models \true$, $\sigma, i \models p$ iff $p\in\sigma[i]$, 
and $\sigma,i \models - \varphi$ iff $i>0$ and $\sigma, i-1 \models \varphi$. 
The satisfaction of
boolean combinations \verb#!# (not), \verb#&&# (and), \verb#||# (or) is  defined in a natural way.
}

The syntax of a \qddc formula over variables $PV$ is given by: 
\[
\begin{array}{lc}
D:= &\ang{\varphi}\ |\ \sq{\varphi}\ |\ \dsq{\varphi}\ |\ 
D\ \verb|^|\ D\ |\ !D\ |\ D~||~D\ |\ D~\&\&~D\ \\ 
&ex~p.\ D\ |\ all~p.\ D\ |\ slen \bowtie c\ |\ scount\ \varphi \bowtie c\ |\ sdur\ \varphi \bowtie c 
\end{array} 
\]
where $\varphi\in\Omega(PV)$, $p\in PV$, 
$c ~ \in\nat$ and $\bowtie\in\{<,\leq,=,\geq,>\}$. 

An \emph{interval} over a word $\sigma$ is of the form $[b,e]$ 
where $b,e\in\dom{\sigma}$ and $b\leq e$. 
Let $\intv{\sigma}$ be the set of all intervals over $\sigma$.
Let $\sigma$ be a word over $2^{PV}$ and let $[b,e]\in\intv{\sigma}$ be an interval. 
Then the satisfaction relation of a \qddc formula $D$ over
$\Sigma$ and interval $[b,e]$ written as $\sigma,[b,e]\models D$, is defined inductively as follows:
\[
\begin{array}{lcl}
\sigma, [b,e]\models\ang{\varphi} & \mathrm{\ iff \ } & b=e \mbox{ and } \sigma,b\models \varphi,\\
\sigma, [b,e]\models\sq{\varphi} & \mathrm{\ iff \ } & b<e \mbox{ and }
                                           \forall b\leq i<e:\sigma,i\models \varphi,\\
\sigma, [b,e]\models\dsq{\varphi} & \mathrm{\ iff \ } & \forall b\leq i\leq e:\sigma,i\models \varphi,\\
\sigma, [b,e]\models\dcurly{\varphi} & \mathrm{\ iff \ } & e=b+1 \mbox{ and }\sigma,b\models \varphi,\\
\sigma, [b,e]\models D_1\verb|^| D_2 & \mathrm{\ iff \ } & \exists b\leq i\leq e:\sigma, [b,i]\models D_1\mbox{ and }\sigma,[i,e]\models D_2,\\
\end{array}
\]
with Boolean combinations $!D$, $D_1~||~D_2$ and $D_1~\&\&~D_2$  defined in the expected way. 
We call word $\sigma'$ a $p$-variant, $p\in PV$, of a word $\sigma$ 
if $\forall i\in\dom{\sigma},\forall q\neq p:q\in \sigma'[i]\iff q\in \sigma[i]$. 
Then $\sigma,[b,e]\models ex~p.~D\iff\sigma',[b,e]\models D$ for some 
$p$-variant $\sigma'$ of $\sigma$ and 
$(all~p.~D) \Leftrightarrow (!ex~p.~!D)$. 
%

%


Entities \slen, \scount and \sdur are called \emph{terms}. 
The term \slen gives the length of the interval in which it is 
measured, $\scount\ \varphi$ where $\varphi\in\Omega(PV)$, counts 
the number of positions including the last point 
in the interval under consideration where $\varphi$ holds, and    
$\sdur\ \varphi$ gives the number of positions excluding the last point 
in the interval where $\varphi$ holds. 
Formally, for $\varphi\in\Omega(PV)$ we have 
$\slen(\sigma, [b,e])=e-b$, $\scount(\sigma,\varphi,[b,e])=\sum_{i=b}^{i=e}\left\{\begin{array}{ll}
					1,&\mbox{if }\sigma,i\models\varphi,\\
					0,&\mbox{otherwise.}
					\end{array}\right\}$ and 
$\sdur(\sigma,\varphi,[b,e])=\sum_{i=b}^{i=e-1}\left\{\begin{array}{ll}
					1,&\mbox{if }\sigma,i\models\varphi,\\
					0,&\mbox{otherwise.}
					\end{array}\right\}$

We also define the following derived constructs: 
$pt=\langle true \rangle$, $ext=!pt$, $\mathbf{ \langle \rangle D} = true\verb|^|D\verb|^|true$, $[]D=(!\langle \rangle!D)$ and $\mathbf{pref(D)}=!((!D)\verb|^|true)$. 
Thus, $\sigma, [b,e] \models []D$ iff $\sigma, [b',e']\models D$ 
for all sub-intervals $b\leq b'\leq e'\leq e$ and $\sigma, [b,e] \models \mathit{pref}(D)$ 
iff $\sigma, [b,e']\models D$ for all prefix intervals $b \leq e' \leq e$.

Finally, we define
$\sigma,i \models D$ iff $\sigma, [0,i] \models D$, and 
$\sigma \models D$ iff $\sigma, [0,$ $\len{\sigma}-1] \models D$ 
and $L(D) = \{ \sigma\mid\sigma \models D \}$, the set of behaviours accepted by $D$. Let $D$ be valid, denoted $\models_{dc} D$, iff $L(D) =
(2^{PV})^+$.

\begin{theorem}
\label{theorem:formula-automaton}
 \cite{Pan01a} For every formula $D$ over variables $PV$ we can construct a Deterministic Finite Automaton (DFA) $\A(D)$ over alphabet $2^{PV}$ 
such that $L(\A(D))=L(D)$. We call $\A(D)$ a \emph{formula automaton} for $D$ or the monitor automaton for $D$. \qed
\end{theorem}
A tool DCVALID implements this formula automaton construction in an efficient manner by internally using the tool MONA \cite{Mon01}. 
It gives {\em minimal, deterministic} automaton (DFA) for the formula $D$.
\oomit{
MONA uses multi-terminal BDDs for efficiently representing the automata and performing operations such as product, projection, determinization and minimization, thereby permitting automaton construction for fairly large formulas, inspite of poor worst case complexity. 
}
We omit the details here. However,  the reader may refer to several papers  on \qddc\/ 
for detailed description and examples of \qddc\/ specifications as well as  its model checking tool DCVALID \cite{Pan01a,MPW17,Pan01b}.

\vspace{-0.4cm}
\subsection{Supervisors and Controllers}
Now we consider \qddc\/  formulas and automata where variables $PV=I \cup O$ are partitioned into disjoint sets of input variables $I$ and output variables $O$. We show how Mealy machines can be represented as special form of Deterministic finite automata (DFA). Supervisors and controllers are Mealy machines with special properties. This representation allows us to use the MONA DFA library \cite{Mon01} to efficiently compute supervisors and controllers in our tool \DCSYNTH.
\begin{definition}[Output-nondeterministic Mealy Machines]
\label{def:nondetmm}
A total and Deterministic Finite Automaton (DFA) over input-output alphabet $\Sigma=2^I \times 2^O$ is a tuple $A=(Q,\Sigma,s,\delta,F)$, as usual, with $\delta:Q \times 2^I \times 2^O \rightarrow Q$. An {\bf output-nondeterministic Mealy machine} is a DFA with a unique reject (or non-final) state $r$ 
which is a sink state i.e. $F= Q - \{r\}$ and $\delta(r,i,o)=r$ for all $i \in 2^I$, $o \in 2^O$. 
\qed
\end{definition}

Intuition is that the transitions from $q \in F$ to $r$ are
forbidden (and kept only for making the DFA total).
Language of any such Mealy machine is prefix-closed. 
Recall that for a Mealy machine, $F=Q-\{r\}$.
A Mealy machine is 
{\bf deterministic} if $\forall s \in F$, $\forall i \in 2^I$, $\exists$ at most one $o \in 2^O$  s.t. $\delta(s,i,o) \not=r$. 
An output-nondeterministic Mealy machine is called {\bf non-blocking} if $\forall s \in F$, $\forall i \in 2^I$ $\exists o \in 2^I$ s.t.
$\delta(s,i,o) \in F$. It follows that 
for all input sequences a non-blocking Mealy machine can produce one or more
output sequence without ever getting into the reject state. 

For a Mealy machine $M$ over variables $(I,O)$, its language $L(M) \subseteq (2^I \times 2^O)^*$. A word $\sigma \in L(M)$ can also be represented as pair $(ii,oo) \in ((2^I)^*,(2^O)^*) $ such that
$\sigma[k] = ii[k] \cup oo[k], \forall k \in dom(\sigma)$. Here $\sigma, ii, oo$ must have the same length. We will not distinguish between $\sigma$ and $(ii,oo)$ in the rest of the paper.
Also, for any input sequence $ii \in (2^I)^*$, we will define $M[ii] = \{ oo ~\mid~ (ii,oo) \in L(M) \}$. 

\begin{definition}[Controllers and Supervisors]
 An output-nondeterministic Mealy machine which is non-blocking is called a {\bf supervisor}.  An deterministic supervisor is called a {\bf controller}.
 \qed
\end{definition}
The non-deterministic choice of outputs in a supervisor denotes unresolved decision.
The determinism ordering below allows supervisors to be refined into controllers.

\oomit{\color{blue}
See Figure \ref{fig:mpnc} in Appendix \ref{section:2cellexample} for an example of a output-nondeterministic Mealy machine which is non-blocking.
We represent output-nondeterministic Mealy machines and controllers as DFAs to take advantage of well established 
BDD based  DFA libraries such as those in tool MONA \cite{Mon01}. This library is used to represent and manipulate the supervisors in tool \DCSYNTH.
}

\begin{definition}[Determinism Order and Sub-supervisor]
 Given two supervisors $Sup_1, Sup_2$ we say that $Sup_2$ is {\em more deterministic} than $Sup_1$, denoted $Sup_1 \leq_{det} Sup_2$,
 iff $L(Sup_2) \subseteq L(Sup_1)$.  We call $Sup_2$ to be a {\em sub-supervisor} of 
 $Sup_1$. \qed
 \end{definition}
Note that being supervisors, they are both non-blocking, and hence 
$\emptyset \subset Sup_2[ii] \subseteq Sup_1[ii]$ for any $ii \in (2^I)^*$.
The supervisor $Sup_2$ may make use of additional memory for resolving and pruning
the non-determinism in $Sup_1$. 

\begin{definition}[Must Dominance]
\label{def:mustDominance}
Given a supervisor $Sup$ and a \qddc formula $D$, both over input-output alphabet $(I,O)$, let the set of must-satisfying inputs $MustInp(Sup,D) = \{ ii \in (2^I)^+ ~\mid~ 
\forall oo.~ ((ii,oo) \in L(Sup) ~\Rightarrow~ (ii,oo) \models_{dc} D) \}$. 
Given two supervisors $Sup_1, Sup_2$ and a \qddc formula $D$, we say that $Sup_2$ {\em must dominates} $Sup_1$ w.r.t. $D$, denoted by $Sup_1 \leq_{must}^{D} Sup_2$, iff\\ $MustInp(Sup_1,D) ~\subseteq~ MustInp(Sup_2,D)$.
\qed 
\end{definition}
The following proposition states that any arbitrary resolution of non-determinism in supervisors preserves the must-guarantees. 
\begin{proposition}
\label{prop:detmust}
  $Sup_1 \leq_{det} Sup_2$ implies that  $Sup_1 \leq_{must}^{D} Sup_2$ for any $D \in \qddc$. \qed
\end{proposition}

For technical convenience, we define a notion of \textit{indicator variable} for a \qddc\/ formula
(regular property). The idea is that the indicator variable $w$ witnesses the truth of a formula $D$ at any point in execution. Thus, $Ind(D,w) \df pref(EP(w) \Leftrightarrow D)$. Here, $\mathbf{EP(w)} = (true \verb#^# \langle w \rangle)$, i.e. $EP(w)$ holds at a point if variable $w$ is true at that point. Hence, $w$ will be $true$ exactly on those points where $D$ is $true$. These indicator variables can be used as auxiliary propositions in another formula using the notion of cascade composition $\ll$ defined below.
\begin{definition}[Cascade Composition]
\label{def:indDef}
Let $D_1, \ldots, D_k$ be \qddc\/ formulas over input-output variables $(I,O)$ and let $W=\{w_1, \ldots, w_k\}$ be the corresponding set of fresh indicator variables i.e. $(I \cup O) \cap W = \emptyset$. Let $D$ be a formula over variables $(I \cup  O \cup W)$. Then, the cascade composition $\ll$ and its equivalent \qddc\/ formula are as follows:
\[
  D \ll \langle Ind(D_1,w_1), \ldots, Ind(D_k,w_k) \rangle
 \quad  \df \quad D \land \bigwedge_{1 \leq i \leq k} ~ pref(EP(w_i) \Leftrightarrow D_i)
\]
This composition gives a formula over input-output variables $(I,O \cup W)$. \qed
\end{definition}
Cascade composition provides a useful ability to modularize a formula using auxiliary
propositions $W$ which witness regular properties given as \qddc formulas. 
\begin{example}
\label{exam:cascade}
Let $\verb#Rb(A)# = (\verb#scount !A <= 3#$) which holds at a point provided the proposition $\verb#A#$ is false at most 3 times in entire past. Let formula $\verb#D'# = \verb#true^<!req>^(slen=1) || true^<ack>#$ which holds at a point provided either the current point satisfies $\verb#ack#$ or the previous point does not satisfy $\verb#req#$. Then, $\verb#Rb(A)# \ll \verb#Ind(D',A)#$ is equivalent to the formula $\verb#(scount !A <=3) && #$  $\verb#pref(EP(A) <=> D')#$. This states that at each point, $\verb#A#$ holds provided $\verb#D'#$ holds, and the whole formula holds provided count of \verb#!A# in past is at-most $3$. Thus the whole formula holds provided $\verb#D'#$ is false at most $3$ times in the entire past. \qed
\end{example}

\section{\DCSYNTH Specification and  Controller Synthesis}
\label{section:dcsynth-spec}

This section gives a brief overview of the soft requirement guided controller synthesis method from \qddc formulas. More details can be found in
the original paper \cite{WPM19}. The method is implemented in a tool \DCSYNTH. This method and the tool will be used for robust controller synthesis in the subsequent sections.
\begin{definition}
\label{def:realizableAndMPNC}
A supervisor  $Sup$ {\em realizes invariance} of  \qddc\/ formula $D$ over variables $(I,O)$, denoted as  ${\bf Sup~ \mathbf{realizes ~\invariant}~ D}$, provided $L(Sup) \subseteq L(D)$. Recall that, by the definition of supervisors, $Sup$ must be non-blocking.
The supervisor $Sup$ is called {\bf maximally permissive} provided for any supervisor
$Sup'$ such that $(Sup' ~\mbox{realizes} ~\invariant~ D)$ we have $Sup \leq_{det} Sup'$. Thus, no other supervisor with larger languages realizes the invariance of $D$.
This $Sup$ is unique upto language equivalence of automata, and the minimum state
maximally permissive supervisor is denoted $\MPNC(D)$. \qed
\end{definition}
A well-known greatest fixed point algorithm for safety synthesis over $A(D)$ gives us 
$\MPNC(D)$ if it is realizable. We omit the details here (see \cite{WPM19}).
\begin{proposition}[\MPNC Monotonicity]
\label{prop:mpncSupervisor}
Given \qddc formulas $D_1$ and $D_2$ over variables $(I,O)$ such that $\models_{dc} (D_1 \Rightarrow D_2)$, we have:
\begin{itemize}
\item $\MPNC(D_2)  \leq_{det} \MPNC(D_1)$, and
\item If $\MPNC(D_1)$ is realizable then $\MPNC(D_2)$ is also realizable.
\qed
\end{itemize}
\end{proposition}
\medskip

A {\bf \DCSYNTH specification} is a tuple 
$(I,O,D^h, D^s)$ where $I$ and $O$ are the set of {\em input} and {\em output} variables, respectively. Formula $D^h$, called the \emph{hard requirement}, and formula $D^s$, called the \emph{soft requirement}, are \qddc\/ formulas over the set of propositions $PV=I \cup O$. 
The  objective in \DCSYNTH\/ is to synthesize a deterministic controller which (a) {\bf invariantly} satisfies the
hard requirement $D^h$, and (b) {\bf optimally} satisfies $D^s$ for as many inputs as possible.

The \DCSYNTH specification $(I,O,D^h, D^s)$  is said to be realizable iff $\MPNC(D^h)$ is realizable (i.e. it exist). The synthesis method first computes $\MPNC(D^h)$. In the next step a sub-supervisor of $\MPNC(D^h)$  which satisfies $D^s$ for ``as many inputs as possible'' is computed. This is formalized using a notion of $H$-optimality w.r.t. the soft requirement $D^s$. Let $H$ be a natural number called horizon. We construct a sub-supervisor of $MPS(D^s)$ by pruning  non-deterministic choice of outputs by retaining  only the outputs
which give the highest expected count of (intermittent) occurrence of $D^s$ over the next $H$ steps of execution. This count is averaged over all input sequences of length $H$. A well known value-iteration algorithm due to Bellman \cite{Bel57}, adapted from optimal
strategy synthesis for Markov Decision Processes, gives us the required sub-supervisor.  Book \cite{Put94} gives a proof of 
$H$-optimality of the value iteration algorithm and resulting supervisor. This construction is implemented in our tool \DCSYNTH. See \cite{WPM19} for details of the algorithm and its symbolic BDD based implementation. The resulting sub-supervisor is denoted by $\GODSC(D^h,D^s,H)$. 
As $\GODSC$ is obtained by pruning non-deterministic outputs in $\MPNC$, we have the following proposition \cite{WPM19}.
\begin{proposition}
 $\MPNC(D^h) \leq_{det} \GODSC(D^h,D^s,H)$ \qed
\end{proposition}
Using Proposition \ref{prop:detmust}, this shows that $\GODSC$ retains all the must-guarantees of $\MPNC(D^h)$. Hence,
$\GODSC(D^h,D^s,H)$ invariantly satisfies $D^h$ and by construction it is $H$-optimal w.r.t $D^s$.
Note that $\GODSC(D^h,D^s,H)$ itself can be non-deterministic as from a state and on a given input there may be more than
one choice of $H$-optimal output. All of these optimal choices are retained in $\GODSC$, giving the {\bf maximally permissive} $H$-optimal sub-supervisor of $\MPNC(D^h)$ w.r.t. $D^s$. Any controller obtained by arbitrarily resolving the non-determinism in $\GODSC(D^h,D^s,H)$ is also $H$-optimal. In tool \DCSYNTH, we allow users to specify a preference ordering $Ord$ on the set of outputs  $2^O$. Any supervisor $Sup$ can be determinized  by 
retaining only the highest ordered output amongst those permitted by $Sup$. This is denoted  $Det_{Ord}(Sup)$. In tool \DCSYNTH, the output ordering is specified by giving  a lexicographically ordered list of output variable literals. This facility is used to determinize $\GODSC(D^h,D^s,H)$ and  $\MPNC(D^h)$ supervisors as required.
\begin{example}
For a supervisor $Sup$ over variables $(I,\{o_1,o_2\})$, an output ordering can be
given as list ($o_1 > ~!o_2$),  Then, the determinization step will select the highest allowed  output in  from ($o_1=true, o_2=false$), ($o_1=true, o_2=true$), ($o_1=false, o_2=false$), ($o_1=false, o_2=true$) in that order. This choice is made for each state and each input. \qed
\end{example} 
In summary, given a \DCSYNTH\/ specification $(I,O,D^h, D^s)$, a horizon value $H$ and a preference ordering $ord$ on outputs $2^O$, the tool \DCSYNTH outputs
maximally permissive supervisors $\MPNC(D^h)$ and $\GODSC(D^h,D^s,H)$ as well as deterministic controllers $Det_{Ord}(\MPNC(D^h))$ and $Det_{ord}(\GODSC(D^h,D^s,H))$.

\oomit{
\paragraph{\textbf{\DCSYNTH specification with indicator definition}}
We propose a syntactic extension of \DCSYNTH specification.
We define \DCSYNTH specification with indicator definitions as tuple \\
\begin{equation}
S=( I,O,W,D^h,D^s) \ll \langle Ind(D_1, w_1) \ldots Ind(D_k, w_k)  \rangle
\end{equation}
Here $D^h,D^s$ are \qddc\/ formulas over variables $I \cup O \cup W$ and
$D_i$ are formulas over variables $I \cup O$ as in the cascade composition Definition \ref{def:indDef}.
\oomit{
It uses additional propositional variables $W=\{w_1, \dots w_k\}$ to track some \qddc formulas using indicator definition \ref{def:indDef}. 
Furthermore, \emph{hard requirement ($D^h$)} and \emph{soft requirement ($D^s$)} formulas are always defined using propositions $w_i \in W$.
So, for \DCSYNTH specification with indicator definitions the 
automaton for $D^h$ and $D^s$ has the alphabet $2^{I \cup O \cup W}$.
}

The specification $S$ above, is defined to be equivalent to the \DCSYNTH\/ specification 
$\hat{S}$ given below:
\[
  ( I,O \cup W, ~~D^h  \ll \langle Ind(D_1, w_1) \ldots Ind(D_k, w_k)  \rangle, 
       ~~ D^s  \ll \langle Ind(D_1, w_1) \ldots Ind(D_k, w_k)  \rangle) 
\]
Hence, $\MPNC$ and $\GODSC$ supervisors can be computed for specification $S$
by first reducing it to $\hat{S}$. Tool \DCSYNTH\/ permits specification with indicator definitions. In the rest of the paper, with slight abuse of notation, we shall not distinguish between $S$ and $\hat{S}$.
}

\oomit{
\begin{proposition}
Every \DCSYNTH specification with indicator definition has an equivalent \DCSYNTH specification, which can be given as follows:
For the specification $( I,O,D^h,D^s) \ll \langle Inddef(D_1, w_1) \ldots Inddef(D_k, w_k)  \rangle$ the corresponding \DCSYNTH specification can be given by $( I,O,D^h  \ll \langle Inddef(D_1, w_1) \ldots Inddef(D_k, w_k)  \rangle, D^s  \ll \langle Inddef(D_1, w_1) \ldots Inddef(D_k, w_k)  \rangle) $

\end{proposition}
}

\oomit{
{\color{red}
Given a specification $(I,O,D^h,D^s)$, a horizon value $H$ (a natural number) and 
a total ordering $<_{ord}$ on the set of outputs $2^O$, the controller synthesis in \DCSYNTH\/ has the following steps: 

\begin{enumerate}
\item Language equivalent DFA $A(D^h)$ and $A(D^s)$  are constructed for formulas $D^h$ and $D^s$. The indicating monitor $A^{Ind}(D^s,w)$ converts the soft requirement DFA $A(D^s)$ into a Mealy machine with same
 states and transitions as $A(D^s)$ but with output $w$ where each transition sets $w=1$ iff its target state is an accepting state of $A(D^s)$. See 
 Appendix \ref{section:supervisor} for its construction.
 
 \item  The maximally permissive supervisor $\MPNC(D^h)$ is constructed  by computing a greatest fixed point over the automaton $A(D^h)=\langle S,2^{I\union O},\delta,F\rangle$ using the standard safety synthesis algorithm \cite{GTW02}. We first compute the \emph{largest} set of winning states $G \subseteq F$ with the following property: $s \in G$ iff $\forall i\exists o:\delta(s,(i,o)) \in G$. Let 
$Cpre(A(D^h),X)= \{s ~\mid~ \forall i \exists o:\delta(s,(i,o)) \in X \}$. 
Then algorithm $\mathit{ComputeWINNING(A(D^h),I,O)}$ iteratively computes $G$ as follows: \\
\hspace*{1cm} G=F; {\em do} G1=G; G=Cpre($A^(D^h)$,G1) {\em while} (G != G1); 

If initial state $s \notin G$, then the specification is \emph{unrealizable}. 
Otherwise, $\MPNC(D^h)$ is obtained by making $G$ the set of final states, retaining
all the transitions in $A(D^h)$ between states in $G$ and redirecting the remaining
transitions of $A(D^h)$ to a unique reject state $r$ which is made a sink state. 

\item  The product $A^{Arena} = \MPNC(D^h) \times A^{Ind}(D^s, w)$ gives the supervisor on which $H$-optimal controller synthesis is carried out, for a given $H$, using the well-known value-iteration algorithm of Bellman \cite{Bel57}. In this algorithm a function $Val(s,p)$  is computed iteratively to assign a value to each state $s$ of $A^{Arena}$ automaton. Here $0 \leq p \leq H$ denotes the iteration number. Constant $0 \leq 
\gamma \leq 1$ is the discounting factor which can be taken as $\gamma=1$ in this paper for simplicity.
For $o \in 2^{(O \cup \{w\})}$ let $wt(o)=1$ if $w\in o$ and $0$ otherwise.
\[
 \begin{array}{l}
 Val(s,0) = 0 \\
 Val(s, p+1) = E_{i \in 2^I}~~ max_{o \in 2^{(O \cup \{w\})} ~:~ \delta(s,(i,o)) \neq r}~  \\
 \hspace*{4cm} \{ wt(o) ~+~ \gamma \cdot Val(\delta(s,(i,o)), p) \} 
\end{array}
\]
Having computed  $Val(s,H)$, the set of  $H$-optimal outputs $O_{max}$  is obtained as follows: 
For each state $s \in A^{Arena}$ and each input $i\in 2^I$, 
\[
 \begin{array}{l}
 O_{max} = \{ o ~\mid~ o=argmax_{o \in 2^O} \{ wt(o) ~+~ \gamma \cdot val(s,H) \\
 \hspace*{2cm} ~\mid~  \delta_{Arena}(s,i,o) = s' \land s'\not=r \} 
\end{array}
\]
Note that $O_{max}$ is a set as more than one output $o$ may satisfy the $argmax$ condition.
Now, supervisor $A^{Arena}$ is pruned by to retain only the transitions with  optimal outputs in set $O_{max}$. This gives us  {\em Maximally permissive $H$-Optimal supervisor} for $D^s$. 
The computation of this supervisor is denoted by $\GODSC(A^{Arena},H)$. This supervisor is denoted by $\GODSC(D^h,D^s,H)$.

\item  The non-deterministic choice of outputs in above $\GODSC$ is resolved in favour of highest
ordered output under the ordering $<_{ord}$.  This gives us the final deterministic controller $Cnt$. 

\end{enumerate}

 The controller $Cnt$ mandatorily satisfies $D^h$ invariantly, and it intermittently, but $H$-optimally, satisfies $D^s$. 

 At all stages of above synthesis, the automata and the supervisors $A(D^h)$, $A(D^s)$, $\MPNC(D^h)$, $A^{Arena}$, $\GODSC(A^{Arena},H)$ and $Cnt$ 
are all represented as semi-symbolic automata (SSDFA) using the MONA \cite{Mon01} DFA data structure. In this representation, the transition function is represented as a multi-terminal BDD. MONA DFA library provides a rich set of automata operations including product, projection, determinization, minimization over the semi-symbolic DFA. We adapt the greatest fixed point computation of $\MPNC(D^h)$ as well as the value iteration algorithm over $A^{Arena}$ to work symbolically over SSDFA. Moreover, at each stage of computation, the automata and supervisors are aggressively minimized.  
}
}

\oomit{
\subsection{Assumptions and Controller Types}
\label{sec:controllertype}

In most of the synthesis examples, we can formulate a desired regular property $C$, termed {\bf commitment}, which the controller should satisfy for as many input sequences as possible. Ideally, this property should be satisfied invariantly. But this may be unrealizable, and
a suitable  {\bf assumption} $A$ on the behaviour of environment may have to be made for $C$ to
hold. Given this pair $(A,C)$ of \qddc\/ formulas over input-output variables $(I,O)$, we specify four standard controller specifications $(I,O,D^h,D^s)$ as follows.
\begin{center}
\begin{tabular}{|c|c|c|}
\hline
Type & Hard Requirement $D^h$ & Soft Requirement $D^s$ \\
\hline
\TYPEZERO & $C$ i.e. ($true \Rightarrow C$) & $true$ \\
\hline
\TYPEONE & $(A \Rightarrow C)$ & $true$ \\ 
\hline
\TYPETWO & $true$ i.e. ($false \Rightarrow C$) & $C$ \\
\hline
\TYPETHREE & $(A \Rightarrow C)$ & C \\
\hline
\end{tabular}
\end{center}
$\TYPEZERO$ controller gives the best guarantee but it may be unrealizable. $\TYPEONE$ controller
provides a firm but conditional guarantee. The $\TYPETWO$ controller tries to achieve $C$
in $H$-optimal fashion irrespective of any assumption, whereas $\TYPETHREE$ Controller provides
firm conditional guarantee and it also tries to satisfy $C$ in $H$-optimal fashion even
when the assumption does not hold. 
}

\vspace*{-0.5 cm}
\section{Specification and Synthesis of robust controllers}
\label{section:robustness}
This section introduces the notion of robustness  and controller synthesis from Robust Specification. The reader is
urged to re-look at the explanation and motivation for robust
specification given in the Introduction section \ref{section:intro}.
We define the 
\emph{\textbf{Robust specification}}
as a triple $(D_A, D_C, \Rb(A))$, where the \qddc formulas $D_A$ and $D_C$ specify regular properties over input-output alphabet $(I,O)$ representing the \emph{assumption} and the \emph{commitment}. Moreover, $\mathbf{\Rb(A)}$ is a formula  over an  indicator variable $A$. Formula $Rb(A)$ is called a \emph{\textbf{Robustness criterion}}. It specifies a generic method for relaxing any arbitrary assumption formula $D_A$.

Given a robustness criterion $Rb(A)$ and concrete assumption formula 
$D_A$, we make use of the cascade composition (see Definition \ref{def:indDef}) to get the {\em relaxed assumption} under which the commitment must hold. The desired relaxed assumption is given by
\[
Rb(A) \ll Ind(D_A,A)
\]
which results in a \qddc\/ formula over the input-output variables 
$(I,O\cup \{A\})$. See Example \ref{exam:cascade}.

 The \emph{Robust specification}  $(D_A, D_C, \Rb(A))$ gives us the \emph{\DCSYNTH specification} below, which is denoted by 
 $\RbSpec(D_A, D_C, \Rb(A))$.
\begin{equation}
\label{eq:rbspec}
(I, (O \cup \{A\}), ~((\Rb(A)\ll Ind(D_A,A)) ~\Rightarrow~ D_C ), ~ D_C ) 
\end{equation}
The hard requirement states that the commitment $D_C$ {\bf must} hold whenever the relaxed assumption $Rb(A) \ll Ind(D_A,A)$ holds. Moreover, it specifies $D_C$ as the soft-requirement. The controller must be
{\bf optimized} so that the soft requirement holds for as many inputs as possible, irrespective of the assumption.

\oomit{
The aim is to synthesize controller which meets commitment ``as long as'' relaxed assumption $\mathit{\Rb(A)}$ is satisfied.
Bloem {\em et~al} \cite{Bloem} have shown that this can be interpreted in several ways giving different notions of robustness. In this section,
we provide a framework for defining (using logic \qddc) various notions of \emph{robustness criteria}. Moreover, we give a uniform method of 
synthesizing a robust controller using the soft-requirement guided synthesis framework of \DCSYNTH specification \cite{WPM19}.

{ \color{red}
For convenience, we assume that we have indicators $A$ and $C$ witnessing the assumption and commitment formulas $D_A$ and $D_C$ respectively.
A robustness criterion specifies under what condition on assumption the commitment must hold.

There can be several \emph{robustness criteria}. For example, the $BeCorrect(A,C)$ criterion requires
that for any point in execution,  if $A$ is continuously true at all points in past then $C$ must hold
at the current point. Traditional controller synthesis follows this criterion.
A more relaxed criterion is $BeCurrentlyCorrect(A,C)$ which states that if assumption is true at
current point in execution, then commitment must be true at the current point. (Both these, and several other
robustness criteria will be formally defined in this section). 
}
}
\oomit{
\begin{definition}[Robustness Criterion and Robust Specification]
Let $A$ and $C$ be indicators for assumption formula $D_A$ and the commitment formula $D_C$, respectively, over the input and output variables $I \cup O$ with indicator definitions $Inddef(D_A,A)$ and $Inddef(D_C,C)$. 

A \emph{Robustness criterion} is a QDDC formula $\Rb(A)$ over the  indicator proposition $A$, specifying the relaxed assumption under which commitment must hold. 
Every robustness criterion $\Rb(A)$ gives rise to its corresponding \emph{Robust Specification}. Robust specification is a \DCSYNTH specifications $(I,(O \cup A \cup C), (\Rb(A) \Rightarrow EP(C)), ~\langle C \rangle)$ together with $Inddef(D_A,A)$ and $Inddef(D_C,C)$ as indicator definitions. 
This \DCSYNTH specification is denoted by \RbSpec$(D_A, D_C)$.

\end{definition}

 We denote the \MPNC for \emph{Robust specification} $(D_A, D_C, \Rb(A))$  by \MPNC$(D_A, D_C, \Rb(A))$ and \GODSC by \GODSC$(D_A,D_C, \Rb(A),H)$ respectively.
}
\oomit{
{\color{red}
For example, the $BeCorrect(A)$ denotes the robustness criterion $[[A]]$ (i.e. $A$ holds invariantly in past) whereas the $BeCurrentlyCorrect(A)$ denotes the robustness criterion $EP(A)$ (currently $A$).
Intuitively, a controller with specification $BeCurrentlyCorrectSpec(A,C)$ will be more robust than a controller with specification $BeCorrectSpec(A,C)$.
}
}

\subsection{Synthesis from Robust Specification}
\vspace*{-0.2cm}
The robust controller synthesis method supported by tool \DCSYNTH is as follows. 
\begin{itemize}
\vspace*{-0.2cm}
\item The user provides the \emph{Robust specification} $(D_A,D_C,Rb(A))$. The user also provides a horizon value $H$.
\item The corresponding \DCSYNTH specification is automatically obtained as outlined in Equation \ref{eq:rbspec}.
The tool \DCSYNTH is used to obtain $\MPNC$ and $\GODSC$ supervisors
as described in Section \ref{section:dcsynth-spec}. These supervisors are denoted as  $\MPNC(D_A,D_C,Rb(A))$ and $\GODSC(D_A,D_C,Rb(A),H)$ respectively. 
Both these supervisors may be output non-deterministic.
The reader should recall that the $\MPNC$ supervisor only guarantees the hard-robustness, whereas  the $\GODSC$ supervisor further improves 
$\MPNC$ by optimizing the soft-robustness (while retaining hard-robustness guarantee).

\item  The user specifies an output order $ord$ as a prioritized list of output literals. This is used to get the deterministic controllers $Det_{ord}(\MPNC(D_A,D_C,Rb(A)))$ and $Det_{ord}(\GODSC(D_A,D_C,Rb(A),H))$, respectively. 
\end{itemize}

\oomit{
\begin{definition}[Universal Robustness Order]
\label{def:universalRobustnessOrder}
Given robustness criteria $Rb^1(A)$ and $Rb^2(A)$,
we say that $Rb^1(A) \leq_{robust} Rb^2(A)$ iff $ \models_{dc} Rb^1(A) \implies Rb^2(A)$.
\end{definition}
}

\oomit{
We propose to use the notion of \emph{must dominance}(definition \ref{def:mustDominance}) to compare between two given supervisors $Sup_1$ and $Sup_2$ w.r.t. a \qddc formula $D_C$, which is satisfied intermittently by both these supervisors. For example the \emph{must dominance} $Sup_1 \leq_{must}^{D_C} Sup_2$ checks that commitment guarantee provided by supervisor $Sup_2$ is at least as good as the commitment guarantee provided by $Sup_1$.  
}

\oomit{
In the rest of this section, we shall propose basic design principles for specification of \emph{Robustness criterion} and give several robustness criteria for controller synthesis. We also establish robustness order between various \emph{Robustness criterion} for qualitative analysis.
}

\oomit{
\subsection{Robust Controllers and Quality Measurement}

A controller is a function $S : (2^I)^+ \rightarrow (2^O)$ i.e. based on past input sequence
it produces current output. 
Equivalently, $S(I_1,\ldots,I_n)$ $=$ $(O_1,\ldots, O_n)$, the sequence of outputs produced by $S$ for given sequence of inputs. 
This also extends uniquely to infinite input sequences. A behaviour can also be rewritten as
$\rho = (I_1 \cup O_1, I_2 \cup O_2, \ldots)$. With some abuse of notation, we will not distinguish
between these equivalent notions of behaviours.

Consider a behaviour $\rho$ of a controller as infinite sequence of observations,
i.e., $\rho \in (2^{I \cup O})^\omega$. Normally, controller is designed to function under some
environmental and plant assumptions, and it is meant to provide guarantees.
Let $A$ be the conjunction of assumptions, and $C$ the conjunction of guarantees.
We assume that $A$ and $C$ are \emph{bounded past time} temporal properties specified in \qddc, 
i.e. $\rho, j \models A$
depends only on $\rho[i,j] \models A$ for $i \leq j$. Moreover, we assume that these assumptions and
commitments are intermittent, allowing for the possibility of recovering from
their failure at a position $j$. The pair $(A, C)$ is called the requirement.

\begin{example}
\textbf{Bus Arbiter}: Let's take an $n$ cell bus arbiter,
which a assumption represented by $A$ that, At most two requests are true in current cycle.
The set of commitments represented by $C$ are Mutual exclusion, nospurious and response time of 2 cycles.
\end{example}

In this example, the assumption is a propositional formula which depends
only on current cycle input. Another variant assumption could be that above
propositional assumption holds continuously for last 3 cycles. This is also a past
time temporal property which is intermittent.

The aim is to design a controller which is \emph{guaranteed} to meet the commitments whenever assumptions are \emph{true}. There are different ways of interpreting the above broad goal. In a strict sense, this can be interpreted
as $G(pref(A) ) C)$. 
We explain the notation below: $\rho, i \models pref(A)$ iff
$\forall j \leq i: \rho, j \models A$. 
Moreover $\rho \models G\phi$ iff $\forall i \in dom(\rho): \rho, i \models \phi$. Thus
formula $G(pref(A) ) C)$ states for any position (also called cycle) that if assumption has been invariantly true in the past, then current position will satisfy
$C$. Note that if assumption $A$ becomes $false$ at some cycle $i$ then commitment is
not required to hold from then on-wards. 
This is a very weak form of guarantee
called $BeCorrect(A, C)$.

Bloem et al argue that controllers should deliver better guarantees than just
BeCorrect. They should try to meet the commitments \emph{as much as possible}
even when assumptions are $false$, and it should be possible to recover from
intermittent failure of assumptions in bounded time. Such controllers are called
\emph{robust controllers}. 

There has been extensive study of algorithms and techniques for automatically synthesizing BeCorrect controllers from temporal logic based specifications
$(A, C)$. These studies mostly focus on the efficiency of the synthesis algorithm.
However, some recent papers do propose algorithms for the synthesis of robust
controllers. Unfortunately, these papers only focus on the efficiency of controller
synthesis and little attention is paid to rigorously analyzing the \emph{quality} of the
synthesized controller. Some times measures such as syntactic size of the controller are considered, but here we focus on the quality of robustness exhibited
by the observable behaviour of the controller.

The central question considered in this paper is the following: given two
controllers $S1$ and $S2$ for requirement $(A, C)$, how do we compare the quality
of two controllers from the point of view of robustness? We shall assume that
both controllers satisfy the $BeCorrect(A, C)$ property. This BeCorrect property
can be validated by model checking, or by synthesizing the controllers by a
correct-by-construction synthesis method. But the two controllers can differ
significantly in their robustness.

Let us assume that we are given an evaluator $Eval(\rho,[0 : n], A, C) \in \mathit{R}$
which for a given behaviour prefix provides a measure of how \emph{good} is
the behaviour prefix. 
For example, $rfreq(\rho, A, C) = (\#\{i \in [0, \#\rho - 1] ~\bar{|}~ \rho, i \models C)$ $/\#\rho$ measures fraction of cycles for which commitment holds during a behaviour.

\begin{itemize}
\item \textbf{Universal Dominance} 
\item \textbf{Worst Case Dominance}
\item \textbf{Expected Dominance}
\end{itemize}

In universal dominance, the controller $S_2$ is required to outperform $S_1$ on every
input sequence. In worst case dominance, the worst case performance of $S_1$ is
less than the worst case performance of $S_2$. In expected dominance, the expected
(mean) value of performance of $S_1$ on inputs of length $n$ is less than the expected
value of performance of $S_2$ on inputs of length $n$. If this holds for almost all $n$
then $S_2$ is better than $S_1$.

In the paper, we give techniques for establishing worst case dominance as
well as expected dominance between controllers. For worst case dominance, we
consider properties called \emph{worst case latency} of $D$ which are defined as infimum/supremum of span $j-i$ of intervals $[i, j]$ satisfying property $D$ in some execution of $S$. A symbolic search techniques for finding lengths of longest/smallest paths in automata can be used to determine the worst case latency. For determining expected dominance, we resort to statistical sampling techniques that
are widely used in performance measurement.

We consider the case studies of controllers for Synchronous Bus Arbiter, and
Mine Pump, to illustrate our approach. Several different controllers for each of
these problems are compared for their robustness using above two techniques.
}

\subsection{Designing and Comparing Robustness Criteria}
A key element of our robust specification $(D_A,D_C,Rb(A))$ is the
robustness criterion $Rb(A)$. It provides a generic method of relaxing
the assumption $D_A$. In this section we study how to compare two
robustness criteria. We also formulate some methods for systematic design of robustness criteria. 

Given robustness criteria $Rb_1(A)$ and $Rb_2(A)$, we say that
$Rb_1(A)$ implies $Rb_2$ provided 
$\models_{dc} Rb_1(A) \Rightarrow Rb_2(A)$. This gives us the {\bf implication ordering} on robustness criteria. The following theorem shows that implication ordering improves the worst case guarantee of $D_C$ holding but it makes supervisors less realizable.
\begin{theorem}
\label{lem:robustnessMPNCLemma}
Let  $\models_{dc} \Rb^1(A) \Rightarrow \Rb^2(A)$.  Then for all \qddc\/ formulas
$D_A,D_C$, we have \\
(a) $\MPNC(D_A,D_C,\Rb^1(A)) ~~\leq_{must}^{D_C}~~ 
\MPNC(D_A,D_C,\Rb^2(A))$. \\
(b) If the specification $(D_A,D_C,\Rb^2(A))$ is realizable then the specification \\ $(D_A,D_C,\Rb^1(A))$ is also realizable.
\end{theorem}
\begin{proof}
 For $1 \leq i \leq 2$, we have
 $\MPNC(D_A,D_C,Rb^i(A))$ equals $\MPNC(\phi_i)$ where 
 $\phi_i =  ((Rb^i(A) \land pref(EP(A) \Leftrightarrow D_A)) \Rightarrow D_C) $.
 Now, as $\models_{dc} Rb^1(A) \Rightarrow Rb^2(A)$, we have
 $\models_{dc}  (Rb^2(A) \Rightarrow D_C)  ~~\Rightarrow~~  (Rb^1(A) \Rightarrow D_C)$
 and hence  $\models_{dc}  ((Rb^2(A) \land D) \Rightarrow D_C)  ~~\Rightarrow~~  ((Rb^1(A) \land D) \Rightarrow D_C)$ for any formula $D$.
 Thus $\models_{dc} \phi_2 \Rightarrow \phi_1$. Then, by Proposition
 \ref{prop:mpncSupervisor}, we have $\MPNC(\phi_1) \leq_{det} \MPNC(\phi_2)$. From this, by applying
 proposition \ref{prop:detmust}, we get the desired result 
  $\MPNC(\phi_1) \leq^{D_C}_{must} \MPNC(\phi_2)$. We also get (b). \qed
 \end{proof}

\paragraph{\textbf{Error-Types: Defining Intervals with Assumption Errors.}} Let proposition $A$ designate points in behaviour where assumption is satisfied. In terms of this, we define Erroneous intervals by giving Error-type formulas. For examples, intervals with at least three violations of $A$ can be given by (\verb#scount !A >=3#).
\begin{itemize}
 \item Intervals where $A$ does not hold at the end-point.\\
     \verb# LocalErr(A) = (true^<!A>)# 
 \item Intervals having more that $k$ assumption violations (\verb#!A#).\\
       \verb#CountErr(A,k) = (scount !A > k)#.
 \item Burst error interval is an interval where the assumption proposition \verb#A# remains false invariantly. So, the intervals where there is burst error of length more than $k$ is given by the formula \verb#BurstErr(A,k)# = \verb#([[!A]] && slen >= k)#. 
The intervals containing a subinterval with \verb#BurstError(A,k)# are given by \\
       \verb#HasBurstErr(A,k) = (<>(BurstErr(A,k)))#.
 \item We define an interval to be a recovery interval, if $A$ is invariantly $true$ in that interval. Its length is
 called \emph{recovery period}. Intervals where all recovery periods are of length less than $b$ is denoted by\\
  \verb#HasNoRecovery(A,b) =# \verb#([]([[A]] => slen < b-1))#\\
  Given an error-type formula $Err$ (one of the above items), 
  we define \emph{recovery error} as the intervals having underlying error $Err$, but without a recovery interval of length $b$. \\
       \verb#RecoveryErr(b,Err) = (Err && HasNoRecovery(A,b))#

\oomit{
 \item Interval where there number of cycles between two consecutive errors two consecutive errors is less than $k$. \\
  \verb#SepErr(A,k) (((<!(-A)>)^[[A]]^(<+(!A)>)) && (slen < k-1))#
 \item Interval which have subintervals with SepErr(A,k). \\
 \verb#HasSepErr(A,k) (<>(SepErr(A,k)))#
 }
\end{itemize}
Each formula in the above list is called an {\em Error-type} formula.

\begin{proposition}
\label{prop:errorprop1}
For all Error-type formulas $Err, Err_1$ and $Err_2$ we have \\
 (a) $\mathit{\models_{dc} BurstErr(A,k)} \Rightarrow HasBurstErr(A,k)$\\
 (b) $\mathit{\models_{dc} HasBurstErr(A,k) \Rightarrow CountErr(A,k)}$\\ 
 (c) If $j < k$ then  $\mathit{\models_{dc} CountErr(A,k) \Rightarrow CountErr(A,j)}$ and \\
 \hspace*{1cm} $\mathit{\models_{dc} HasBurstErr(A,k) \Rightarrow HasBurstErr(A,j)}$  \\
 (d) $\mathit{\models_{dc} RecoveryErr(b,Err) \Rightarrow Err}$\\
 (e) If $\mathit{\models_{dc} Err_1 \Rightarrow Err_2}$ then \\
 \hspace*{1cm}$\mathit{\models_{dc} RecoveryErr(b,Err_1) \Rightarrow RecoveryErr(b,Err_2)}$
\end{proposition}
\begin{proof} 
\vspace*{-0.3cm}
These implications hold straightforwardly using the \qddc semantics. For example,
$\mathit{RecoveryErr(b,Err)}$ = $\mathit{HasNoRecovery(A,b)}$ $\&\&$ $Err$  (by definition). This logically implies $Err$ giving us (d). We omit the remaining proofs.\qed
\end{proof}

\paragraph{\textbf{Error-Scope: Restricting the Occurrence of Errorneous Intervals}} 
Let $Err$ be any of the Error-type formulas giving erroneous intervals. 
We now specify the restrictions on occurrence of these errors  by
suitable \qddc formulas. These are called {\em Error-scope} formulas. They are
parameterized by $Err$.
\begin{itemize}
 \item The error never occurs in past for any interval. \\
           \verb#NeverInPast(Err) = !<>Err#.
 \item Error does not occur in any suffix interval. \\
           \verb#NeverInSuffix(Err) = !(true^Err)#
 \item Error does not occur in past in any interval of 
 $b$ or less cycles. \\
           \verb#NeverInPastLen(b,Err) = !<>(slen <= b-1 && Err)#
 \item Error does not occur in an interval spanning last $b$ or less cycles. \\ 
         \verb#NeverInSuffixLen(b,Err) = !(true^((slen<=b-1) && Err))#
\end{itemize}

\begin{proposition} 
\label{prop:errorprop2}
For any $Error\verb#-#type$ formula $Err, Err_1$ and $Err_2$ we have \\
(a) $\mathit{\models_{dc} NeverInPast(Err)~\Rightarrow~NeverInSuffix(Err)}$ \\
(b) $\mathit{\models_{dc} NeverInPastLen(b,Err)~\Rightarrow ~ NeverInSuffixLen(b,Err)}$ \\
(c) $\mathit{\models_{dc} NeverInPast(Err) ~ \Rightarrow ~NeverInPastLen(b,Err)}$ \\
(d) $\mathit{\models_{dc} NeverInSuffix(Err)~\Rightarrow ~NeverInSuffixLen(b,Err)}$\\
(e) For any scope formula $SCP(Err)$ defined above, we have \\
\hspace*{1cm} if $\mathit{\models_{dc} Err_1 \Rightarrow Err_2}$, then $\mathit{\models_{dc} SCP(Err_2) \Rightarrow SCP(Err_1)}$ 
\end{proposition}
\begin{proof}
\vspace*{-0.3cm}
The proofs of these implications are immediate from definitions using \qddc semantics. For example,
($true\verb#^#Err$) implies  ($true\verb#^#Err\verb#^#true$) which equals $\langle\rangle Err$ by \qddc. Hence, $\mathit{NeverInPast(Err)}$ which equals
$!\langle\rangle Err$ implies $\mathit{!(true}\verb#^#\mathit{Err})$ which equals $\mathit{NeverInSuffix(Err)}$. This gives us (a). 

As a second example, $\mathit{NeverInPast(Err)}$ $=\mathit{(!\langle\rangle (Err))}$ (by definition),  which implies  $\mathit{!\langle\rangle((slen \leq (b-1)) ~\&\&~ Err)}$. This, by definition, equals the formula, $\mathit{NeverInPastLen(b,Err)}$. Hence we get (c). We omit the remaining proofs.   \qed
\end{proof}

\begin{table}
\caption{Robustness criteria $Rb(A)$ defined using Error-types and Error-Scope formulas. There may use additional integer parameters $K,B$.}
\label{tab:robustformula}

\begin{tabular}{|c|c|}
\hline
Robustness   &  Definition     
\\
 Criterion  &   of $Rb(A)$ 
\\ \hline\hline

AssumeFalse(A) &
\verb#(false)#
\\ \hline

BeCorrect(A)  			& \verb#NeverInPast(LocalErr(A))#
\\ \hline

BeCurrentlyCorrect(A)		&  
\verb#NeverInSuffix(LocalErr(A))#
\\ \hline

ResCnt(A,K,B)	& 
\verb#(NeverInPast(RecoveryErr(B, CountErr(A,K))))#

\\ \hline 

ResCntInt(A,K,B) &
\verb#(NeverInSuffix(RecoveryErr(B, CountErr(A,K))))#
\\ \hline

ResBurst(A,K,B)	& 
\verb#(NeverInPast(RecoveryErr(B, HasBurstErr(A,K))))#

\\ \hline 
ResBurstInt(A,K,B) &
\verb#(NeverInSuffix(RecoveryErr(B, HasBurstErr(A,K))))#
\\ \hline

LenCnt(A,K,B) &
\verb#(NeverInPastLen(B,CountErr(A,K)))#
\\ \hline

LenCntInt(A,K,B)			&  
\verb#(NeverInSuffixLen(B,CountErr(A,K)))#
\\ \hline 

LenBurst(A,K,B) &
\verb#(NeverInPastLen(B,HasBurstErr(A,K)))#
\\ \hline
LenBurstInt(A,K,B)			&  
\verb#(NeverInSuffixLen(B,HasBurstErr(A,K)))#
\\ \hline 

AssumeTrue(A) &
\verb#(true)#
\\ \hline

\end{tabular}
\end{table}

We combine various {\em Error-type} formulas with {\em Error-scope} formulas (defined above) to obtain a wide variety of {\em Robustness Criteria}. These are tabulated in Table  \ref{tab:robustformula}. These provide users with diverse ways of achieving robust specification. We
give a brief intuitive description of some of the defined criteria.
\begin{itemize}
 \item Criterion \verb#BeCorrect(A)# holds at a point if \verb#!A# never occurs in its past. By contrast \verb#BeCurrentlyCorrect(A)# holds at a point (irrespective of the past) if \verb#A# holds at the point. Thus, the later is implied by the former.
 \item Criterion \verb#LenCntInt(A,K,B)# with integer parameters $K,B$ holds at a point $i$ provided in last $B$ cycles from $i$, violation \verb#!A# occurs at most $k$ times. The past beyond last $B$ cycles does not affect its truth.
 \item A resilient controller recovers from past errors provided \verb#A# holds continuously for $B$ cycles (such an interval is called a recovery interval). The criterion
 \verb#ResCntInt(A,K,B)# holds at a point provided after the last recovery interval the count of violation of $A$ is at most $k$.
 \item A burst error of length $K$ is said to occurs if \verb#!A# occurs continuously for $K$ points.
 Criterion \verb#LenBurstInt(A,K,B)# holds at point $i$ provided in last $B$ cycles before $i$ there is no burst error of length $K$.
\end{itemize}

The robustness criteria can be ordered by {\bf implication ordering}
which also improves the hard robustness of the synthesized controller (See Theorem \ref{lem:robustnessMPNCLemma}).
Figure \ref{fig:robustnessHierarchy} gives the implication ordering between the robustness criteria of Table \ref{tab:robustformula}.  Hence, in specification $(D_A,D_C,\Rb(A))$, the user must use the weakest criterion (under the implication ordering)  from the Figure \ref{fig:robustnessHierarchy} which makes the specification realizable. 

\begin{figure}
\centering
\begin{tikzpicture}[scale=0.9]
	\node[ellipse] (12) at (3, 2) {AssumeTrue};
	\node[ellipse] (11) at (3,1) {LenBurstInt};
    \node[ellipse] (10) at (0,0) {LenCntInt};
    \node[ellipse] (9) at (5,0) {ResBurstInt};
    \node[ellipse] (8) at (8,1) {BeCurrentlyCorrect};
    \node[ellipse] (7) at (0, -2)  {ResCntInt};
    \node[ellipse] (6) at (5, -1) {ResBurst$\leftrightarrow$LenBurst};
    \node[ellipse] (5) at (5, -2) {LenCnt};
    \node[ellipse] (4) at (5, -2.9) {ResCnt};
    \node[ellipse] (2) at (7, -3.6) {BeCorrect};
    \node[ellipse] (1) at (7, -4.5) {AssumeFalse};
	\draw [->] (1) edge (2);
    \draw [->] (2) edge (4);
    \draw [->] (4) edge (5);
    \draw [->] (4) edge (7);
    \draw [->] (5) edge (6);
    \draw [->] (5) edge (10);
    \draw [->] (7) edge (9);
    \draw [->] (7) edge (10);
    \draw [->] (2) edge (8);
    \draw [->] (6) edge (9);
    \draw [->] (9) edge (11);
    \draw [->] (10) edge (11);
    \draw [->] (11) edge (12);
    \draw [->] (8) edge (12);
    
\end{tikzpicture}
\caption{Implication order on the robustness criteria of Table \ref{tab:robustformula}. Here $A \rightarrow B$ denotes the validity 
$\models_{dc} A \Rightarrow B$.}
\label{fig:robustnessHierarchy}
\end{figure}
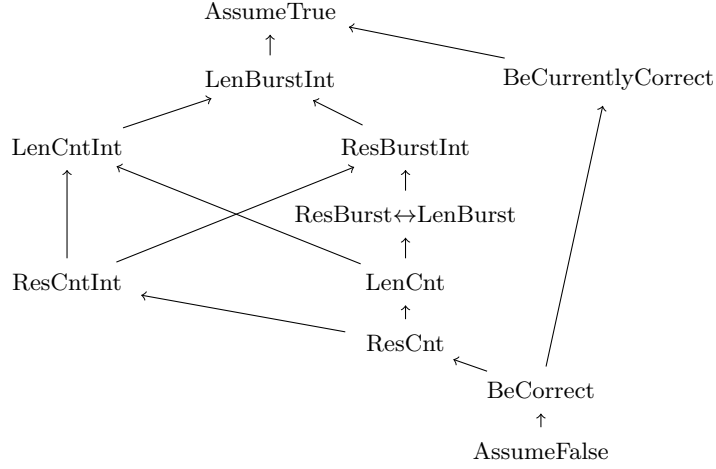

\vspace*{-0.3cm}
\begin{theorem}
 All the implications given in Figure \ref{fig:robustnessHierarchy} are valid.
\end{theorem}
\begin{proof}
\vspace*{-0.3cm}
The proofs of these implications follow easily from the definitions of the robustness criteria by using propositions \ref{prop:errorprop1} and \ref{prop:errorprop2}. As an example, we prove that $\models_{dc} BeCorrect(A) \Rightarrow BeCurrentlyCorrect(A)$. 

Formula $\mathit{BeCorrect(A)}$ equals (by definition) $\mathit{NeverInPast(LocalErr(A))}$, which by Proposition \ref{prop:errorprop2}(a) implies $\mathit{NeverInSuffix(LocalErr(A))}$ which by definition equals
$\mathit{BeCurrentlyCorrect(A)}$, thus proving the claim.

As a second example, we prove that 
$\mathit{\models_{dc} ResCntInt(A)}$ $\Rightarrow$ $\mathit{ResBurstInt(A)}$.\\
We have
$\mathit{\models_{dc} HasBurstErr(A,k)}$ $\Rightarrow$ $\mathit{CountErr(A,k)}$ by Proposition \ref{prop:errorprop1}(b).
 Thus, $\mathit{\models_{dc} RecoveryErr(b,HasBurstErr(A,k))}$ $\Rightarrow$ 
$\mathit{RecoveryErr(b,CountErr(A,k))}$ by Prop \ref{prop:errorprop1}(e).
Then, using Prop. \ref{prop:errorprop2}(e) and instantiating $SCP$ by
$\mathit{NeverInSuffix}$, we get 
$\mathit{\models_{dc} NeverInSuffix(RecoveryErr(b,CountErr(A,k))})$ 
$\Rightarrow$ \\
$\mathit{NeverInSuffix(RecoveryErr(b,HasBurstError(A,k)))}$. 
This proves the result. We omit the remaining proofs. 
\qed
\end{proof}

\oomit{
We now briefly describe few of these robustness criteria given in table \ref{tab:robustformula} and the corresponding specification.
\begin{itemize}
\item {\bf BeCorrect(A):} states that assumptions are never violated in past. The corresponding Robust specification ($D_A,D_C,BeCorrect(A)$) stats that if assumption has held invariantly so far then commitment must hold now. Most of the synthesis methods use this kind of specification, which is not a very robust specification, as the formula $\textit{BeCorrect(A)}$ invariantly evaluates to false after any assumption violation.

 \item  {\bf BeCurrentlyCorrect(A)} This robustness criterion is an intermittent version \textit{BeCorrect(A)}. It checks whether the assumption is violated in the current cycle. Hence, in the controller for corresponding robust specification the commitment will hold whenever assumption holds.
 It may be noted that this is a much robust criterion than $BeCorrect(A)$ as the satisfaction of commitment $D_C$ does not depend upon the past violations of assumption.
 
 \item  {\bf ResBurstInt(A,K,B)} This criterion stats that the maximum length of burst error (Assumption violation) is at-most $K$, between the last two recovery intervals of period $B$. 
 
 \item  {\bf LenCntInt(A,K,B)} This criteria states that, the number of assumption violations are less than $K$, in an interval of spanning last $B$ or less cycles.
\end{itemize}
We omit the detailed explanation of other robustness criteria given in Table \ref{tab:robustformula}, as they are also straight forward to understand from their formulation.  
}

\oomit{
Now we consider various robustness criterion \textbf{\Rb(A)} specified by the combination of Error types and scope of its occurrence 
as shown in the Table~\ref{tab:robustformula}. Each of these robustness criterion tries to tolerate one or more error type (pattern of assumption violation). The goal of robust synthesis is to meet the commitments with the error tolerant assumption.
 
\begin{itemize}
\item {\color{blue} {\bf $BeCorrect(A)$:} states that assumptions are never violated in past. The corresponding Robust specification ($D_A,D_C,BeCorrect(A)$) stats that if assumption has held invariantly so far then commitment must hold now. It may be noted that this Robustness criterion is non-intermittent because as soon as $BeCorrect(A)$ gets violated it will never be satisfied again. Most of the synthesis methods use this kind of specification, which is not a very robust specification. 

 \item  {\bf BeCurrentlyCorrect(A)} This robustness criterion is an intermittent version \textit{BeCorrect(A)}. It checks whether the assumptions are violated in this cycle. The corresponding robust specification $(D_A,D_C,BeCurrentlyCorrect(A))$  states that commitment must hold whenever assumption holds.
 It may be noted that this is a much robust specification that $D_A,D_C,BeCorrect(A)$ as the satisfaction of commitment $D_C$ does not depend upon the past violations of assumption.

\oomit{ 
 \item {\bf Degraded-Performance} Let \verb#Ad# $\subseteq$ \verb#A# and  \verb#Cd# $\subseteq$ \verb#C#, where
 \verb#Ad,Cd# denote reduced set of assumptions and commitments which specify degraded behaviour of system when fewer assumptions
 hold.
 
\item \textbf{Resilient synthesis} requires synthesis of controller which works under weaker assumptions than Be-Correct notion in order to tolerate errors 
in assumptions. For example, BeCurrentlyCorrect requires commitment to hold now if assumption holds now irrespective of whether it has held in past.
Several other notions of resilience are given below. A notion of $k,b$-resilience was proposed by Ehler and Topcu \cite{ET14}, 
and further adapted by Bloem as $k$-robustness \cite{BKKW15}.

\begin{itemize}

 \item {\bf $K,B-SepBurst$} If in any past period, the \emph{maximum} period of continuous assumption violation is of length $B$ and \emph{minimum} period for which assumption holds (once it is satisfied) is of length $K$, then the commitment should hold in present cycle.
}

\item {\bf ResCnt(A,K,B)} Let a sub-interval where assumption is continuously true for $B$ or more cycles be called a \emph{recovery interval} of length $B$. 
 The robustness criterion \textit{KBResCnt(A)} states that between \emph{any two} recovery interval of length $B$, the maximum number of assumption violations is at most $K$. 
\item {\bf ResCntInt(A,K,B)} \textit{ResCntInt(A,K,B)} is an intermittent version of \textit{ResCnt(A,K,B)}, in which we state that, between the \emph{last two} recovery periods,  maximum number of assumption violations are at most $K$. This robustness criteria is intermittent as it depends only on the assumption violation in last two recovery periods, instead of whole past.  
 
 \oomit{
 \item {\bf $k$-Atmost} If in past assumptions have been violated at most $k$ times so far then commitment must be met. 
 \item {\bf $K,B$-Variant-1} If in any past period of length $B$ cycles the assumption has been violated at most $K$ times then the commitment must hold in present cycle. The problem is this specification requires atleast $B$ cycles to determine whether guarantees are to be met. Therefore, it does not have
 any guarantee in first $B$ cycles.
 \item {\bf $K,B$-Variant-3} Another interesting variant is that, if in any period if the number of assumption violations are less that $k$ or the assumptions are being continuously satisfied for atleast $B$-cycles then the commitment should be met in this cycle.
 }

 \item  {\bf ResBurstInt(A,K,B)} This robustness criterion is a variant of \textit{ResCntInt(A,K,B)}, it stats that the length of burst error (Assumption violation) is at-most $K$, instead of number of errors between two recovery periods. 
 
 \item  {\bf LenCntInt(A,K,B)} This robustness criteria states that, the number of assumption violations are less than $K$, in an interval of spanning last $B$ or less cycles.
 
 \item  {\bf LenBurstInt(A,K,B)} This robustness criteria states that, the maximum length of burst error (assumption violations) is less than $K$ cycles, in an interval of spanning last $B$ or less cycles.

The non-intermittent versions {\bf ResBurst(A,K,B)}, {\bf LenCnt(A,K,B)} and {\bf LenBurst(A,K,B)} can be easily understood from the description of corresponding intermittent versions.
 }
\end{itemize}
}
\oomit{
\begin{table}
\caption{Robust synthesis criteria and their specifications.}
\label{tab:robustformula}
In the table \verb#A# is an indicator for conjunction of assumptions
and \verb#C# denotes indicator for conjunction of commitments. Thus, \verb#!A# denotes violation of at least one assumption.
\begin{tabular}{|c|c|}
\hline
Robustness   &  Robustness   
\\
 Notion  &   Criteria  

\\ \hline\hline
Be-Correct(A,C)  			&  \verb#pref(A)#	
\\ \hline

$K,B$-Sep(A,C,K,B)			& 
\verb#([](([[!A]] => slen <= K-1) && #
\\& \verb#(((<!(-A)>)^[[A]]^(<+(!A)>)) => (slen >= B-1)))) # 
\\ \hline 

$K,B$-Resilient(A,C,K,B)			& 
\verb#(!(true ^ ((scount !A > K) && #
\\ & \verb#[]([[A]] => slen<B-1))^true)) # 
\\ \hline 

$K,B$-ResilientInt(A,C,K,B) &
\verb#(!(true ^ ((scount !A > K) && #
\\ & \verb#[]([[A]] => slen<B-1)))) #

\\ \hline

$K,B$-SepInt(A,C,K,B) &
\verb#(!(true ^ (  (<>( [[!A]] && slen > K-1)) && #
\\
 & \verb#([]([[A]] => slen<B-1)))))#

\\ \hline

$K,B$-Variant-2(A,C,K,B)			&  \verb#((slen<(B-1) && (scount !A <= K)) || #
\\& \verb#( true ^ (slen=(B-1) && (scount !A <= K))))# 
\\ \hline 

Be-Currently-Correct(A,C)		& \verb#EP(A)#  
\\ \hline

\oomit{
$k,b$-Variant-3(A,C,k,b)			&   \verb#([]((slen=b-1 && scount !A <= k) => EP(C))) &&# 
\\  
& \verb# pref((slen<b && scount !A < k) => EP(C))#  
\\ \hline  

$k,b$-Sep(A,C,k,b)			& \verb#([]( [[!A]] => slen <= k-1) &&#  
\\
& \verb#[]((!<-A>)^[[A]]^(!<+A>) => slen >= b-1) => EP(C))# 
\\ \hline 

Degraded-Performance(A,C,Ad,Cd)		& \verb#((EP(A) => EP(C)) && (EP(Ad) => EP(Cd))# 
\\ \hline

Greedy(C)				&		\verb#true#	
\\ \hline
}

\end{tabular}

\end{table}

Note that all the above Robustness notions are only based on the \emph{Hard Requirements} and hence they 
have an associated \emph{Guarantee}. Now we define notions based on \emph{Soft Requirements}. 
\begin{itemize}

\item {\bf Never-Give-Up} 
In addition to \emph{Hard Requirements}  to be met, in Never-Give-Up all the commitments are asserted as \emph{Soft Requirements}(given by \verb#EP(C)#, where \verb#C# denotes the indicator for conjunction of commitments). This directs the
 synthesizer attempt to make the commitments true even when assumptions do not hold, and for as many inputs as possible.
 
 Note that criterion such as \emph{Never-Give-Up} can be combined with any hard requirement based robustness notion. We denote the \emph{Never-Give-Up} version of robustness for a given robustness notion \verb#RbSpec# by \verb#RbSpec-Soft# e.g. the \emph{Never-Give-Up} version of \verb#Be-Correct# is denoted by \verb#Be-Correct-Soft#. 
We will show the effect of \emph{soft requirements} on 
robust controller synthesis using \verb#Never-Give-Up# version for each hard requirement based robustness notion defined in table \ref{tab:robustformula}. Moreover, designer may selectively apply these criteria to specific assumptions and commitments. \DCSYNTH permits full flexibility in making such choices.

 \item {\bf Greedy} Here commitments are given as soft requirements ignoring the assumptions. The synthesis algorithm tries to make as many commitments
 true as possible irrespective of whether the assumption are met or not.
 This notion does not give any formal Guarantee as 
 there are no \emph{Hard Requirements} specified explicitly, but it will be the most robust controller.

\end{itemize}
}

\oomit{
\begin{lemma}

  We have \verb#Be-Correct(A)# $\implies$ \verb#K,B-SepBurst(A)# $\implies$ \verb#K,B-ResCnt(A)# $\implies$ \verb#K,B-ResCntInt(A)# $\implies$ \verb#K,B-ResBurstInt(A)#, Hence the Robustness order is \verb#Be-Correct(A)# $\leq_{robust}$ \verb#K,B-SepBurst(A)# $\leq_{robust}$ \verb#K,B-ResCnt(A)# $\leq_{robust}$ \verb#K,B-ResCntInt(A)# $\leq_{robust}$ \verb#K,B-ResBurstInt(A)#
\end{lemma}

\begin{lemma}
  We have \verb#Be-Correct(A)# $\implies$ \verb#K,B-SepBurst(A)# $\implies$ \verb#K,B-ResCnt(A)#  $\implies$
\verb#K,B-ResCntInt(A)#  $\implies$ \verb#K,B-VarInt(A)#, Hence the Robustness order is \verb#Be-Correct(A)# $\leq_{robust}$ \verb#K,B-SepBurst(A)# $\leq_{robust}$ \verb#K,B-ResCnt(A)# 
$\leq_{robust}$ \verb#K,B-ResCntInt(A)#
$\leq_{robust}$ \verb#K,B-VarInt(A))#
\end{lemma}

\begin{lemma}
 We have \verb#Be-Correct(A)# $\implies$ \verb#Be-CurrentlyCorrect(A)#, hence \verb#Be-Correct(A)# $\leq_{robust}$ \verb#Be-CurrentlyCorrect(A)#. However, \verb#Be-CurrentlyCorrect(A)# is incomparable with other Robustness criteria and hence there does not exist an order between \verb#Be-CurrentlyCorrect(A)# and other robustness notions.

\end{lemma}
}
\oomit{
Now we formally establish the ordering(implication) relation between various robustness criterion for a given Assumption indicator $A$ and constants $k$ and $b$, with following proposition.
}
\oomit{
{\color{red}
\begin{lemma}
{\color{black}
The robustness order for various robustness criterion given in table \ref{tab:robustformula} follows the order shown in figure \ref{fig:robustnessHierarchy}. 
A Robustness criterion $\Rb^2(A)$, which is more robust that $\Rb^1(A)$ is shown in the figure by and arrow from $\Rb^1(A)$ to $\Rb^2(A)$. Thus, the arrows in the figure depicts the implication relation between the Robustness criterion.
}
\oomit{
\begin{enumerate}

\item \verb#BeCorrect(A)# $\leq_{robust}$ \verb#K,B-SepBurst(A)# $\leq_{robust}$ \verb#K,B-ResCnt(A)# $\leq_{robust}$ \verb#K,B-ResCntInt(A)# $\leq_{robust}$ \verb#K,B-ResBurstInt(A)#

\item \verb#BeCorrect(A)# $\leq_{robust}$ \verb#K,B-SepBurst(A)# $\leq_{robust}$ \verb#K,B-ResCnt(A)# 
$\leq_{robust}$ \verb#K,B-ResCntInt(A)#
$\leq_{robust}$ \verb#K,B-LenCntInt(A))#

\item \verb#BeCorrect(A)# $\leq_{robust}$ \verb#K,B-SepBurst(A)# $\leq_{robust}$ \verb#K,B-ResCnt(A)# $\leq_{robust}$ \verb#K,B-LenCnt(A)# $\leq_{robust}$
\verb#K,B-ResBurst(A)# $\leq_{robust}$ \verb#K,B-ResBurstInt(A)#

\item \verb#BeCorrect(A)# $\leq_{robust}$ \verb#K,B-SepBurst(A)# $\leq_{robust}$ \verb#K,B-ResCnt(A)# $\leq_{robust}$ \verb#K,B-LenCnt(A)# $\leq_{robust}$
 \verb#K,B-LenCntInt(A)#

\item \verb#BeCorrect(A)# $\leq_{robust}$ \verb#BeCurrentlyCorrect(A)#. However, \verb#BeCurrentlyCorrect(A)# is incomparable with other Robustness criteria 
\end{enumerate}
}
\begin{proof}
The proof will be carried out in two parts.
First we will establish the order between all types of $KB$ robustness criterion.
We formulate any $KB$ robustness criterion as a tuple (INTERVAL, TYPE, RECOVERY), where the INTERVAL talks about whether the error interval is a fixed length interval (denoted by $Len$ i.e. interval of length $B$ or less cycle) of recovery interval (denoted by $Res$ i.e. interval with recovery period of $B$ length), the TYPE talks about whether the error is a burst error (denoted by $Burst$ i.e. burst error of length $K$) or counting error (denoted by $Cnt$ i.e. number of error in the interval is $K$) and RECOVERY talks about whether the robustness criteria is intermittent (denoted by $Int$ i.e. criterion gets satisfied intermittently) or not. Based on this formulation we can encoded all the $KB$ robustness criterion e.g. $KBResBurstInt$ = $(Res, Burst, Int)$ for a given $K$ and $B$.
Now we following can be stated for all $X \in \{Res, Len\}$, $Y \in \{Cnt, Burst\}$ and $Z \in \{Int, NonInt\}$:\\
(a) $(Res, Y, Z) \Rightarrow (Len, Y, Z)$ \\
(b) $(X, Cnt, Z) \Rightarrow (X, Burst, Z)$ \\
(c) $(X, Y, NonInt) \Rightarrow (X, Y, Int)$ \\
\end{proof}
\end{lemma}
}
}
\oomit{
{\color{red} The arrows in the Figure \ref{fig:robustnessHierarchy} shows the implication relation between the relaxed assumptions $Rb(A)$. Therefore the bottom of the hierarchy starts with $AssumeFalse(A)$ which implies every other relaxed assumption. Hence, the controller synthesized with relaxed assumption $AssumeFalse(A)$ may never meet the commitment and hence the least robust. At the top of the hierarchy we have $AssumeTrue$ and using this relaxed assumption gives the most robust controller which has to meet the commitment irrespective of the assumption is being satisfied or not. However, the specification $RbSpec(AssumeTrue(A), C)$ is unrealizable in most of the cases.
We therefore compare the controllers based on $KBLenBurstInt$ and $BeCurrentlyCorrect$ relaxed assumptions, which are incomparable and generate different controllers for corresponding $RbSpec$. Both of these are the most robust controller in terms of must dominance.}
}

\vspace*{-0.6cm}
\section{Case Study: A Synchronous Bus Arbiter}
\label{section:motivation}


An $n$-cell synchronous bus arbiter has inputs $\{ req_i\}$  and outputs $\{ack_i\}$ 
where $1 \leq i \leq n$.  In any cycle, a subset of $\{req_i\}$ is true and the 
controller must set one of the corresponding $ack_i$ to true. 
The arbiter \textbf{commitment},  {\bf $ArbCommit(n,k)$}, is the conjunction of the following four properties.
\begin{small}

\begin{equation}
\label{eq:ArbProperties}
 \begin{array}{l}
  Mutex(n) \df true\textrm{\textasciicircum} \langle ~\land_{i \neq j} ~\neg ( ack_i \land ack_j) ~ \rangle \\
  NoLoss(n) \df true\textrm{\textasciicircum} \langle ~~ (\lor_i req_i) \Rightarrow (\lor_j ack_j) ~\rangle \\
  NoSpurious(n) \df true\textrm{\textasciicircum} \langle ~~\land_i ~(ack_i \Rightarrow req_i) ~~\rangle \\
Response(n,k) = (\land_{1 \leq i \leq n} ~(Resp(req_i,ack_i,k)) \quad \mbox{where} \\
 \quad Resp(req,ack,k) = 
true\textrm{\textasciicircum}(([[req]]\ \&\&\ (slen =(k-1)))~\Rightarrow \\
   \hspace*{2cm} true\textrm{\textasciicircum}(scount\ ack > 0 \ \&\&\ (slen =(k-1)))   
 \end{array}
\end{equation}

\end{small}
In \qddc\/, 
The formula 
$true\textrm{\textasciicircum} \langle P \rangle$ holds at a point $i$ in execution if the $P$ holds at that point. 
Thus, the formula $Mutex(n)$ gives mutual exclusion of acknowledgments; $NoLoss(n)$ states that
if there is at least one request then there must be  an acknowledgment; and 
$NoSpurious$(n) states that acknowledgment is only given to a requesting cell.
Formula $true\textrm{\textasciicircum}(([[req]]\ \&\&\ (slen =(k-1)))$ states that in the 
last $k$ cycles $req$ is invariantly true. Similarly,
the formula $true\textrm{\textasciicircum}(scount\ ack > 0 \ \&\&\ (slen =(k-1)))$  states that in 
last $k$ cycles the $ack$ has been $true$ at least once. Then, the formula $Resp(req,ack,k)$  states that if $req$ has been continuously true in last $k$
cycles, there must be at least one $ack$ within last $k$ cycles.
So, $Response(n,k)$ says that each cell requesting continuously for last $k$ cycles must get an acknowledgment within last $k$ cycles.

A controller can invariantly satisfy $ArbCommit(n,k)$ if $n \leq k$. Tool  \DCSYNTH gives us
a concrete controller for the instance  ($D^h=ArbCommit(6,6)$, $D^s=true$).
It is easy to see that there is no controller which can invariantly satisfy $ArbCommit(n,k)$ if $k < n$.
Consider the case when all $req_i$ are continuously true. Then, it is not possible to give response to
every cell in less than $n$ cycles due to mutual exclusion of $req_i$. 

To handle such desirable but unrealizable requirement we make an assumption. Let the proposition
$Atmost(n,i)$ be defined as $\forall S \subseteq  \{1 \ldots n\}, |S| \leq i. ~~  \land_{j \notin S} \neg req_j$.  It states that  at most $i$ out of total $n$ requests can be true simultaneously. Then, the {\bf arbiter assumption} is the formula
{\bf $ArbAssume(n,i)$} = $true\verb|^|\ang{Atmost(n,i)}$, which states that $Atmost(n,i)$  holds at the current cycle.

So the specification for the synchronous arbiter is an assumption-commitment pair $(ArbAssume(n,i), ArbCommit(n,k))$, which is denoted by {\bf $Arbiter(n,k,i)$}. 
Figure \ref{figure:arbiterspecification} in Appendix \ref{section:arbiterTextualSpecification}  gives the textual syntax for the robust specification
$(Arbiter(4,3,2))$ with $BeCurrentlyCorrect(A)$ criterion in tool \DCSYNTH.

\subsubsection{Expected Case Performance}
We can  measure the performance of a synthesized controller $Cnt$ in meeting a regular requirement $D$ by evaluating $E_{D}(Cnt)$, the expected value of $D$ holding in the behaviours of $Cnt$ in long run under random (i.e. iid) inputs \cite{Bel57}. The tool \DCSYNTH\/ allows $E_{D}(Cnt)$ to be computed as outlined below.

The product $Cnt \times A(D)$ of the controller with the deterministic recognizer $A(D)$ for formula $D$ gives us a deterministic recognizer automaton which is in a final state precisely on inputs where $Cnt$ satisfies the property $D$. We can translate this into a {\em Discrete Time Markov Chain}, denoted 
$M_{unif}(Cnt  \times A(D))$, by assigning uniform discrete probabilities to all the inputs from any state. In constructing this Markov chain, we have assumed that the inputs are $iid$, i.e. they occur independently of the past, and all inputs are equally likely. Standard techniques from Markov chain analysis allow us to compute the {\em Expected value} of being in accepting states in long runs
of $M_{unif}(Cnt  \times A(D))$, which also gives us $E_D(Cnt)$. A leading probabilistic model checking tool MRMC implements this computation \cite{KZHHJ11}. In \DCSYNTH, we provide a facility to compute $M_{unif}(Cnt  \times A(D))$ in a format accepted by the tool MRMC. Hence, using MRMC, we are able to compute $E_D(Cnt)$ \qed

\vspace{-0.3cm}
\subsection{Experimental Results}
\label{sec:casestudy}

\oomit{
Every case study involves a basic specification $(D_A,D_C)$ consisting of user specified assumption and commitment formulas. Given a robustness criterion $Rb(A)$ and default output ordering $ord$, we can synthesize
a determinized $\MPNC$ controller called $\detMPNC$ which enforces the \emph{hard robustness}. We can also synthesize a determinized $\GODSC$ controller which further refines $\MPNC$ by optimizing \emph{soft robustness} $D_C$. The reader may refer to Section 
\ref{sec:robustsynthmethod} for the synthesis method.
}

The synchronous bus arbiter case study specifies the assumption-commitment pair $(D_A,D_C)$ for an $Arbiter(n,k,i)$ consisting of $n$-cells with response time requirement of $k$ cycles, under the assumption that at most $i$ requests occur in each cycle. 
We denote by $\detMPNC(Arb)$ and $\detGODSC(Arb)$ the $\MPNC$ and the 
$\GODSC$ for $Arbiter(4,3,2)$ determinized under the default output ordering \verb#a1 > a2 > a3 > a4#. We use $H=50$ in synthesizing the $\GODSC$. Controllers were synthesized for various robustness criteria, and  their performance was measured as the corresponding expected values $E_{D_C}(\detMPNC(Arb))$ and $E_{D_C}(\detGODSC(Arb))$.
Table \ref{tab:arbmineExpectedValues} provides these values under the
columns titled $E(Arb-\MPNC)$ and $E(Arb-\GODSC)$, respectively.

Similarly, for a case study of a Minepump Controller specification
(omitted here for brevity but it can be found in Appendix \ref{section:minepumpcasestudy}), we synthesized the
determinized $\MPNC$ and $\GODSC$ controllers under the default output $PumpOn$ and various robustness criteria. We measured the performance of these controllers as the expected values $E_{D_C}(\detMPNC(MP))$ and $E_{D_C}(\detGODSC(MP))$. Table \ref{tab:arbmineExpectedValues} provides these values under the
columns titled $E(MP-\MPNC)$ and $E(MP-\GODSC)$, respectively.

In all the above cases, the tool \DCSYNTH performs efficiently by giving the required controllers for Arbiter within 1 seconds and for Minepump within 3 seconds. The detailed statistics for each of the example can be found in Appendix \ref{section:arbiterTextualSpecification} and \ref{section:minepumpcasestudy}. All these experiments were done on Ubuntu 18.04 system with Intel i5 64 bit, 2.5 GHz processor and 4 GB memory.

\begin{table}
\caption{Expected Value of Commitment $D_C$ holding in Long Run for Controllers synthesized under various Robustness Criteria and integer parameters (K,B).
}
\label{tab:arbmineExpectedValues}
\begin{tabular}{|c|c|c||c|c|c|}
\hline

 \multicolumn{3}{|c||}{Arbiter(4,3,2)}
& \multicolumn{3}{c|}{Minepump(8,2,6,2)}
\\
\hline
Robustness & $E(ARB-$ & $E(ARB-$
& Robustness & $E(MP-$ & $E(MP-$
\\

Criterion & $\MPNC)$ & $\GODSC)$
& Criterion & $\MPNC)$ & $\GODSC)$
\\

\hline

 AssumeFalse & \multirow{6}{*}{0.000000} &	\multirow{10}{*}{0.998175}
& AssumeFalse	& \multirow{6}{*}{0.000000} & \multirow{11}{*}{0.997070}  
 \\
\cline{1-1}
\cline{4-4}
BeCorrect &	 &	 
& BeCorrect	& 	 & 	
\\
\cline{1-1}
\cline{4-4}


		
 ResCnt(1,3) &	 &	\oomit{0.992286 &	0.687500}
 & ResCnt(2,8)	& 	& 	 
 
\\
\cline{1-1}	
\cline{4-4}


LenCnt(1,3)	&   & \oomit{0.992286	 &	0.687500}
& LenCnt(2,8) & &
\\
\cline{1-1}	
\cline{4-4}
ResBurst(1,3) &  & \oomit{ 0.992286 &	0.687500}
& ResBurst(2,8) & &	
\\
\cline{1-1}	
\cline{4-4}


LenBurst(1,3)	& & \oomit{0.000000  &	0.687500 &0.992286	 &	0.687500}
& LenBurst(2,8) & &
\\
\cline{1-2}	
\cline{4-5}
		
 ResCntInt(1,3)	& 0.544309 &
 	\oomit{\multirow{2}{*}{0.992286} }
& ResCntInt(2,8)	& \multirow{2}{*}{0.000966}	& 	
\\
\cline{1-2}	
\cline{4-4}

 
 ResBurstInt(1,3) & 0.669069  & \oomit{ 0.992286 &	0.687500}
& ResBurstInt(2,8) & &
\\
\cline{1-2}	
\cline{4-5}


 LenCntInt(1,3)	& 0.768066 &	
& LenCntInt(2,8)	& 0.0027342	&	 
 \\
\cline{1-2}
\cline{4-5}	

 LenBurstInt(1,3)	& 0.835205 & \oomit{	0.992286 &	0.687500}
 & LenBurstInt(2,8)	& 0.004514	& 	

\\
\cline{1-5}	

BeCurrentlyCorrect &	0.687500 &	0.992647 
&  BeCurrentlyCorrect &	0.997070	&	

\\
\hline

		
\end{tabular}
\end{table}

An examination of Table \ref{tab:arbmineExpectedValues} is quite enlightening. We state our main finding.
\begin{itemize}
 \item In both the case studies, the robustness criterion $AssumeTrue$ leads to unrealizable specification whereas all other criteria give realizable specifications.
 \item In both the case studies, for the $\MPNC$ controllers, the Expected value of commitment $D_C$ increases  with the implication ordering given in Figure \ref{fig:robustnessHierarchy}. The value ranges from $0$ to $83\%$ for the Arbiter and $0$ to $99\%$ for the Mine pump. Thus, the robustness criterion has a
 major effect on the performance of the synthesized $\MPNC$ controller. Also, for many robustness criteria, the Expected $D_C$ value is $0$. This happens as these criteria (defined using the Error-Scope formulas $\mathit{NeverInPast}$) are non-recoverable -- once the criterion becomes false it remains false in future. Hence, it is desirable to use recoverable criteria defined using the Error-scope formulas $\mathit{NeverInSuffix}$. 
 \item The $\GODSC$ controller which improves the $\MPNC$ controller by optimizing the
 soft requirement $D_C$ has an overwhelming impact on the expected value of $D_C$. In both the case studies, the value is above $99\%$ irrespective of the robustness criterion.
 Thus, soft-robustness vastly improves the expected performance of the controller and should be preferred over $\MPNC$.
 \item We often get  $\GODSC$ controllers with the same/similar expected value of $D_C$ for several robustness criteria. It should be noted that these are different controllers providing  distinct hard-robustness guarantees. For example, in the $Arbiter(4,3,2)$ case study, the controllers  $BeCurrentlyCorrect$ is must-incomparable with all the others, whereas for $Minepump(8,2,6,2)$, this is must equivalent but not identical with any of the other other controllers. Hence, the circumstances (relaxed assumptions) under which they {\em guarantee} $D_C$ are quite different. This shows that a combination of the hard and the soft robustness, as supported by our tool \DCSYNTH, is useful. 
 \item We have also experimented with various default output orderings. Default values have quite small impact on the expected values of $D_C$ in $\GODSC$. However, they significantly impact the no. of states in resulting controller.
\end{itemize}

\vspace*{-0.5cm}
\section{Contributions and Related Work}
\label{section:discussion}
A robust controller should continue to function (i.e maintain its
commitment) under failure of environmental/plant assumptions. When such failures are transient, the controller should be able to recover from the failure by reestablishing the commitment in bounded time \cite{MRT11,BEJK14}.

The main contribution of this paper is a logical framework of hard robustness for specifying and {\em relaxing} (weakening) assumptions under which a controller works. A formula based technique for relaxing any user specified regular assumption is developed. As the experiments show, such weakening has
a marked impact on the Expected value of Commitment.
Soft robustness pertains to the ability of the controller to maintain commitment ``as much as possible'' irrespective of any assumption; this is an optimization problem. We have proposed a method for synthesizing a controller which guarantees the hard robustness, and it optimizes the soft robustness. With case studies, we have shown the impact of hard and soft robustness on both the worst case and the expected behaviour of the controller. These experiments show that the combination of hard and soft robustness, as proposed here, is beneficial. Logic \qddc\/, based on Duration Calculus of Zhou et al \cite{ZH04}, provides a very powerful vocabulary for stating robustness properties.

Several authors have investigated the notion of robustness \cite{BEJK14,MRT11,ET14}. Bloem {\em et al} \cite{BEJK14} provide a classification of different robustness notions; including the \verb#BeCorrect# and \verb#BeCurrentlyCorrect# criteria. They term our soft robustness as ``Never Give Up'' notion. Ehler et al  \cite{ET14} and Bloem \cite{BKKW15} address the notion of resilience; an ability to recover from errors in bounded error-free period. We have shown that many such notions can be specified in our framework. We also give a method to compare robustness notions by an implication order and we use a must dominance order to compare the worst case behaviour of corresponding supervisors. Moreover, a uniform synthesis method is  applied to these defined notions.

Traditional reactive synthesis has only focused on ``correct-by-construction'' from LTL properties. Recent work increasingly considers regular properties (see  Belta {\em et al}  \cite{UB19}). Controller Synthesis from regular properties was pioneered by Ramadge and Wonham \cite{RW87,RW89,LRT17}. Ehler {\em et al}  \cite{ELTV17} compares the supervisory control with the reactive synthesis framework. Here we enhance the Ramadge-Wonham framework to robustness.

\bibliographystyle{plain}
\bibliography{awRef1}

\clearpage
\appendix
\section{Arbiter \DCSYNTH Specification and Robust controllers}
\label{section:arbiterTextualSpecification}
This section gives the \DCSYNTH specification for $BeCurrentlyCorrect$ robust specification of Arbiter and robust controller synthesis form this specification.
\vspace*{-0.5cm}
\begin{figure}[!b]
\caption{Arbiter specification in \DCSYNTH}
\label{figure:arbiterspecification}
\begin{scriptsize}
\framebox{\parbox[t][][t]{\columnwidth}{
$\begin{array}{l}
\mathrm{\textsf{\#qsf "arbiter"}}\\
\mathrm{\textsf{interface\{}}\\
\quad\mathrm{\textsf{input r1, r2, r3, r4;}}\\
\quad\mathrm{\textsf{output a1, a2, a3, a4, A,C;}}\\
\quad\mathrm{\textsf{constant n=3;}}\\
\mathrm{\textsf{\}}}\\
\mathrm{\textsf{definitions\{}}\\
\mathrm{\textsf{// Specification 1: The Acknowlegments shold be exclusive}}\\
\mathrm{\textsf{dc exclusion()\{}}\\
\quad\mathrm{\textsf{true\textasciicircum\textless (a1 =\textgreater !(a2 $||$ a3 $||$ a4)) \&\& (a2 =\textgreater !(a1 $||$ a3 $||$ a4)) \&\& }}\\
\quad\mathrm{\textsf{(a3 =\textgreater !(a1$||$a2$||$a4)) \&\& (a4=\textgreater !(a1$||$a2$||$a3)))\textgreater;}}\\
\mathrm{\textsf{\}}}\\
\mathrm{\textsf{dc noloss()\{}}\\
\quad\mathrm{\textsf{true\textasciicircum\textless (r1 $||$ r2 $||$ r3 $||$ r4) =\textgreater (a1 $||$ a2 $||$ a3 $||$ a4)\textgreater; }}\\
\mathrm{\textsf{\}}}\\
\mathrm{\textsf{//If bus access (ack) should be granted only if there is a request}}\\
\mathrm{\textsf{dc nospuriousack(a1, r1)\{}}\\
\quad\mathrm{\textsf{true\textasciicircum\textless (a1) =\textgreater (r1)\textgreater; }}\\
\mathrm{\textsf{\}}}\\
\mathrm{\textsf{//n cycle response i.e. slen=n-1}}\\
\mathrm{\textsf{dc response(r1,a1)\{}}\\
\quad\mathrm{\textsf{true\textasciicircum (slen=n-1 \&\& [[r1]])  =\textgreater true\textasciicircum(slen=n-1 \&\& (scount a1 \textgreater= 1)); }}\\
\mathrm{\textsf{\}}}\\
\mathrm{\textsf{dc ArbAssume\_4\_2()\{}}\\
\quad\mathrm{\textsf{true\textasciicircum\textless(!r1 \&\& !r2 \&\& !r3 \&\& !r4) $||$ (r1 \&\& !r2 \&\& !r3 \&\& !r4) $||$ }}\\
\quad\mathrm{\textsf{(!r1 \&\& r2 \&\& !r3 \&\& !r4) $||$ (!r1 \&\& !r2 \&\& r3 \&\& !r4) $||$  }}\\
\quad\mathrm{\textsf{(!r1 \&\& !r2 \&\& !r3 \&\& r4)  $||$ (r1 \&\& r2 \&\& !r3 \&\& !r4) $||$}}\\
\quad\mathrm{\textsf{(r1 \&\& !r2 \&\& r3 \&\& !r4) $||$ (r1 \&\& !r2 \&\& !r3 \&\& r4)  $||$}}\\
\quad\mathrm{\textsf{(!r1 \&\& r2 \&\& r3 \&\& !r4) $||$ (!r1 \&\& r2 \&\& !r3 \&\& r4) $||$  }}\\
\quad\mathrm{\textsf{(!r1 \&\& !r2 \&\& r3 \&\& r4) \textgreater;  }}\\

\mathrm{\textsf{\}}}\\

\mathrm{\textsf{dc guranteeInv()\{}}\\
\quad\mathrm{\textsf{exclusion() \&\& noloss(HH2O, HCH4, PUMPON) \&\& nospuriousack(a1,r1) \&\&}}\\ 
\quad\mathrm{\textsf{nospuriousack(a2,r2) \&\& nospuriousack(a3,r3) \&\& nospuriousack(a4,r4) \&\&}}\\ 
\quad\mathrm{\textsf{nospuriousack(a5,r5);}}\\ 
\mathrm{\textsf{\}}}\\

\mathrm{\textsf{dc guranteeResp()\{}}\\
\quad\mathrm{\textsf{response(r1,a1) \&\& response(r2,a2) \&\& response(r3,a3) \&\&}}\\ 
\quad\mathrm{\textsf{response(r4,a4) \&\& response(r5,a5);}}\\ 
\mathrm{\textsf{\}}}\\

\mathrm{\textsf{dc ArbCommit\_4\_3()\{}}\\
\quad\mathrm{\textsf{guranteeInv() \&\& guranteeResp() ;}}\\ 
\mathrm{\textsf{\}}}\\
\mathrm{\textsf{\}}}\\

\mathrm{\textsf{indefinitions\{}}\\
\quad\mathrm{\textsf{A : ArbAssume\_4\_2();}}\\
\quad\mathrm{\textsf{C : ArbCommit\_4\_3();}}\\
\mathrm{\textsf{\}}}\\
\mathrm{\textsf{hardreq\{}}\\
\quad\mathrm{\textsf{useind A, C;}}\\
\quad\mathrm{\textsf{BeCurrentlyCorrect(A) =\textgreater EP(C);}}\\
\mathrm{\textsf{\}}}\\
\mathrm{\textsf{softreq\{}}\\
\quad\mathrm{\textsf{useind C;}}\\
\quad\mathrm{\textsf{(C);}}\\
\mathrm{\textsf{\}}}\\
\end{array}$}
}
\end{scriptsize}
\end{figure}

\oomit{
\begin{verbatim}

m2l-str;
var2 r1, r2, r3, r4, r5, a1, a2, a3, a4, a5, ga3;
var1 i;

(
all2 a1N, a2N, a3N, a4N, a5N, ga3N:
((import("ArbHardAssume(5,3,2)/GODSC.dfa", 
       r1->r1, r2->r2, r3->r3, r4->r4, r5->r5, 
       a1->a1N, a2->a2N, a3->a3N, a4->a4N, a5->a5N, ga3->ga3N)
)=>  (i in ga3N) 
))
=>
(
(import("ArbHardAssumeSoft(5,3,2)/GODSC.dfa", 
       r1->r1, r2->r2, r3->r3, r4->r4, r5->r5, 
       a1->a1, a2->a2, a3->a3, a4->a4, a5->a5, ga3->ga3) 
)=>(i in ga3));

\end{verbatim}
}


\subsubsection{Comparison based on Must Dominance}
This section shows the comparison of supervisors $\MPNC$ and $\GODSC$ for various robust specifications $(D_A,D_C,\Rb(A))$, using must dominance measure formulated in definition \ref{def:mustDominance}.

As the supervisors are finite state mealy machines and a commitment $D_C$ is a regular property, we can use validity checking of \qddc to check whether 
$Sup_1 \leq_{must}^{D_C} Sup_2$ for supervisors $Sup_1$ and $Sup_2$.  In tool \DCSYNTH, we provide a facility to decide must-dominance between two supervisors. The tool also gives a  counter example if must dominance fails. Following are the major observations.

\begin{itemize}
\item $\MPNC$ supervisors for $Arbiter(4,3,2)$  
are found to follow exactly the same robustness order as given by theoretical result in Figure \ref{fig:robustnessHierarchy}. The "$==$" indicates that the supervisors are \emph{identical}, whereas "$=$" indicates that they are \emph{must equivalent} but not identical.

\item $\GODSC$ supervisors does not have theoretical order. However, it is observed for $Arbiter(4,3,2)$ that $\GODSC$ supervisor for  $BeCurrentlyCorrect$ specification becomes incomparable with all other supervisors. All the other supervisors become identical as shown in Figure \ref{fig:arbGODSCRobustnessHierarchy}. 
We use "$==$" to indicate that the supervisors are \emph{identical}, whereas "$=$" indicates that they are \emph{must equivalent} but not identical.

\end{itemize}

\oomit{
\begin{verbatim}

MUST DOMINENCE RESULTS SUMMARIZED

MPNC comparison

   Minepump

		BeCorrectnd < KBSepBurstnd < (KBResCntnd == KBLenCntnd == KBResCntesBurst) 
		< (KBResCntIntnd == KBResBurstIntnd ==  KBLenCntIntnd) < BeCurrentlyCorrectnd


############
   Arbiter

	track1: BeCorrectnd < KBSepBurstnd < KBResCntnd < KBResCntIntnd < KBResBurstIntnd 
	(KBResBurstIntnd, KBLenCntIntnd and BeCurrentlyCorrectnd are incomparable)
	
	track1a:  BeCorrectnd < KBSepBurstnd < KBResCntnd < KBLenCntnd < KBResBurstnd 
	< KBResBurstIntnd
	
	track2: KBResCntIntnd < KBLenCntIntnd
	track2a: KBLenCntnd < KBLenCntIntnd
	track3: BeCorrectnd < BeCurrentlyCorrectnd


################################################################################################################
GODSC comparison

   Minepump

		Greedynd = BeCorrectcnd = KBSepBurstcnd = (KBResCntcnd == KBLenCntcnd == KBResBurstcnd)
		 = (KBResCntIntcnd == KBResBurstIntcnd == KBLenCntIntcnd) = BeCurrentlyCorrectcnd


############
   Arbiter
		Greedynd = (BeCorrectcnd == KBSepBurstcnd == KBResCntcnd == KBLenCntcnd == KBResBurstcnd) 
		< (KBResCntIntcnd == KBResBurstIntcnd) < KBLenCntIntcnd 
		(BeCurrentlyCorrectcnd is incomparable with BeCorrectcnd, KBLenCntIntcnd, KBResBurstIntcnd and Greedynd)


\end{verbatim}

\begin{verbatim}
################################################################################################################
Determinized-MPNC comparison

   Minepump(PumpOn)  		
   
   BeCorrectd1 < KBSepBurstd1 < (KBResCntd1 == KBLenCntd1 == KBResBurstd1) < (KBResCntIntd1 
   == KBResBurstIntd1 == KBLenCntIntd1) < BeCurrentlyCorrectd1

   Minepump(PumpOff) 		BeCorrectd0 = KBSepBurstd0 =  (KBResCntd0  == KBLenCntd0 == KBResBurstd0) = (KBResCntIntd0 
   == KBResBurstIntd0 == KBLenCntIntd0) = BeCurrentlyCorrectd0


###########
   Arbiter(a1>a2>a3>a4) 	
   
   track1. BeCorrectd < KBSepBurstd < (KBResCntd == KBLenCntd) < KBResCntIntd < KBResBurstIntd
   track1a. BeCorrectd < KBSepBurstd < (KBResCntd == KBLenCntd) < KBResBurstd < KBResBurstIntd 
				track2. BeCorrectd<BeCurrentlyCorrectd
				track3. KBResCntIntd<KBLenCntIntd
				(BeCurrentlyCorrectd, KBLenCntIntd and KBResBurstIntd are incomparable)

##################################################################################################################
Determinized GODSC comparison

   Minepump(PumpOn)   
   
   (BeCorrectcd1 == KBSepBurstcd1 == KBResCntcd1 == KBLenCntcd1 == KBResBurstcd1 
   == KBResCntIntcd1 == KBResBurstIntcd1 == KBLenCntIntcd1 
   == BeCurrentlyCorrectcd1) = Greedyd1 

   Minepump(PumpOff) 	
   
   BeCorrectcd0 = KBSepBurstcd0 = (KBResCntcd0 == KBLenCntcd0 == KBResBurstcd0) 
   = (KBResCntIntcd0 == KBResBurstIntcd0 == KBLenCntIntcd0)  = BeCurrentlyCorrectcd0

############
   Arbiter		
   
   (BeCorrectcd == KBSepBurstcd == KBResCntcd == KBLenCntcd == KBResBurstcd 
   == KBResCntIntcd == KBResBurstIntcd == KBLenCntIntcd) = Greedyd 
   (BeCurrentlyCorrectcd is incomparable with BeCorrectcd, KBLenCntIntcd, KBResCntIntcd, 
   KBResBurstIntcd and Greedyd)


\end{verbatim}
}

\oomit{
\begin{figure}
\centering
\begin{tikzpicture}
\label{fig:arbMPNCRobustnessHierarchy}

     \node[ellipse] (10) at (0,0) {KBLenCntInt};
    \node[ellipse] (9) at (3,0) {KBResBurstInt};
    \node[ellipse] (8) at (6,0) {BeCurrentlyCorrect};
    \node[ellipse] (7) at (0, -2)  {KBResCntInt};
    \node[ellipse] (6) at (3, -1) {KBResBurst};
    \node[ellipse] (5) at (3, -2) {KBLenCnt};
    \node[ellipse] (4) at (3, -3) {KBResCnt};
    \node[ellipse] (3) at (3, -4) {KBSepBurst};
    \node[ellipse] (2) at (4, -5) {BeCorrect};
    \draw [->] (2) edge (3);
    \draw [->] (3) edge (4);
    \draw [->] (4) edge (5);
    \draw [->] (4) edge (7);
    \draw [->] (5) edge (6);
    \draw [->] (5) edge (10);
    \draw [->] (7) edge (9);
    \draw [->] (7) edge (10);
    \draw [->] (2) edge (8);
    \draw [->] (6) edge (9);    
\end{tikzpicture}
\caption{Robust synthesis criteria hierarchy for Must Dominance between \MPNC for (Arbiter \TYPEONE).}
\end{figure}
}

\begin{figure}
\begin{center}
\begin{tikzpicture}

    \node (3) at (6,0) {\textit{BeCurrentlyCorrect}};

	\node (4) at (0,0) 			
	{$\left(
	\begin{array}{l}
	LenCntInt==LenBurst== \\
	ResCntInt==ResBurstInt== \\
	AssumeFalse=BeCorrect==ResCnt== \\
	LenCnt==ResBurst==LenBurst
	\end{array}	
	\right)$};

\end{tikzpicture}
\end{center}
\caption{Robustness order for Must Dominance between \GODSC for $Arbiter(4,3,2)$.}
\label{fig:arbGODSCRobustnessHierarchy}

\end{figure}

\subsection{Expected Case performance}
Expected values of various robust controllers are given in Table \ref{tab:arbExpectedValues}. It is evident that expected values of meeting the commitment using hard robustness (given as \MPNC in column 2) as well as soft robustness (given as \GODSC in column 4) are useful and maximize the holding of commitments.

\begin{table}
\caption{Expected Values of meeting the commitments for various robust controllers of Arbiter(4,3,2). The default output order to determinize the supervisor is $a1>a2>a3>a4$.}
\label{tab:arbExpectedValues}
\begin{tabular}{|c|c|c|c|c|}
\hline

Arbiter(4,3,2) & \multicolumn{2}{c|}{\MPNC} & \multicolumn{2}{c|}{\GODSC}
\\
\hline
& $E(C)$ & $E(A)$ & $E(C)$ & $E(A)$\\
\hline

 AssumeFalse & 0.00000 & 0.687500 &	0.998175	& 0.687500\\
\hline

BeCorrect &	0.000000	& 0.687500 &	\multirow{5}{*}{0.998175} &	\multirow{5}{*}{0.687500} \\
\cline{1-3}


		
 ResCnt(1,3) &	0.000000 &	0.687500 & &	\oomit{0.992286 &	0.687500}
\\
\cline{1-3}	


LenCnt(1,3)	& 0.000000  &	0.687500 & & \oomit{0.992286	 &	0.687500}
\\
\cline{1-3}	
ResBurst(1,3) & \multirow{2}{*}{0.000000} & \multirow{2}{*}{0.687500} & & \oomit{ 0.992286 &	0.687500}	
\\
\cline{1-1}	


LenBurst(1,3)	& & & & \oomit{0.000000  &	0.687500 &0.992286	 &	0.687500}
\\
\hline	
		
 ResCntInt(1,3)	& 0.544309 &	0.687500 &	\multirow{2}{*}{0.998175} & \multirow{2}{*}{0.687500}
\\
\cline{1-3}	

 
 ResBurstInt(1,3) & 0.669069 &	0.687500 & & \oomit{ 0.992286 &	0.687500}
\\
\hline


 LenCntInt(1,3)	& 0.7680663 &	0.687500 &	\multirow{2}{*}{0.998175} &	\multirow{2}{*}{0.687500}
\\
\cline{1-3}	

 LenBurstInt(1,3)	& 0.835205 &	0.687500 & & \oomit{	0.992286 &	0.687500}
\\
\hline	

BeCurrentlyCorrect &	0.687500 &	0.687500 &	0.992647 & 0.687500
\\
\hline


\end{tabular}
\end{table}

\subsection{Tool Performance}
\label{section:arbiterToolPerformance}
{\color{black} In this section we provide the detailed statistics on synthesis of various robust supervisors/controllers for Arbiter examples. Table~\ref{tab:performanceRobustArbiterSynthesis} provides the details for Arbiter example.
Appendix \ref{section:minepumpcasestudy} (Table~\ref{tab:performanceRobustMinepumpSynthesis}) provides the Minepump specification and related dominance and performance study.

The Table \ref{tab:performanceRobustArbiterSynthesis} provides the number of states and time required to synthesize the Monitor Automaton, the corresponding \MPNC, \GODSC and Controller (using default value). It is evident that the Monitor construction from \qddc specification and the bounded horizon computation of \GODSC based on soft requirements takes most of the time.
It can also be noted that \GODSC for most of the robust specifications (e.g. BeCorrect, KBLenCnt, KBLenBurst, KBResCnt, KBResBurst) having different behaviour in \MPNC specification becomes identical for \GODSC specification (See  Controller stats in Table \ref{tab:performanceRobustArbiterSynthesis}).
This shows the effect of soft requirements in synthesis of optimal controller.

\begin {table}[t]
\caption {\DCSYNTH Statistics for synthesis of various Robust Arbiter example $Arbiter(4,3,2)$.
Parameters Used: Disc Factor = 0.900000, H = 50, Type = Avg-Max, Delta = 0.000100} 
\label{tab:performanceRobustArbiterSynthesis}
\begin{center}
	\begin{tabular}	{|c|c||c|c|c|c|}
	\hline
	& 
	\multicolumn{1}{c||}{} 
	& \multicolumn{4}{c|}	 {\textbf{ Synthesis (States/Time)}}
	 
	\\
	\hline
		\textbf{Sr} &
		\textbf{Robustness}  & 
		\textbf{Monitor}  & 
		\textbf{\MPNC} & 
		\textbf{\GODSC} &
	 	\textbf{Controller} 
		\\
		\textbf{No} &
		\textbf{Criterion}  &
		\textbf{Stats} &
		\textbf{Stats} &
		\textbf{Stats} &
		\textbf{Stats} 
		\\
	     \hline
	     \multicolumn{6}{|c|}{\MPNC} \\
	     \hline
	     & AssumeFalse  & 
	     83/0.085  & 
	      82/0.001649 &
	      1/0.056157 &
	      2/0.000066   
	     \\
	     \hline
	     & 
	    BeCorrect  & 
	     125/0.062  & 
	     91/0.001059  &
	     11/0.060316  &    
	     9/0.001715 \\
	     \hline
	      &
	     BeCurrentlyCorrect  & 
	     116/0.076  & 
	      49/0.001497 &
	      49/0.020199 & 
	      8/0.009847  
	     \\
	     \hline
	     & 
	    ResCnt(1,3)  & 
	     333/0.093  & 
	      143/0.002115 &
	      63/0.076879 & 
	      39/0.018669   
	     \\
	     \hline
	     & 
	    ResCntInt(1,3)  & 
	     390/0.097  & 
	      194/0.003133 &
	      92/0.110108 & 
	      40/0.044147  
	     \\
	     \hline
	     & 
	    ResBurst(1,3)  & 
	     249/0.096  & 
	      125/0.002435 &
	      45/0.071756    &
	      25/0.009753 
	     \\
	     \hline
	     & 
	    ResBurstInt(1,3)  & 
	     291/0.079  & 
	      167/0.003822 &
	      65/0.098957   &
	      27/0.024754    
	     \\
	     \hline
	     & 
	    LenCnt(1,3) & 
	     291/0.097  & 
	      134/0.002498  &
	      54/0.072853 &
	      32/0.013871  
	     \\
	     \hline
	     & 
	    LenCntInt(1,3) & 
	     335/0.085  & 
	      166/0.005103 &
	      101/0.086818 &
	      33/0.016391    
	     \\
	     \hline
	      & 
	    LenBurst(1,3) & 
	     249/0.068  & 
	      125/0.001902 &
	      45/0.067751 &
	      25/0.009520
	     \\
	     \hline
			& 
	    LenBurstInt(1,3) & 
	     273/0.066  & 
	      143/0.002874 &
	      78/0.072999 &
	      25/0.007139
	     \\
	     \hline
	     & 
	    AssumeTrue  & 
	     \multicolumn{4}{c|}{Unrealizable}   
	     \\
	     \hline
		\multicolumn{6}{|c|}{\GODSC} \\
	     \hline		& 
	    AssumeFalse  & 
	     83/0.078  & 
	      82/0.000963 &
	      62/0.331783 &
	      45/0.007677   
	     \\
	     \hline
	      & 
	    BeCorrect  & 
	     125/0.084  & 
	     91/0.001706  &
	     62/0.345245 &    
	     45/0.007486 \\
	     \hline
	      & 
	     
	    BeCurrentlyCorrect & 
	      116/0.060 & 
	      49/0.001021 &
	      44/0.161940 & 
	      30/0.004811 
	     \\
	     \hline
\oomit{	      & 
	    KBSepBurst(1,3)  & 
	     333/  & 
	      144/0.003455 &
	      144/0.076732 &   
	      40/0.029527 
	     \\
	     \hline
	      & 
	    KBSepBurst(1,3) & 
	     333  & 
	      144/0.002055 &
	      77/0.392736  &   
	      43/0.016500
	     \\
	     \hline
	   }
	      
	      & 
	    ResCnt(1,3) & 
	     333/0.085  & 
	      143/0.002007 &
	      62/0.391268  & 
	      45/0.007501 
	     \\
	     \hline
	      
	      & 
	    ResCntInt(1,3)  & 
	     390/0.088  & 
	      194/0.002743 &
	       62/0.466318 & 
	      45/0.008130
	     \\
	     \hline
		 
	      & 
	    ResBurst(1,3)  & 
	     249/0.072  & 
	      125/0.001617  &
	      62/0.371823  &
	      45/0.007497 
	     \\
	     \hline	     
	     
	      & 
	    ResBurstInt(1,3)  & 
	     291/0.070  & 
	      167/0.002580 &
	      62/0.449197 &
	      45/0.007509  
	     \\
	     \hline

	      & 
	    LenCnt(1,3)  & 
	     291/0.102  & 
	      134/0.002341 &
	     62/0.380235 &
	      45/0.007851
	     \\
	     \hline

	      & 
	    LenCntInt(1,3)  & 
	     335/0.065  & 
	      166/0.002191 &
	      62/0.509995 &
	      45/0.014236 
	     \\
	     \hline
	     
	      & 
	    LenBurst(1,3)  & 
	     249/0.083  & 
	      125/0.002779 &
	      62/0.374552  &
	      45/0.007500 
	     \\
	     \hline	    
		 
	      & 
	    LenBurstInt(1,3)  & 
	     273/0.090  & 
	      143/0.002446 &
	       62/0.378385 &
	      45/0.007507
	     \\
	     \hline	     
	     & AssumeTrue  & 
	     \multicolumn{4}{c|}{Unrealizable}   
	     \\
	     \hline
	\end{tabular}
\end{center}
\end{table}

\oomit{
\begin{table}
\caption{Synthesis of Robust Arbiters $Arb-Inv_{Hard}^{Assume}(4,3,2)$ in DCSynth}
\label{tab:robustArbiterSynthesis}
\begin{tabular}{|c|c|c|c|c|c|c|c|c|}
\hline
 Robustness  & \multicolumn{3}{|c|}{Hard} & \multicolumn{3}{|c|}{Synthesis Parameters} & \multicolumn{2}{|c|}{\GODSC} \\
\cline{2-9}
 Notion & \MPNC & State & Time & $\gamma$ & $\epsilon$ & $H$ & States & Time\\
 \hline
 \verb#Be-Correct(A,C)# & 75 & 35 & 0.1 & 0.9 & 0.001 & 4 & 49 & 0.47 \\
 \hline
 \verb#Be-Currently-Correct(A,C)# & Un & - & - & - & - & - & - & - \\
 \hline
 \verb#k-Atmost(A,C,1)# & 118 & 55 & 0.123 & 0.9 & 0.001 & 4 & 49 & 0.58 \\
 \hline
 \verb#k,b-Resilient(A,C,1,4)# & 136 & 72 & 0.13 & 0.9 & 0.001 & 4 & 49 & 0.62 \\
 \hline
 \verb#k,b-Variant-1(A,C,1,4)# & 149 & 72 & 0.15 & 0.9 & 0.001 & 4 & 49 & 0.69 \\
 \hline
 \verb#k,b-Variant-2(A,C,1,4)# & 136 & 64 & 0.14 & 0.9 & 0.001 & 4 & 49 & 0.65\\
 \hline
 \verb#k,b-Sep(A,C,1,3)# & 128 & 66 & 0.13 & 0.9 & 0.001 & 4 & 49 & 0.61 \\
 \hline
 \verb#Greedy(A,C)# & 66 & - & - & 0.9 & 0.001 & 4 & 49 & 0.42\\
 \hline
\end{tabular}
\end{table}

\begin{table}
\caption{Synthesis of Robust Arbiters $Arb_{Hard}^{Assume}(4,3,1)$ in DCSynth}
\label{tab:robustArbiterSynthesis}
\begin{tabular}{|c|c|c|c|c|c|c|c|c|}
\hline
 Robustness  & \multicolumn{3}{|c|}{Hard} & \multicolumn{3}{|c|}{Synthesis Parameters} & \multicolumn{2}{|c|}{\GODSC} \\
\cline{2-9}
 Notion & \MPNC & State & Time & $\gamma$ & $\epsilon$ & $H$ & States & Time\\
 \hline
 \verb#Be-Correct(A,C)# & 83 & 3 & 0.6 & 0.9 & 0.001 & 50 & 46 & 32.48 \\
 \hline
 \verb#Be-Currently-Correct(A,C)# & 82 & 2 & 0.57 & 0.9 & 0.001 & 50 & 46 & 30.98 \\
 \hline
 \verb#k-Atmost(A,C,1)# & 98 & 4 & 0.60 & 0.9 & 0.001 & 50 & 46 & 32.61 \\
 \hline
 \verb#k-Atmost(A,C,2)# & 163 & 5 & 0.64 & 0.9 & 0.001 & 50 & 46 & 34.85 \\
 \hline
 \verb#k,b-Resilient(A,C,1,4)# & 100 & 7 & 0.61 & 0.9 & 0.001 & 50 & 46 & 32.88 \\
 \hline
 \verb#k,b-Resilient(A,C,2,4)# & 167 & 11 & 0.65 & 0.9 & 0.001 & 50 & 46 & 35.0 \\ \hline
 \verb#k,b-Variant-1(A,C,1,4)# & 229 & 11 & 1.09 & 0.9 & 0.001 & 50 & 46 & 59.65 \\
 \hline
 \verb#k,b-Variant-1(A,C,2,4)# & 374 & 14 & 1.79 & 0.9 & 0.001 & 50 & 46 & 96.88 \\
 \hline
  \verb#k,b-Variant-2(A,C,1,4)# & 206 & 8 & 1.01 & 0.9 & 0.001 & 50 & 46 & 54.79\\
 \hline
 \verb#k,b-Variant-2(A,C,2,4)# & 236 & 8 & 1.30 & 0.9 & 0.001 & 50 & 46 & 70.26\\
 \hline
 \verb#k,b-Sep(A,C,1,3)# & 116 & 7 & 0.65 & 0.9 & 0.001 & 50 & 46 & 35.2 \\
 \hline
 \verb#k,b-Sep(A,C,2,2)# & 180 & 8 & 0.69 & 0.9 & 0.001 & 50 & 46 & 37.58 \\
 \hline
 \verb#Greedy(A,C)# & 82 & - & - & 0.9 & 0.001 & 50 & 46 & 32.3\\
 \hline
\end{tabular}
\end{table}

\begin{table}
\caption{Synthesis of Robust Arbiters $Arb-Inv_{Hard}^{Assume}(4,3,1)$ in DCSynth}
\label{tab:robustArbiterSynthesis}
\begin{tabular}{|l|c|c|c|c|c|c|c|c|}
\hline
 Robustness  & \multicolumn{3}{|c|}{Hard} & \multicolumn{3}{|c|}{Synthesis Parameters} & \multicolumn{2}{|c|}{\GODSC} \\
\cline{2-9}
 Notion & \MPNC & State & Time & $\gamma$ & $\epsilon$ & $H$ & States & Time\\
 \hline
 \verb#Be-Correct(A,C)# & 66 & 28 & 0.08 & 0.9 & 0.001 & 4 & 49 & 0.40 \\
 \hline
 \verb#Be-Currently-Correct(A,C)# & 66 & 28 & 0.08 & 0.9 & 0.001 & 4 & 49 & 0.41 \\
 \hline
 \verb#k-Atmost(A,C,1)# & 66 & 28 & 0.08 & 0.9 & 0.001 & 4 & 49 & 0.40 \\
 \hline
 \verb#k-Atmost(A,C,2)# & 66 & 28 & 0.08 & 0.9 & 0.001 & 4 & 49 & 0.40 \\
 \hline 
 \verb#k,b-Resilient(A,C,1,4)# & 66 & 28 & 0.09 & 0.9 & 0.001 & 4 & 49 & 0.40 \\
 \hline
 \verb#k,b-Resilient(A,C,2,4)# & 66 & 28 & 0.09 & 0.9 & 0.001 & 4 & 49 & 0.41 \\
 \hline
 \verb#k,b-Variant-1(A,C,1,4)# & 66 & 28 & 0.08 & 0.9 & 0.001 & 4 & 49 & 0.41 \\
 \hline
 \verb#k,b-Variant-1(A,C,2,4)# & 66 & 28 & 0.08 & 0.9 & 0.001 & 4 & 49 & 0.41 \\
  \hline
 \verb#k,b-Variant-2(A,C,1,4)# & 66 & 28 & 0.08 & 0.9 & 0.001 & 4 & 49 & 0.41\\
 \hline
 \verb#k,b-Variant-2(A,C,2,4)# & 66 & 28 & 0.08 & 0.9 & 0.001 & 4 & 49 & 0.40\\
 \hline
 \verb#k,b-Sep(A,C,1,3)# & 66 & 28 & 0.1 & 0.9 & 0.001 & 4 & 49 & 0.41 \\
 \hline
 \verb#k,b-Sep(A,C,2,2)# & 66 & 28 & 0.1 & 0.9 & 0.001 & 4 & 49 & 0.41 \\
 \hline
 \verb#Greedy(A,C)# & 66 & - & - & 0.9 & 0.001 & 4 & 49 & 0.40\\
 \hline
\end{tabular}
\end{table}
}

\oomit{
\begin{table}
\caption{Synthesis of Robust Arbiters $Minepump(8,2,6,2)$ in DCSynth}
\label{tab:robustArbiterSynthesis}
\begin{tabular}{|l|c|c|c|c|c|c|c|c|c|c|}
\hline
 Robustness  & \multicolumn{4}{|c|}{Hard} & \multicolumn{3}{|c|}{Synthesis Parameters} & \multicolumn{3}{|c|}{\GODSC} \\
\cline{2-11}
 Notion & $A^{hard}$ & $A^{mpnc}$ & $cntrl$ & Time & $\gamma$ & $\epsilon$ & $H$ & $A^{soft}$ & $A^{godc}$ & Time\\
 \hline
 \verb#Be-Correct(A,C)# & 6235 & 5326 & 202 & 0.22 & 0.9 & 0.001 & 10 & 5204 & 1433 & 0.68 \\
 \hline
 \verb#Be-Currently-Correct(A,C)# & 6038 & 801 & 305 & 0.07 & 0.9 & 0.001 & 10 & 5204  & 305 & 0.11 \\
 \hline
 \verb#k-Atmost(A,C,2)# & 7524 & 5323 & 367 & 0.22 & 0.9 & 0.001 & 10 & 5204 & 1433 & 0.66 \\
 \hline
 \verb#k,b-Resilient(A,C,2,10)# & 7524 & 5323 & 367 & 0.22 & 0.9 & 0.001 & 10 & 5204 & 1433 & 0.65 \\
 \hline
 \verb#k,b-Variant-1(A,C,2,10)# & 16650 & 8330 & 959 & 0.37 & 0.9 & 0.001 & 10 & 5204 & 1510 & 0.94 \\
 \hline
 \verb#k,b-Variant-2(A,C,2,10)# & 13558 & 8330 & 943 & 0.35 & 0.9 & 0.001 & 10 & 5204 & 1510 & 0.94\\
 \hline
 \verb#k,b-Sep(A,C,2,8)# & 9040 & 5635 & 501 & 0.23 & 0.9 & 0.001 & 10 & 5204 & 1433 & 0.69 \\
 \hline
 \verb#Greedy(A,C)# & 5204 & 5203 & - & - & 0.9 & 0.001 & 10 & 5204 & 1433 & 0.62\\
 \hline

\end{tabular}
\end{table}

\begin{table}
\caption{Synthesis of Robust Arbiters $Minepump(8,2,6,2)$ in DCSynth (WITH REQUIRED INDICATORS ONLY)}
\label{tab:robustArbiterSynthesis}
\begin{tabular}{|l|c|c|c|c|c|c|c|c|c|c|}
\hline
 Robustness  & \multicolumn{4}{|c|}{Hard} & \multicolumn{3}{|c|}{Synthesis Parameters} & \multicolumn{3}{|c|}{\GODSC} \\
\cline{2-11}
 Notion & $A^{hard}$ & $A^{mpnc}$ & $cntrl$ & Time & $\gamma$ & $\epsilon$ & $H$ & $A^{soft}$ & $A^{godc}$ & Time\\
 \hline
 \verb#Be-Correct(A,C)# & 6235 & 5326 & 202 & 0.21  & 0.9 & 0.001 & 10 & 56 & 1433 & 0.67 \\
 \hline
 \verb#BeCurrentlyCorrect(A,C)# & 6038 & 801 & 305 & 0.04  & 0.9 & 0.001 & 10 & 56 & 305 & 0.08 \\
 \hline
 \verb#k-Atmost(A,C,2)# & 7524 & 5323 & 367 & 0.20  & 0.9 & 0.001 & 10 & 56 & 1433 & 0.63 \\
 \hline
 \verb#k,b-Resilient(A,C,2,10)# & 7524 & 5323 & 367 & 0.20  & 0.9 & 0.001 & 10 & 56 & 1433 & 0.64 \\
 \hline
 \verb#k,b-Variant-1(A,C,2,10)# & 16650 & 8330 & 959 & 0.34  & 0.9 & 0.001 & 10 & 56 & 1510 & 0.93 \\
 \hline
 \verb#k,b-Variant-2(A,C,2,10)# & 13558 & 8330 & 943 & 0.33  & 0.9 & 0.001 & 10 & 56 & 1510 & 0.95 \\
 \hline
 \verb#k,b-Sep(A,C,2,8)# & 9040 & 5635 & 501 & 0.22  & 0.9 & 0.001 & 10 & 56 & 1433 & 0.67 \\
 \hline
 \verb#Greedy(A,C)# & 2 & 1 & - & -  & 0.9 & 0.001 & 10 & 56 & 47 & 0.02 \\
 \hline
\end{tabular}
\end{table}

\begin{figure}
{
\def\lignefine{\linethickness{0.05pt}}
\def\ligneepaisse{\linethickness{2pt}}
\noindent
\setlength{\unitlength}{1mm}
\begin{picture}(136.0000,60.0000)(-5.0000,0.0000)
\fboxsep 0pt
\lignefine
\color{black}
\multiput(0.0000,-5.0000)(4.0000,0){30}{\line(0,1){60.0000}}
\put(2.0000,54.0000){\scriptsize\makebox(0,0)[t]{1}}
\put(2.0000,-4.0000){\scriptsize\makebox(0,0)[b]{1}}
\put(6.0000,54.0000){\scriptsize\makebox(0,0)[t]{2}}
\put(6.0000,-4.0000){\scriptsize\makebox(0,0)[b]{2}}
\put(10.0000,54.0000){\scriptsize\makebox(0,0)[t]{3}}
\put(10.0000,-4.0000){\scriptsize\makebox(0,0)[b]{3}}
\put(14.0000,54.0000){\scriptsize\makebox(0,0)[t]{4}}
\put(14.0000,-4.0000){\scriptsize\makebox(0,0)[b]{4}}
\put(18.0000,54.0000){\scriptsize\makebox(0,0)[t]{5}}
\put(18.0000,-4.0000){\scriptsize\makebox(0,0)[b]{5}}
\put(22.0000,54.0000){\scriptsize\makebox(0,0)[t]{6}}
\put(22.0000,-4.0000){\scriptsize\makebox(0,0)[b]{6}}
\put(26.0000,54.0000){\scriptsize\makebox(0,0)[t]{7}}
\put(26.0000,-4.0000){\scriptsize\makebox(0,0)[b]{7}}
\put(30.0000,54.0000){\scriptsize\makebox(0,0)[t]{8}}
\put(30.0000,-4.0000){\scriptsize\makebox(0,0)[b]{8}}
\put(34.0000,54.0000){\scriptsize\makebox(0,0)[t]{9}}
\put(34.0000,-4.0000){\scriptsize\makebox(0,0)[b]{9}}
\put(38.0000,54.0000){\scriptsize\makebox(0,0)[t]{10}}
\put(38.0000,-4.0000){\scriptsize\makebox(0,0)[b]{10}}
\put(42.0000,54.0000){\scriptsize\makebox(0,0)[t]{11}}
\put(42.0000,-4.0000){\scriptsize\makebox(0,0)[b]{11}}
\put(46.0000,54.0000){\scriptsize\makebox(0,0)[t]{12}}
\put(46.0000,-4.0000){\scriptsize\makebox(0,0)[b]{12}}
\put(50.0000,54.0000){\scriptsize\makebox(0,0)[t]{13}}
\put(50.0000,-4.0000){\scriptsize\makebox(0,0)[b]{13}}
\put(54.0000,54.0000){\scriptsize\makebox(0,0)[t]{14}}
\put(54.0000,-4.0000){\scriptsize\makebox(0,0)[b]{14}}
\put(58.0000,54.0000){\scriptsize\makebox(0,0)[t]{15}}
\put(58.0000,-4.0000){\scriptsize\makebox(0,0)[b]{15}}
\put(62.0000,54.0000){\scriptsize\makebox(0,0)[t]{16}}
\put(62.0000,-4.0000){\scriptsize\makebox(0,0)[b]{16}}
\put(66.0000,54.0000){\scriptsize\makebox(0,0)[t]{17}}
\put(66.0000,-4.0000){\scriptsize\makebox(0,0)[b]{17}}
\put(70.0000,54.0000){\scriptsize\makebox(0,0)[t]{18}}
\put(70.0000,-4.0000){\scriptsize\makebox(0,0)[b]{18}}
\put(74.0000,54.0000){\scriptsize\makebox(0,0)[t]{19}}
\put(74.0000,-4.0000){\scriptsize\makebox(0,0)[b]{19}}
\put(78.0000,54.0000){\scriptsize\makebox(0,0)[t]{20}}
\put(78.0000,-4.0000){\scriptsize\makebox(0,0)[b]{20}}
\put(82.0000,54.0000){\scriptsize\makebox(0,0)[t]{21}}
\put(82.0000,-4.0000){\scriptsize\makebox(0,0)[b]{21}}
\put(86.0000,54.0000){\scriptsize\makebox(0,0)[t]{22}}
\put(86.0000,-4.0000){\scriptsize\makebox(0,0)[b]{22}}
\put(90.0000,54.0000){\scriptsize\makebox(0,0)[t]{23}}
\put(90.0000,-4.0000){\scriptsize\makebox(0,0)[b]{23}}
\put(94.0000,54.0000){\scriptsize\makebox(0,0)[t]{24}}
\put(94.0000,-4.0000){\scriptsize\makebox(0,0)[b]{24}}
\put(98.0000,54.0000){\scriptsize\makebox(0,0)[t]{25}}
\put(98.0000,-4.0000){\scriptsize\makebox(0,0)[b]{25}}
\put(102.0000,54.0000){\scriptsize\makebox(0,0)[t]{26}}
\put(102.0000,-4.0000){\scriptsize\makebox(0,0)[b]{26}}
\put(106.0000,54.0000){\scriptsize\makebox(0,0)[t]{27}}
\put(106.0000,-4.0000){\scriptsize\makebox(0,0)[b]{27}}
\put(110.0000,54.0000){\scriptsize\makebox(0,0)[t]{28}}
\put(110.0000,-4.0000){\scriptsize\makebox(0,0)[b]{28}}
\put(114.0000,54.0000){\scriptsize\makebox(0,0)[t]{29}}
\put(114.0000,-4.0000){\scriptsize\makebox(0,0)[b]{29}}
\put(-1.0000,44.0000){\line(1,0){118.0000}}
\put(-1.0000,48.0000){\line(1,0){118.0000}}
\put(-1.0000,38.0000){\line(1,0){118.0000}}
\put(-1.0000,42.0000){\line(1,0){118.0000}}
\put(-1.0000,32.0000){\line(1,0){118.0000}}
\put(-1.0000,36.0000){\line(1,0){118.0000}}
\put(-1.0000,26.0000){\line(1,0){118.0000}}
\put(-1.0000,30.0000){\line(1,0){118.0000}}
\put(-1.0000,20.0000){\line(1,0){118.0000}}
\put(-1.0000,24.0000){\line(1,0){118.0000}}
\put(-1.0000,14.0000){\line(1,0){118.0000}}
\put(-1.0000,18.0000){\line(1,0){118.0000}}
\put(-1.0000,8.0000){\line(1,0){118.0000}}
\put(-1.0000,12.0000){\line(1,0){118.0000}}
\put(-1.0000,2.0000){\line(1,0){118.0000}}
\put(-1.0000,6.0000){\line(1,0){118.0000}}
\ligneepaisse
\color{blue}
\put(-1.0000,46.0000){\color{blue}\normalsize\makebox(0,0)[r]{req1}}
\put(12.0000,44.0000){\line(0,1){4.0000}}
\put(16.0000,48.0000){\line(0,-1){4.0000}}
\put(36.0000,44.0000){\line(0,1){4.0000}}
\put(48.0000,48.0000){\line(0,-1){4.0000}}
\put(72.0000,44.0000){\line(0,1){4.0000}}
\put(100.0000,48.0000){\line(0,-1){4.0000}}
\put(0.0000,44.0000){\line(1,0){4.0000}}
\put(4.0000,44.0000){\line(1,0){4.0000}}
\put(8.0000,44.0000){\line(1,0){4.0000}}
\put(12.0000,48.0000){\line(1,0){4.0000}}
\put(16.0000,44.0000){\line(1,0){4.0000}}
\put(20.0000,44.0000){\line(1,0){4.0000}}
\put(24.0000,44.0000){\line(1,0){4.0000}}
\put(28.0000,44.0000){\line(1,0){4.0000}}
\put(32.0000,44.0000){\line(1,0){4.0000}}
\put(36.0000,48.0000){\line(1,0){4.0000}}
\put(40.0000,48.0000){\line(1,0){4.0000}}
\put(44.0000,48.0000){\line(1,0){4.0000}}
\put(48.0000,44.0000){\line(1,0){4.0000}}
\put(52.0000,44.0000){\line(1,0){4.0000}}
\put(56.0000,44.0000){\line(1,0){4.0000}}
\put(60.0000,44.0000){\line(1,0){4.0000}}
\put(64.0000,44.0000){\line(1,0){4.0000}}
\put(68.0000,44.0000){\line(1,0){4.0000}}
\put(72.0000,48.0000){\line(1,0){4.0000}}
\put(76.0000,48.0000){\line(1,0){4.0000}}
\put(80.0000,48.0000){\line(1,0){4.0000}}
\put(84.0000,48.0000){\line(1,0){4.0000}}
\put(88.0000,48.0000){\line(1,0){4.0000}}
\put(92.0000,48.0000){\line(1,0){4.0000}}
\put(96.0000,48.0000){\line(1,0){4.0000}}
\put(100.0000,44.0000){\line(1,0){4.0000}}
\put(104.0000,44.0000){\line(1,0){4.0000}}
\put(108.0000,44.0000){\line(1,0){4.0000}}
\put(112.0000,44.0000){\line(1,0){4.0000}}
\color{blue}
\put(-1.0000,40.0000){\color{blue}\normalsize\makebox(0,0)[r]{req2}}
\put(4.0000,38.0000){\line(0,1){4.0000}}
\put(8.0000,42.0000){\line(0,-1){4.0000}}
\put(20.0000,38.0000){\line(0,1){4.0000}}
\put(36.0000,42.0000){\line(0,-1){4.0000}}
\put(48.0000,38.0000){\line(0,1){4.0000}}
\put(0.0000,38.0000){\line(1,0){4.0000}}
\put(4.0000,42.0000){\line(1,0){4.0000}}
\put(8.0000,38.0000){\line(1,0){4.0000}}
\put(12.0000,38.0000){\line(1,0){4.0000}}
\put(16.0000,38.0000){\line(1,0){4.0000}}
\put(20.0000,42.0000){\line(1,0){4.0000}}
\put(24.0000,42.0000){\line(1,0){4.0000}}
\put(28.0000,42.0000){\line(1,0){4.0000}}
\put(32.0000,42.0000){\line(1,0){4.0000}}
\put(36.0000,38.0000){\line(1,0){4.0000}}
\put(40.0000,38.0000){\line(1,0){4.0000}}
\put(44.0000,38.0000){\line(1,0){4.0000}}
\put(48.0000,42.0000){\line(1,0){4.0000}}
\put(52.0000,42.0000){\line(1,0){4.0000}}
\put(56.0000,42.0000){\line(1,0){4.0000}}
\put(60.0000,42.0000){\line(1,0){4.0000}}
\put(64.0000,42.0000){\line(1,0){4.0000}}
\put(68.0000,42.0000){\line(1,0){4.0000}}
\put(72.0000,42.0000){\line(1,0){4.0000}}
\put(76.0000,42.0000){\line(1,0){4.0000}}
\put(80.0000,42.0000){\line(1,0){4.0000}}
\put(84.0000,42.0000){\line(1,0){4.0000}}
\put(88.0000,42.0000){\line(1,0){4.0000}}
\put(92.0000,42.0000){\line(1,0){4.0000}}
\put(96.0000,42.0000){\line(1,0){4.0000}}
\put(100.0000,42.0000){\line(1,0){4.0000}}
\put(104.0000,42.0000){\line(1,0){4.0000}}
\put(108.0000,42.0000){\line(1,0){4.0000}}
\put(112.0000,42.0000){\line(1,0){4.0000}}
\color{blue}
\put(-1.0000,34.0000){\color{blue}\normalsize\makebox(0,0)[r]{req3}}
\put(8.0000,32.0000){\line(0,1){4.0000}}
\put(12.0000,36.0000){\line(0,-1){4.0000}}
\put(48.0000,32.0000){\line(0,1){4.0000}}
\put(88.0000,36.0000){\line(0,-1){4.0000}}
\put(0.0000,32.0000){\line(1,0){4.0000}}
\put(4.0000,32.0000){\line(1,0){4.0000}}
\put(8.0000,36.0000){\line(1,0){4.0000}}
\put(12.0000,32.0000){\line(1,0){4.0000}}
\put(16.0000,32.0000){\line(1,0){4.0000}}
\put(20.0000,32.0000){\line(1,0){4.0000}}
\put(24.0000,32.0000){\line(1,0){4.0000}}
\put(28.0000,32.0000){\line(1,0){4.0000}}
\put(32.0000,32.0000){\line(1,0){4.0000}}
\put(36.0000,32.0000){\line(1,0){4.0000}}
\put(40.0000,32.0000){\line(1,0){4.0000}}
\put(44.0000,32.0000){\line(1,0){4.0000}}
\put(48.0000,36.0000){\line(1,0){4.0000}}
\put(52.0000,36.0000){\line(1,0){4.0000}}
\put(56.0000,36.0000){\line(1,0){4.0000}}
\put(60.0000,36.0000){\line(1,0){4.0000}}
\put(64.0000,36.0000){\line(1,0){4.0000}}
\put(68.0000,36.0000){\line(1,0){4.0000}}
\put(72.0000,36.0000){\line(1,0){4.0000}}
\put(76.0000,36.0000){\line(1,0){4.0000}}
\put(80.0000,36.0000){\line(1,0){4.0000}}
\put(84.0000,36.0000){\line(1,0){4.0000}}
\put(88.0000,32.0000){\line(1,0){4.0000}}
\put(92.0000,32.0000){\line(1,0){4.0000}}
\put(96.0000,32.0000){\line(1,0){4.0000}}
\put(100.0000,32.0000){\line(1,0){4.0000}}
\put(104.0000,32.0000){\line(1,0){4.0000}}
\put(108.0000,32.0000){\line(1,0){4.0000}}
\put(112.0000,32.0000){\line(1,0){4.0000}}
\color{blue}
\put(-1.0000,28.0000){\color{blue}\normalsize\makebox(0,0)[r]{req4}}
\put(16.0000,26.0000){\line(0,1){4.0000}}
\put(48.0000,30.0000){\line(0,-1){4.0000}}
\put(56.0000,26.0000){\line(0,1){4.0000}}
\put(112.0000,30.0000){\line(0,-1){4.0000}}
\put(0.0000,26.0000){\line(1,0){4.0000}}
\put(4.0000,26.0000){\line(1,0){4.0000}}
\put(8.0000,26.0000){\line(1,0){4.0000}}
\put(12.0000,26.0000){\line(1,0){4.0000}}
\put(16.0000,30.0000){\line(1,0){4.0000}}
\put(20.0000,30.0000){\line(1,0){4.0000}}
\put(24.0000,30.0000){\line(1,0){4.0000}}
\put(28.0000,30.0000){\line(1,0){4.0000}}
\put(32.0000,30.0000){\line(1,0){4.0000}}
\put(36.0000,30.0000){\line(1,0){4.0000}}
\put(40.0000,30.0000){\line(1,0){4.0000}}
\put(44.0000,30.0000){\line(1,0){4.0000}}
\put(48.0000,26.0000){\line(1,0){4.0000}}
\put(52.0000,26.0000){\line(1,0){4.0000}}
\put(56.0000,30.0000){\line(1,0){4.0000}}
\put(60.0000,30.0000){\line(1,0){4.0000}}
\put(64.0000,30.0000){\line(1,0){4.0000}}
\put(68.0000,30.0000){\line(1,0){4.0000}}
\put(72.0000,30.0000){\line(1,0){4.0000}}
\put(76.0000,30.0000){\line(1,0){4.0000}}
\put(80.0000,30.0000){\line(1,0){4.0000}}
\put(84.0000,30.0000){\line(1,0){4.0000}}
\put(88.0000,30.0000){\line(1,0){4.0000}}
\put(92.0000,30.0000){\line(1,0){4.0000}}
\put(96.0000,30.0000){\line(1,0){4.0000}}
\put(100.0000,30.0000){\line(1,0){4.0000}}
\put(104.0000,30.0000){\line(1,0){4.0000}}
\put(108.0000,30.0000){\line(1,0){4.0000}}
\put(112.0000,26.0000){\line(1,0){4.0000}}
\color{red}
\put(-1.0000,22.0000){\color{red}\normalsize\makebox(0,0)[r]{ack1}}
\put(12.0000,20.0000){\line(0,1){4.0000}}
\put(16.0000,24.0000){\line(0,-1){4.0000}}
\put(36.0000,20.0000){\line(0,1){4.0000}}
\put(40.0000,24.0000){\line(0,-1){4.0000}}
\put(44.0000,20.0000){\line(0,1){4.0000}}
\put(48.0000,24.0000){\line(0,-1){4.0000}}
\put(76.0000,20.0000){\line(0,1){4.0000}}
\put(80.0000,24.0000){\line(0,-1){4.0000}}
\put(84.0000,20.0000){\line(0,1){4.0000}}
\put(88.0000,24.0000){\line(0,-1){4.0000}}
\put(92.0000,20.0000){\line(0,1){4.0000}}
\put(96.0000,24.0000){\line(0,-1){4.0000}}
\put(0.0000,20.0000){\line(1,0){4.0000}}
\put(4.0000,20.0000){\line(1,0){4.0000}}
\put(8.0000,20.0000){\line(1,0){4.0000}}
\put(12.0000,24.0000){\line(1,0){4.0000}}
\put(16.0000,20.0000){\line(1,0){4.0000}}
\put(20.0000,20.0000){\line(1,0){4.0000}}
\put(24.0000,20.0000){\line(1,0){4.0000}}
\put(28.0000,20.0000){\line(1,0){4.0000}}
\put(32.0000,20.0000){\line(1,0){4.0000}}
\put(36.0000,24.0000){\line(1,0){4.0000}}
\put(40.0000,20.0000){\line(1,0){4.0000}}
\put(44.0000,24.0000){\line(1,0){4.0000}}
\put(48.0000,20.0000){\line(1,0){4.0000}}
\put(52.0000,20.0000){\line(1,0){4.0000}}
\put(56.0000,20.0000){\line(1,0){4.0000}}
\put(60.0000,20.0000){\line(1,0){4.0000}}
\put(64.0000,20.0000){\line(1,0){4.0000}}
\put(68.0000,20.0000){\line(1,0){4.0000}}
\put(72.0000,20.0000){\line(1,0){4.0000}}
\put(76.0000,24.0000){\line(1,0){4.0000}}
\put(80.0000,20.0000){\line(1,0){4.0000}}
\put(84.0000,24.0000){\line(1,0){4.0000}}
\put(88.0000,20.0000){\line(1,0){4.0000}}
\put(92.0000,24.0000){\line(1,0){4.0000}}
\put(96.0000,20.0000){\line(1,0){4.0000}}
\put(100.0000,20.0000){\line(1,0){4.0000}}
\put(104.0000,20.0000){\line(1,0){4.0000}}
\put(108.0000,20.0000){\line(1,0){4.0000}}
\put(112.0000,20.0000){\line(1,0){4.0000}}
\color{red}
\put(-1.0000,16.0000){\color{red}\normalsize\makebox(0,0)[r]{ack2}}
\put(4.0000,14.0000){\line(0,1){4.0000}}
\put(8.0000,18.0000){\line(0,-1){4.0000}}
\put(20.0000,14.0000){\line(0,1){4.0000}}
\put(24.0000,18.0000){\line(0,-1){4.0000}}
\put(28.0000,14.0000){\line(0,1){4.0000}}
\put(32.0000,18.0000){\line(0,-1){4.0000}}
\put(48.0000,14.0000){\line(0,1){4.0000}}
\put(52.0000,18.0000){\line(0,-1){4.0000}}
\put(56.0000,14.0000){\line(0,1){4.0000}}
\put(60.0000,18.0000){\line(0,-1){4.0000}}
\put(64.0000,14.0000){\line(0,1){4.0000}}
\put(68.0000,18.0000){\line(0,-1){4.0000}}
\put(72.0000,14.0000){\line(0,1){4.0000}}
\put(76.0000,18.0000){\line(0,-1){4.0000}}
\put(80.0000,14.0000){\line(0,1){4.0000}}
\put(84.0000,18.0000){\line(0,-1){4.0000}}
\put(88.0000,14.0000){\line(0,1){4.0000}}
\put(92.0000,18.0000){\line(0,-1){4.0000}}
\put(96.0000,14.0000){\line(0,1){4.0000}}
\put(100.0000,18.0000){\line(0,-1){4.0000}}
\put(104.0000,14.0000){\line(0,1){4.0000}}
\put(108.0000,18.0000){\line(0,-1){4.0000}}
\put(112.0000,14.0000){\line(0,1){4.0000}}
\put(0.0000,14.0000){\line(1,0){4.0000}}
\put(4.0000,18.0000){\line(1,0){4.0000}}
\put(8.0000,14.0000){\line(1,0){4.0000}}
\put(12.0000,14.0000){\line(1,0){4.0000}}
\put(16.0000,14.0000){\line(1,0){4.0000}}
\put(20.0000,18.0000){\line(1,0){4.0000}}
\put(24.0000,14.0000){\line(1,0){4.0000}}
\put(28.0000,18.0000){\line(1,0){4.0000}}
\put(32.0000,14.0000){\line(1,0){4.0000}}
\put(36.0000,14.0000){\line(1,0){4.0000}}
\put(40.0000,14.0000){\line(1,0){4.0000}}
\put(44.0000,14.0000){\line(1,0){4.0000}}
\put(48.0000,18.0000){\line(1,0){4.0000}}
\put(52.0000,14.0000){\line(1,0){4.0000}}
\put(56.0000,18.0000){\line(1,0){4.0000}}
\put(60.0000,14.0000){\line(1,0){4.0000}}
\put(64.0000,18.0000){\line(1,0){4.0000}}
\put(68.0000,14.0000){\line(1,0){4.0000}}
\put(72.0000,18.0000){\line(1,0){4.0000}}
\put(76.0000,14.0000){\line(1,0){4.0000}}
\put(80.0000,18.0000){\line(1,0){4.0000}}
\put(84.0000,14.0000){\line(1,0){4.0000}}
\put(88.0000,18.0000){\line(1,0){4.0000}}
\put(92.0000,14.0000){\line(1,0){4.0000}}
\put(96.0000,18.0000){\line(1,0){4.0000}}
\put(100.0000,14.0000){\line(1,0){4.0000}}
\put(104.0000,18.0000){\line(1,0){4.0000}}
\put(108.0000,14.0000){\line(1,0){4.0000}}
\put(112.0000,18.0000){\line(1,0){4.0000}}
\color{red}
\put(-1.0000,10.0000){\color{red}\normalsize\makebox(0,0)[r]{ack3}}
\put(8.0000,8.0000){\line(0,1){4.0000}}
\put(12.0000,12.0000){\line(0,-1){4.0000}}
\put(52.0000,8.0000){\line(0,1){4.0000}}
\put(56.0000,12.0000){\line(0,-1){4.0000}}
\put(60.0000,8.0000){\line(0,1){4.0000}}
\put(64.0000,12.0000){\line(0,-1){4.0000}}
\put(68.0000,8.0000){\line(0,1){4.0000}}
\put(72.0000,12.0000){\line(0,-1){4.0000}}
\put(0.0000,8.0000){\line(1,0){4.0000}}
\put(4.0000,8.0000){\line(1,0){4.0000}}
\put(8.0000,12.0000){\line(1,0){4.0000}}
\put(12.0000,8.0000){\line(1,0){4.0000}}
\put(16.0000,8.0000){\line(1,0){4.0000}}
\put(20.0000,8.0000){\line(1,0){4.0000}}
\put(24.0000,8.0000){\line(1,0){4.0000}}
\put(28.0000,8.0000){\line(1,0){4.0000}}
\put(32.0000,8.0000){\line(1,0){4.0000}}
\put(36.0000,8.0000){\line(1,0){4.0000}}
\put(40.0000,8.0000){\line(1,0){4.0000}}
\put(44.0000,8.0000){\line(1,0){4.0000}}
\put(48.0000,8.0000){\line(1,0){4.0000}}
\put(52.0000,12.0000){\line(1,0){4.0000}}
\put(56.0000,8.0000){\line(1,0){4.0000}}
\put(60.0000,12.0000){\line(1,0){4.0000}}
\put(64.0000,8.0000){\line(1,0){4.0000}}
\put(68.0000,12.0000){\line(1,0){4.0000}}
\put(72.0000,8.0000){\line(1,0){4.0000}}
\put(76.0000,8.0000){\line(1,0){4.0000}}
\put(80.0000,8.0000){\line(1,0){4.0000}}
\put(84.0000,8.0000){\line(1,0){4.0000}}
\put(88.0000,8.0000){\line(1,0){4.0000}}
\put(92.0000,8.0000){\line(1,0){4.0000}}
\put(96.0000,8.0000){\line(1,0){4.0000}}
\put(100.0000,8.0000){\line(1,0){4.0000}}
\put(104.0000,8.0000){\line(1,0){4.0000}}
\put(108.0000,8.0000){\line(1,0){4.0000}}
\put(112.0000,8.0000){\line(1,0){4.0000}}
\color{red}
\put(-1.0000,4.0000){\color{red}\normalsize\makebox(0,0)[r]{ack4}}
\put(16.0000,2.0000){\line(0,1){4.0000}}
\put(20.0000,6.0000){\line(0,-1){4.0000}}
\put(24.0000,2.0000){\line(0,1){4.0000}}
\put(28.0000,6.0000){\line(0,-1){4.0000}}
\put(32.0000,2.0000){\line(0,1){4.0000}}
\put(36.0000,6.0000){\line(0,-1){4.0000}}
\put(40.0000,2.0000){\line(0,1){4.0000}}
\put(44.0000,6.0000){\line(0,-1){4.0000}}
\put(100.0000,2.0000){\line(0,1){4.0000}}
\put(104.0000,6.0000){\line(0,-1){4.0000}}
\put(108.0000,2.0000){\line(0,1){4.0000}}
\put(112.0000,6.0000){\line(0,-1){4.0000}}
\put(0.0000,2.0000){\line(1,0){4.0000}}
\put(4.0000,2.0000){\line(1,0){4.0000}}
\put(8.0000,2.0000){\line(1,0){4.0000}}
\put(12.0000,2.0000){\line(1,0){4.0000}}
\put(16.0000,6.0000){\line(1,0){4.0000}}
\put(20.0000,2.0000){\line(1,0){4.0000}}
\put(24.0000,6.0000){\line(1,0){4.0000}}
\put(28.0000,2.0000){\line(1,0){4.0000}}
\put(32.0000,6.0000){\line(1,0){4.0000}}
\put(36.0000,2.0000){\line(1,0){4.0000}}
\put(40.0000,6.0000){\line(1,0){4.0000}}
\put(44.0000,2.0000){\line(1,0){4.0000}}
\put(48.0000,2.0000){\line(1,0){4.0000}}
\put(52.0000,2.0000){\line(1,0){4.0000}}
\put(56.0000,2.0000){\line(1,0){4.0000}}
\put(60.0000,2.0000){\line(1,0){4.0000}}
\put(64.0000,2.0000){\line(1,0){4.0000}}
\put(68.0000,2.0000){\line(1,0){4.0000}}
\put(72.0000,2.0000){\line(1,0){4.0000}}
\put(76.0000,2.0000){\line(1,0){4.0000}}
\put(80.0000,2.0000){\line(1,0){4.0000}}
\put(84.0000,2.0000){\line(1,0){4.0000}}
\put(88.0000,2.0000){\line(1,0){4.0000}}
\put(92.0000,2.0000){\line(1,0){4.0000}}
\put(96.0000,2.0000){\line(1,0){4.0000}}
\put(100.0000,6.0000){\line(1,0){4.0000}}
\put(104.0000,2.0000){\line(1,0){4.0000}}
\put(108.0000,6.0000){\line(1,0){4.0000}}
\put(112.0000,2.0000){\line(1,0){4.0000}}
\end{picture}
}
\caption{Example Simulation of Robust Controller for Never-Give-Up}
\label{fig:neverGiveUpSim}
\end{figure}

The simulation of controller produced for Never-Give-Up Strategy is given is Figure
\ref{fig:neverGiveUpSim}. The assumptions starts violating after step 15,
where more than 2 requests are true simultaneously, but the controller tries to meet
as many requirements as possible.

}

\oomit{
\subsection{Findings}
\label{sec:findings}
{ \color{black} From the case study experiments we have the following observations based on Expected case and Worst case performance measurement.
\begin{itemize}
\item The must dominance based comparison of various \MPNC supervisor shows that supervisor obtained from more robust specification \emph{must dominates} the supervisor synthesized from robust specification ordered below it. So to obtain a robust controller, the most robust criteria  (for which the corresponding robust specification is realizable) should be chosen. 
\item It could also be seen from the expected values comparison of \MPNC supervisors, that choosing the more robust criteria helps in getting the controller which meets the commitments more often (as shown by expected value figures) on a long run than.
\item  It could also be observed that the expected value of H-optimal controller (obtained by determinizing the H-optimal supervisor \GODSC) meeting the commitments drastically improves when compared with the non-optimized controller (obtained by determinizing the \MPNC) irrespective of whether assumptions are meeting.  

\item The expected values of all the controller obtained by determinizing the \GODSC is very high and in most cases equal for all the robustness criterion, although the controller are not identical. This shows that the soft requirement guided synthesis is able to find out the most desirable (meeting the commitment as much as possible) controller irrespective of the robustness criterion used. But each of the these controllers have different hard guarantees associated with it and hence they are not identical. Therefore, the choice of robustness criterion (hard robustness) along with the soft requirements gives the best robust synthesis method.
\end{itemize}

So, as a general rule we can be obtained most robust controller with respect to Must dominance and Expected value, by determinizing the \GODSC supervisor of \emph{realizable} and most robust specification. 

}
}
\subsection{Simulation of Robust Controllers}
The robust controllers synthesized are encoded as Lustre models. We give the simulation (on same inputs) traces of $Arbiter(4,3,2)$ example using Lustre simulator.
It can be seen how each one of them generator different output trace for the same inputs. 
\begin{figure}[!h]
\centering
\includegraphics[width=\textwidth, keepaspectratio]{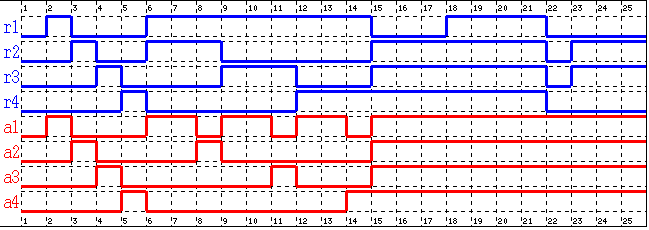}
\caption{Simulation of \MPNC BeCorrect for Arbiter(4,3,2)}
\label{fig:monitor2cella}
\end{figure}

\begin{figure}[!h]
\centering
\includegraphics[width=\textwidth, keepaspectratio]{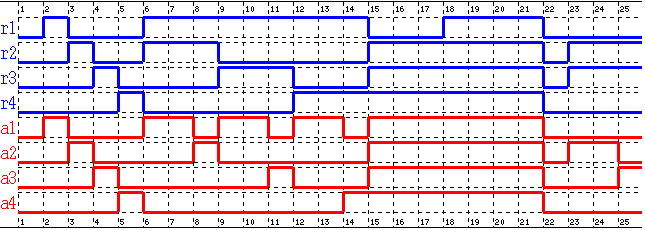}
\caption{Simulation of \MPNC BeCurrentlyCorrect for Arbiter(4,3,2)}
\label{fig:monitor2cellb}
\end{figure}

\begin{figure}[!h]
\centering
\includegraphics[width=\textwidth, keepaspectratio]{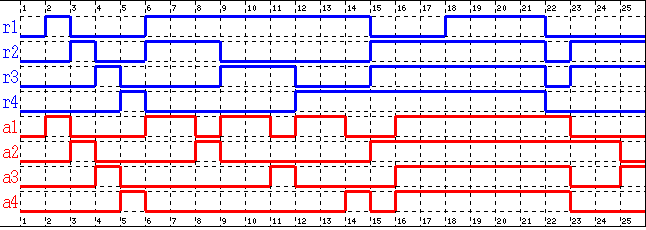}
\caption{Simulation of \MPNC LenCntInt for Arbiter(4,3,2)}
\label{fig:monitor2cellc}
\end{figure}

\begin{figure}[!h]
\centering
\includegraphics[width=\textwidth, keepaspectratio]{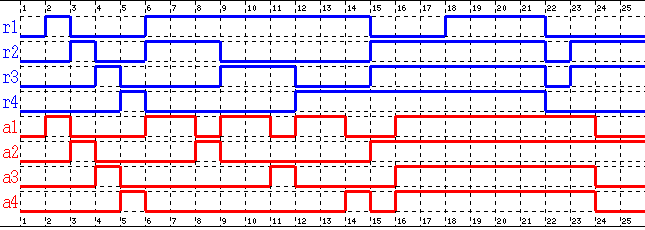}
\caption{Simulation of \MPNC ResBurstInt for Arbiter(4,3,2)}
\label{fig:monitor2celld}
\end{figure}

\oomit{
{
\def\lignefine{\linethickness{0.05pt}}
\def\ligneepaisse{\linethickness{2pt}}
\noindent
\setlength{\unitlength}{0.75mm}
\begin{picture}(170.0000,190.0000)(-19.0000,-5.0000)
\lignefine
\multiput(0.0000,-5.0000)(6.0000,0){26}{\line(0,1){190.0000}}
\put(3.0000,184.0000){\scriptsize\makebox(0,0)[t]{1}}
\put(3.0000,-4.0000){\scriptsize\makebox(0,0)[b]{1}}
\put(9.0000,184.0000){\scriptsize\makebox(0,0)[t]{2}}
\put(9.0000,-4.0000){\scriptsize\makebox(0,0)[b]{2}}
\put(15.0000,184.0000){\scriptsize\makebox(0,0)[t]{3}}
\put(15.0000,-4.0000){\scriptsize\makebox(0,0)[b]{3}}
\put(21.0000,184.0000){\scriptsize\makebox(0,0)[t]{4}}
\put(21.0000,-4.0000){\scriptsize\makebox(0,0)[b]{4}}
\put(27.0000,184.0000){\scriptsize\makebox(0,0)[t]{5}}
\put(27.0000,-4.0000){\scriptsize\makebox(0,0)[b]{5}}
\put(33.0000,184.0000){\scriptsize\makebox(0,0)[t]{6}}
\put(33.0000,-4.0000){\scriptsize\makebox(0,0)[b]{6}}
\put(39.0000,184.0000){\scriptsize\makebox(0,0)[t]{7}}
\put(39.0000,-4.0000){\scriptsize\makebox(0,0)[b]{7}}
\put(45.0000,184.0000){\scriptsize\makebox(0,0)[t]{8}}
\put(45.0000,-4.0000){\scriptsize\makebox(0,0)[b]{8}}
\put(51.0000,184.0000){\scriptsize\makebox(0,0)[t]{9}}
\put(51.0000,-4.0000){\scriptsize\makebox(0,0)[b]{9}}
\put(57.0000,184.0000){\scriptsize\makebox(0,0)[t]{10}}
\put(57.0000,-4.0000){\scriptsize\makebox(0,0)[b]{10}}
\put(63.0000,184.0000){\scriptsize\makebox(0,0)[t]{11}}
\put(63.0000,-4.0000){\scriptsize\makebox(0,0)[b]{11}}
\put(69.0000,184.0000){\scriptsize\makebox(0,0)[t]{12}}
\put(69.0000,-4.0000){\scriptsize\makebox(0,0)[b]{12}}
\put(75.0000,184.0000){\scriptsize\makebox(0,0)[t]{13}}
\put(75.0000,-4.0000){\scriptsize\makebox(0,0)[b]{13}}
\put(81.0000,184.0000){\scriptsize\makebox(0,0)[t]{14}}
\put(81.0000,-4.0000){\scriptsize\makebox(0,0)[b]{14}}
\put(87.0000,184.0000){\scriptsize\makebox(0,0)[t]{15}}
\put(87.0000,-4.0000){\scriptsize\makebox(0,0)[b]{15}}
\put(93.0000,184.0000){\scriptsize\makebox(0,0)[t]{16}}
\put(93.0000,-4.0000){\scriptsize\makebox(0,0)[b]{16}}
\put(99.0000,184.0000){\scriptsize\makebox(0,0)[t]{17}}
\put(99.0000,-4.0000){\scriptsize\makebox(0,0)[b]{17}}
\put(105.0000,184.0000){\scriptsize\makebox(0,0)[t]{18}}
\put(105.0000,-4.0000){\scriptsize\makebox(0,0)[b]{18}}
\put(111.0000,184.0000){\scriptsize\makebox(0,0)[t]{19}}
\put(111.0000,-4.0000){\scriptsize\makebox(0,0)[b]{19}}
\put(117.0000,184.0000){\scriptsize\makebox(0,0)[t]{20}}
\put(117.0000,-4.0000){\scriptsize\makebox(0,0)[b]{20}}
\put(123.0000,184.0000){\scriptsize\makebox(0,0)[t]{21}}
\put(123.0000,-4.0000){\scriptsize\makebox(0,0)[b]{21}}
\put(129.0000,184.0000){\scriptsize\makebox(0,0)[t]{22}}
\put(129.0000,-4.0000){\scriptsize\makebox(0,0)[b]{22}}
\put(135.0000,184.0000){\scriptsize\makebox(0,0)[t]{23}}
\put(135.0000,-4.0000){\scriptsize\makebox(0,0)[b]{23}}
\put(141.0000,184.0000){\scriptsize\makebox(0,0)[t]{24}}
\put(141.0000,-4.0000){\scriptsize\makebox(0,0)[b]{24}}
\put(147.0000,184.0000){\scriptsize\makebox(0,0)[t]{25}}
\put(147.0000,-4.0000){\scriptsize\makebox(0,0)[b]{25}}
\put(-1.0000,157.7500){\line(1,0){152.0000}}
\put(-1.0000,178.0000){\line(1,0){152.0000}}
\put(-1.0000,135.5000){\line(1,0){152.0000}}
\put(-1.0000,155.7500){\line(1,0){152.0000}}
\put(-1.0000,113.2500){\line(1,0){152.0000}}
\put(-1.0000,133.5000){\line(1,0){152.0000}}
\put(-1.0000,91.0000){\line(1,0){152.0000}}
\put(-1.0000,111.2500){\line(1,0){152.0000}}
\put(-1.0000,68.7500){\line(1,0){152.0000}}
\put(-1.0000,89.0000){\line(1,0){152.0000}}
\put(-1.0000,46.5000){\line(1,0){152.0000}}
\put(-1.0000,66.7500){\line(1,0){152.0000}}
\put(-1.0000,24.2500){\line(1,0){152.0000}}
\put(-1.0000,44.5000){\line(1,0){152.0000}}
\put(-1.0000,2.0000){\line(1,0){152.0000}}
\put(-1.0000,22.2500){\line(1,0){152.0000}}
\ligneepaisse
\put(-1.0000,167.8750){\normalsize\makebox(0,0)[r]{r1}}
\put(6.0000,157.7500){\line(0,1){20.2500}}
\put(12.0000,178.0000){\line(0,-1){20.2500}}
\put(30.0000,157.7500){\line(0,1){20.2500}}
\put(84.0000,178.0000){\line(0,-1){20.2500}}
\put(102.0000,157.7500){\line(0,1){20.2500}}
\put(126.0000,178.0000){\line(0,-1){20.2500}}
\put(0.0000,157.7500){\line(1,0){6.0000}}
\put(6.0000,178.0000){\line(1,0){6.0000}}
\put(12.0000,157.7500){\line(1,0){6.0000}}
\put(18.0000,157.7500){\line(1,0){6.0000}}
\put(24.0000,157.7500){\line(1,0){6.0000}}
\put(30.0000,178.0000){\line(1,0){6.0000}}
\put(36.0000,178.0000){\line(1,0){6.0000}}
\put(42.0000,178.0000){\line(1,0){6.0000}}
\put(48.0000,178.0000){\line(1,0){6.0000}}
\put(54.0000,178.0000){\line(1,0){6.0000}}
\put(60.0000,178.0000){\line(1,0){6.0000}}
\put(66.0000,178.0000){\line(1,0){6.0000}}
\put(72.0000,178.0000){\line(1,0){6.0000}}
\put(78.0000,178.0000){\line(1,0){6.0000}}
\put(84.0000,157.7500){\line(1,0){6.0000}}
\put(90.0000,157.7500){\line(1,0){6.0000}}
\put(96.0000,157.7500){\line(1,0){6.0000}}
\put(102.0000,178.0000){\line(1,0){6.0000}}
\put(108.0000,178.0000){\line(1,0){6.0000}}
\put(114.0000,178.0000){\line(1,0){6.0000}}
\put(120.0000,178.0000){\line(1,0){6.0000}}
\put(126.0000,157.7500){\line(1,0){6.0000}}
\put(132.0000,157.7500){\line(1,0){6.0000}}
\put(138.0000,157.7500){\line(1,0){6.0000}}
\put(144.0000,157.7500){\line(1,0){6.0000}}
\put(-1.0000,145.6250){\normalsize\makebox(0,0)[r]{r2}}
\put(12.0000,135.5000){\line(0,1){20.2500}}
\put(18.0000,155.7500){\line(0,-1){20.2500}}
\put(30.0000,135.5000){\line(0,1){20.2500}}
\put(48.0000,155.7500){\line(0,-1){20.2500}}
\put(84.0000,135.5000){\line(0,1){20.2500}}
\put(126.0000,155.7500){\line(0,-1){20.2500}}
\put(132.0000,135.5000){\line(0,1){20.2500}}
\put(0.0000,135.5000){\line(1,0){6.0000}}
\put(6.0000,135.5000){\line(1,0){6.0000}}
\put(12.0000,155.7500){\line(1,0){6.0000}}
\put(18.0000,135.5000){\line(1,0){6.0000}}
\put(24.0000,135.5000){\line(1,0){6.0000}}
\put(30.0000,155.7500){\line(1,0){6.0000}}
\put(36.0000,155.7500){\line(1,0){6.0000}}
\put(42.0000,155.7500){\line(1,0){6.0000}}
\put(48.0000,135.5000){\line(1,0){6.0000}}
\put(54.0000,135.5000){\line(1,0){6.0000}}
\put(60.0000,135.5000){\line(1,0){6.0000}}
\put(66.0000,135.5000){\line(1,0){6.0000}}
\put(72.0000,135.5000){\line(1,0){6.0000}}
\put(78.0000,135.5000){\line(1,0){6.0000}}
\put(84.0000,155.7500){\line(1,0){6.0000}}
\put(90.0000,155.7500){\line(1,0){6.0000}}
\put(96.0000,155.7500){\line(1,0){6.0000}}
\put(102.0000,155.7500){\line(1,0){6.0000}}
\put(108.0000,155.7500){\line(1,0){6.0000}}
\put(114.0000,155.7500){\line(1,0){6.0000}}
\put(120.0000,155.7500){\line(1,0){6.0000}}
\put(126.0000,135.5000){\line(1,0){6.0000}}
\put(132.0000,155.7500){\line(1,0){6.0000}}
\put(138.0000,155.7500){\line(1,0){6.0000}}
\put(144.0000,155.7500){\line(1,0){6.0000}}
\put(-1.0000,123.3750){\normalsize\makebox(0,0)[r]{r3}}
\put(18.0000,113.2500){\line(0,1){20.2500}}
\put(24.0000,133.5000){\line(0,-1){20.2500}}
\put(48.0000,113.2500){\line(0,1){20.2500}}
\put(66.0000,133.5000){\line(0,-1){20.2500}}
\put(84.0000,113.2500){\line(0,1){20.2500}}
\put(126.0000,133.5000){\line(0,-1){20.2500}}
\put(132.0000,113.2500){\line(0,1){20.2500}}
\put(0.0000,113.2500){\line(1,0){6.0000}}
\put(6.0000,113.2500){\line(1,0){6.0000}}
\put(12.0000,113.2500){\line(1,0){6.0000}}
\put(18.0000,133.5000){\line(1,0){6.0000}}
\put(24.0000,113.2500){\line(1,0){6.0000}}
\put(30.0000,113.2500){\line(1,0){6.0000}}
\put(36.0000,113.2500){\line(1,0){6.0000}}
\put(42.0000,113.2500){\line(1,0){6.0000}}
\put(48.0000,133.5000){\line(1,0){6.0000}}
\put(54.0000,133.5000){\line(1,0){6.0000}}
\put(60.0000,133.5000){\line(1,0){6.0000}}
\put(66.0000,113.2500){\line(1,0){6.0000}}
\put(72.0000,113.2500){\line(1,0){6.0000}}
\put(78.0000,113.2500){\line(1,0){6.0000}}
\put(84.0000,133.5000){\line(1,0){6.0000}}
\put(90.0000,133.5000){\line(1,0){6.0000}}
\put(96.0000,133.5000){\line(1,0){6.0000}}
\put(102.0000,133.5000){\line(1,0){6.0000}}
\put(108.0000,133.5000){\line(1,0){6.0000}}
\put(114.0000,133.5000){\line(1,0){6.0000}}
\put(120.0000,133.5000){\line(1,0){6.0000}}
\put(126.0000,113.2500){\line(1,0){6.0000}}
\put(132.0000,133.5000){\line(1,0){6.0000}}
\put(138.0000,133.5000){\line(1,0){6.0000}}
\put(144.0000,133.5000){\line(1,0){6.0000}}
\put(-1.0000,101.1250){\normalsize\makebox(0,0)[r]{r4}}
\put(24.0000,91.0000){\line(0,1){20.2500}}
\put(30.0000,111.2500){\line(0,-1){20.2500}}
\put(66.0000,91.0000){\line(0,1){20.2500}}
\put(126.0000,111.2500){\line(0,-1){20.2500}}
\put(0.0000,91.0000){\line(1,0){6.0000}}
\put(6.0000,91.0000){\line(1,0){6.0000}}
\put(12.0000,91.0000){\line(1,0){6.0000}}
\put(18.0000,91.0000){\line(1,0){6.0000}}
\put(24.0000,111.2500){\line(1,0){6.0000}}
\put(30.0000,91.0000){\line(1,0){6.0000}}
\put(36.0000,91.0000){\line(1,0){6.0000}}
\put(42.0000,91.0000){\line(1,0){6.0000}}
\put(48.0000,91.0000){\line(1,0){6.0000}}
\put(54.0000,91.0000){\line(1,0){6.0000}}
\put(60.0000,91.0000){\line(1,0){6.0000}}
\put(66.0000,111.2500){\line(1,0){6.0000}}
\put(72.0000,111.2500){\line(1,0){6.0000}}
\put(78.0000,111.2500){\line(1,0){6.0000}}
\put(84.0000,111.2500){\line(1,0){6.0000}}
\put(90.0000,111.2500){\line(1,0){6.0000}}
\put(96.0000,111.2500){\line(1,0){6.0000}}
\put(102.0000,111.2500){\line(1,0){6.0000}}
\put(108.0000,111.2500){\line(1,0){6.0000}}
\put(114.0000,111.2500){\line(1,0){6.0000}}
\put(120.0000,111.2500){\line(1,0){6.0000}}
\put(126.0000,91.0000){\line(1,0){6.0000}}
\put(132.0000,91.0000){\line(1,0){6.0000}}
\put(138.0000,91.0000){\line(1,0){6.0000}}
\put(144.0000,91.0000){\line(1,0){6.0000}}
\put(-1.0000,78.8750){\normalsize\makebox(0,0)[r]{a1}}
\put(6.0000,68.7500){\line(0,1){20.2500}}
\put(12.0000,89.0000){\line(0,-1){20.2500}}
\put(30.0000,68.7500){\line(0,1){20.2500}}
\put(42.0000,89.0000){\line(0,-1){20.2500}}
\put(48.0000,68.7500){\line(0,1){20.2500}}
\put(60.0000,89.0000){\line(0,-1){20.2500}}
\put(66.0000,68.7500){\line(0,1){20.2500}}
\put(78.0000,89.0000){\line(0,-1){20.2500}}
\put(108.0000,68.7500){\line(0,1){20.2500}}
\put(126.0000,89.0000){\line(0,-1){20.2500}}
\put(0.0000,68.7500){\line(1,0){6.0000}}
\put(6.0000,89.0000){\line(1,0){6.0000}}
\put(12.0000,68.7500){\line(1,0){6.0000}}
\put(18.0000,68.7500){\line(1,0){6.0000}}
\put(24.0000,68.7500){\line(1,0){6.0000}}
\put(30.0000,89.0000){\line(1,0){6.0000}}
\put(36.0000,89.0000){\line(1,0){6.0000}}
\put(42.0000,68.7500){\line(1,0){6.0000}}
\put(48.0000,89.0000){\line(1,0){6.0000}}
\put(54.0000,89.0000){\line(1,0){6.0000}}
\put(60.0000,68.7500){\line(1,0){6.0000}}
\put(66.0000,89.0000){\line(1,0){6.0000}}
\put(72.0000,89.0000){\line(1,0){6.0000}}
\put(78.0000,68.7500){\line(1,0){6.0000}}
\put(84.0000,68.7500){\line(1,0){6.0000}}
\put(90.0000,68.7500){\line(1,0){6.0000}}
\put(96.0000,68.7500){\line(1,0){6.0000}}
\put(102.0000,68.7500){\line(1,0){6.0000}}
\put(108.0000,89.0000){\line(1,0){6.0000}}
\put(114.0000,89.0000){\line(1,0){6.0000}}
\put(120.0000,89.0000){\line(1,0){6.0000}}
\put(126.0000,68.7500){\line(1,0){6.0000}}
\put(132.0000,68.7500){\line(1,0){6.0000}}
\put(138.0000,68.7500){\line(1,0){6.0000}}
\put(144.0000,68.7500){\line(1,0){6.0000}}
\put(-1.0000,56.6250){\normalsize\makebox(0,0)[r]{a2}}
\put(12.0000,46.5000){\line(0,1){20.2500}}
\put(18.0000,66.7500){\line(0,-1){20.2500}}
\put(42.0000,46.5000){\line(0,1){20.2500}}
\put(48.0000,66.7500){\line(0,-1){20.2500}}
\put(84.0000,46.5000){\line(0,1){20.2500}}
\put(90.0000,66.7500){\line(0,-1){20.2500}}
\put(102.0000,46.5000){\line(0,1){20.2500}}
\put(114.0000,66.7500){\line(0,-1){20.2500}}
\put(120.0000,46.5000){\line(0,1){20.2500}}
\put(126.0000,66.7500){\line(0,-1){20.2500}}
\put(132.0000,46.5000){\line(0,1){20.2500}}
\put(144.0000,66.7500){\line(0,-1){20.2500}}
\put(0.0000,46.5000){\line(1,0){6.0000}}
\put(6.0000,46.5000){\line(1,0){6.0000}}
\put(12.0000,66.7500){\line(1,0){6.0000}}
\put(18.0000,46.5000){\line(1,0){6.0000}}
\put(24.0000,46.5000){\line(1,0){6.0000}}
\put(30.0000,46.5000){\line(1,0){6.0000}}
\put(36.0000,46.5000){\line(1,0){6.0000}}
\put(42.0000,66.7500){\line(1,0){6.0000}}
\put(48.0000,46.5000){\line(1,0){6.0000}}
\put(54.0000,46.5000){\line(1,0){6.0000}}
\put(60.0000,46.5000){\line(1,0){6.0000}}
\put(66.0000,46.5000){\line(1,0){6.0000}}
\put(72.0000,46.5000){\line(1,0){6.0000}}
\put(78.0000,46.5000){\line(1,0){6.0000}}
\put(84.0000,66.7500){\line(1,0){6.0000}}
\put(90.0000,46.5000){\line(1,0){6.0000}}
\put(96.0000,46.5000){\line(1,0){6.0000}}
\put(102.0000,66.7500){\line(1,0){6.0000}}
\put(108.0000,66.7500){\line(1,0){6.0000}}
\put(114.0000,46.5000){\line(1,0){6.0000}}
\put(120.0000,66.7500){\line(1,0){6.0000}}
\put(126.0000,46.5000){\line(1,0){6.0000}}
\put(132.0000,66.7500){\line(1,0){6.0000}}
\put(138.0000,66.7500){\line(1,0){6.0000}}
\put(144.0000,46.5000){\line(1,0){6.0000}}
\put(-1.0000,34.3750){\normalsize\makebox(0,0)[r]{a3}}
\put(18.0000,24.2500){\line(0,1){20.2500}}
\put(24.0000,44.5000){\line(0,-1){20.2500}}
\put(60.0000,24.2500){\line(0,1){20.2500}}
\put(66.0000,44.5000){\line(0,-1){20.2500}}
\put(90.0000,24.2500){\line(0,1){20.2500}}
\put(96.0000,44.5000){\line(0,-1){20.2500}}
\put(108.0000,24.2500){\line(0,1){20.2500}}
\put(114.0000,44.5000){\line(0,-1){20.2500}}
\put(120.0000,24.2500){\line(0,1){20.2500}}
\put(126.0000,44.5000){\line(0,-1){20.2500}}
\put(144.0000,24.2500){\line(0,1){20.2500}}
\put(0.0000,24.2500){\line(1,0){6.0000}}
\put(6.0000,24.2500){\line(1,0){6.0000}}
\put(12.0000,24.2500){\line(1,0){6.0000}}
\put(18.0000,44.5000){\line(1,0){6.0000}}
\put(24.0000,24.2500){\line(1,0){6.0000}}
\put(30.0000,24.2500){\line(1,0){6.0000}}
\put(36.0000,24.2500){\line(1,0){6.0000}}
\put(42.0000,24.2500){\line(1,0){6.0000}}
\put(48.0000,24.2500){\line(1,0){6.0000}}
\put(54.0000,24.2500){\line(1,0){6.0000}}
\put(60.0000,44.5000){\line(1,0){6.0000}}
\put(66.0000,24.2500){\line(1,0){6.0000}}
\put(72.0000,24.2500){\line(1,0){6.0000}}
\put(78.0000,24.2500){\line(1,0){6.0000}}
\put(84.0000,24.2500){\line(1,0){6.0000}}
\put(90.0000,44.5000){\line(1,0){6.0000}}
\put(96.0000,24.2500){\line(1,0){6.0000}}
\put(102.0000,24.2500){\line(1,0){6.0000}}
\put(108.0000,44.5000){\line(1,0){6.0000}}
\put(114.0000,24.2500){\line(1,0){6.0000}}
\put(120.0000,44.5000){\line(1,0){6.0000}}
\put(126.0000,24.2500){\line(1,0){6.0000}}
\put(132.0000,24.2500){\line(1,0){6.0000}}
\put(138.0000,24.2500){\line(1,0){6.0000}}
\put(144.0000,44.5000){\line(1,0){6.0000}}
\put(-1.0000,12.1250){\normalsize\makebox(0,0)[r]{a4}}
\put(24.0000,2.0000){\line(0,1){20.2500}}
\put(30.0000,22.2500){\line(0,-1){20.2500}}
\put(78.0000,2.0000){\line(0,1){20.2500}}
\put(84.0000,22.2500){\line(0,-1){20.2500}}
\put(96.0000,2.0000){\line(0,1){20.2500}}
\put(102.0000,22.2500){\line(0,-1){20.2500}}
\put(108.0000,2.0000){\line(0,1){20.2500}}
\put(114.0000,22.2500){\line(0,-1){20.2500}}
\put(120.0000,2.0000){\line(0,1){20.2500}}
\put(126.0000,22.2500){\line(0,-1){20.2500}}
\put(0.0000,2.0000){\line(1,0){6.0000}}
\put(6.0000,2.0000){\line(1,0){6.0000}}
\put(12.0000,2.0000){\line(1,0){6.0000}}
\put(18.0000,2.0000){\line(1,0){6.0000}}
\put(24.0000,22.2500){\line(1,0){6.0000}}
\put(30.0000,2.0000){\line(1,0){6.0000}}
\put(36.0000,2.0000){\line(1,0){6.0000}}
\put(42.0000,2.0000){\line(1,0){6.0000}}
\put(48.0000,2.0000){\line(1,0){6.0000}}
\put(54.0000,2.0000){\line(1,0){6.0000}}
\put(60.0000,2.0000){\line(1,0){6.0000}}
\put(66.0000,2.0000){\line(1,0){6.0000}}
\put(72.0000,2.0000){\line(1,0){6.0000}}
\put(78.0000,22.2500){\line(1,0){6.0000}}
\put(84.0000,2.0000){\line(1,0){6.0000}}
\put(90.0000,2.0000){\line(1,0){6.0000}}
\put(96.0000,22.2500){\line(1,0){6.0000}}
\put(102.0000,2.0000){\line(1,0){6.0000}}
\put(108.0000,22.2500){\line(1,0){6.0000}}
\put(114.0000,2.0000){\line(1,0){6.0000}}
\put(120.0000,22.2500){\line(1,0){6.0000}}
\put(126.0000,2.0000){\line(1,0){6.0000}}
\put(132.0000,2.0000){\line(1,0){6.0000}}
\put(138.0000,2.0000){\line(1,0){6.0000}}
\put(144.0000,2.0000){\line(1,0){6.0000}}
\label{fig:simDetGODSCBeCurrentlyCorrect}
\end{picture}
}

{
\def\lignefine{\linethickness{0.05pt}}
\def\ligneepaisse{\linethickness{2pt}}
\noindent
\setlength{\unitlength}{0.75mm}
\begin{picture}(170.0000,190.0000)(-19.0000,-5.0000)
\lignefine
\multiput(0.0000,-5.0000)(6.0000,0){26}{\line(0,1){190.0000}}
\put(3.0000,184.0000){\scriptsize\makebox(0,0)[t]{1}}
\put(3.0000,-4.0000){\scriptsize\makebox(0,0)[b]{1}}
\put(9.0000,184.0000){\scriptsize\makebox(0,0)[t]{2}}
\put(9.0000,-4.0000){\scriptsize\makebox(0,0)[b]{2}}
\put(15.0000,184.0000){\scriptsize\makebox(0,0)[t]{3}}
\put(15.0000,-4.0000){\scriptsize\makebox(0,0)[b]{3}}
\put(21.0000,184.0000){\scriptsize\makebox(0,0)[t]{4}}
\put(21.0000,-4.0000){\scriptsize\makebox(0,0)[b]{4}}
\put(27.0000,184.0000){\scriptsize\makebox(0,0)[t]{5}}
\put(27.0000,-4.0000){\scriptsize\makebox(0,0)[b]{5}}
\put(33.0000,184.0000){\scriptsize\makebox(0,0)[t]{6}}
\put(33.0000,-4.0000){\scriptsize\makebox(0,0)[b]{6}}
\put(39.0000,184.0000){\scriptsize\makebox(0,0)[t]{7}}
\put(39.0000,-4.0000){\scriptsize\makebox(0,0)[b]{7}}
\put(45.0000,184.0000){\scriptsize\makebox(0,0)[t]{8}}
\put(45.0000,-4.0000){\scriptsize\makebox(0,0)[b]{8}}
\put(51.0000,184.0000){\scriptsize\makebox(0,0)[t]{9}}
\put(51.0000,-4.0000){\scriptsize\makebox(0,0)[b]{9}}
\put(57.0000,184.0000){\scriptsize\makebox(0,0)[t]{10}}
\put(57.0000,-4.0000){\scriptsize\makebox(0,0)[b]{10}}
\put(63.0000,184.0000){\scriptsize\makebox(0,0)[t]{11}}
\put(63.0000,-4.0000){\scriptsize\makebox(0,0)[b]{11}}
\put(69.0000,184.0000){\scriptsize\makebox(0,0)[t]{12}}
\put(69.0000,-4.0000){\scriptsize\makebox(0,0)[b]{12}}
\put(75.0000,184.0000){\scriptsize\makebox(0,0)[t]{13}}
\put(75.0000,-4.0000){\scriptsize\makebox(0,0)[b]{13}}
\put(81.0000,184.0000){\scriptsize\makebox(0,0)[t]{14}}
\put(81.0000,-4.0000){\scriptsize\makebox(0,0)[b]{14}}
\put(87.0000,184.0000){\scriptsize\makebox(0,0)[t]{15}}
\put(87.0000,-4.0000){\scriptsize\makebox(0,0)[b]{15}}
\put(93.0000,184.0000){\scriptsize\makebox(0,0)[t]{16}}
\put(93.0000,-4.0000){\scriptsize\makebox(0,0)[b]{16}}
\put(99.0000,184.0000){\scriptsize\makebox(0,0)[t]{17}}
\put(99.0000,-4.0000){\scriptsize\makebox(0,0)[b]{17}}
\put(105.0000,184.0000){\scriptsize\makebox(0,0)[t]{18}}
\put(105.0000,-4.0000){\scriptsize\makebox(0,0)[b]{18}}
\put(111.0000,184.0000){\scriptsize\makebox(0,0)[t]{19}}
\put(111.0000,-4.0000){\scriptsize\makebox(0,0)[b]{19}}
\put(117.0000,184.0000){\scriptsize\makebox(0,0)[t]{20}}
\put(117.0000,-4.0000){\scriptsize\makebox(0,0)[b]{20}}
\put(123.0000,184.0000){\scriptsize\makebox(0,0)[t]{21}}
\put(123.0000,-4.0000){\scriptsize\makebox(0,0)[b]{21}}
\put(129.0000,184.0000){\scriptsize\makebox(0,0)[t]{22}}
\put(129.0000,-4.0000){\scriptsize\makebox(0,0)[b]{22}}
\put(135.0000,184.0000){\scriptsize\makebox(0,0)[t]{23}}
\put(135.0000,-4.0000){\scriptsize\makebox(0,0)[b]{23}}
\put(141.0000,184.0000){\scriptsize\makebox(0,0)[t]{24}}
\put(141.0000,-4.0000){\scriptsize\makebox(0,0)[b]{24}}
\put(147.0000,184.0000){\scriptsize\makebox(0,0)[t]{25}}
\put(147.0000,-4.0000){\scriptsize\makebox(0,0)[b]{25}}
\put(-1.0000,157.7500){\line(1,0){152.0000}}
\put(-1.0000,178.0000){\line(1,0){152.0000}}
\put(-1.0000,135.5000){\line(1,0){152.0000}}
\put(-1.0000,155.7500){\line(1,0){152.0000}}
\put(-1.0000,113.2500){\line(1,0){152.0000}}
\put(-1.0000,133.5000){\line(1,0){152.0000}}
\put(-1.0000,91.0000){\line(1,0){152.0000}}
\put(-1.0000,111.2500){\line(1,0){152.0000}}
\put(-1.0000,68.7500){\line(1,0){152.0000}}
\put(-1.0000,89.0000){\line(1,0){152.0000}}
\put(-1.0000,46.5000){\line(1,0){152.0000}}
\put(-1.0000,66.7500){\line(1,0){152.0000}}
\put(-1.0000,24.2500){\line(1,0){152.0000}}
\put(-1.0000,44.5000){\line(1,0){152.0000}}
\put(-1.0000,2.0000){\line(1,0){152.0000}}
\put(-1.0000,22.2500){\line(1,0){152.0000}}
\ligneepaisse
\put(-1.0000,167.8750){\normalsize\makebox(0,0)[r]{r1}}
\put(6.0000,157.7500){\line(0,1){20.2500}}
\put(12.0000,178.0000){\line(0,-1){20.2500}}
\put(30.0000,157.7500){\line(0,1){20.2500}}
\put(84.0000,178.0000){\line(0,-1){20.2500}}
\put(102.0000,157.7500){\line(0,1){20.2500}}
\put(126.0000,178.0000){\line(0,-1){20.2500}}
\put(0.0000,157.7500){\line(1,0){6.0000}}
\put(6.0000,178.0000){\line(1,0){6.0000}}
\put(12.0000,157.7500){\line(1,0){6.0000}}
\put(18.0000,157.7500){\line(1,0){6.0000}}
\put(24.0000,157.7500){\line(1,0){6.0000}}
\put(30.0000,178.0000){\line(1,0){6.0000}}
\put(36.0000,178.0000){\line(1,0){6.0000}}
\put(42.0000,178.0000){\line(1,0){6.0000}}
\put(48.0000,178.0000){\line(1,0){6.0000}}
\put(54.0000,178.0000){\line(1,0){6.0000}}
\put(60.0000,178.0000){\line(1,0){6.0000}}
\put(66.0000,178.0000){\line(1,0){6.0000}}
\put(72.0000,178.0000){\line(1,0){6.0000}}
\put(78.0000,178.0000){\line(1,0){6.0000}}
\put(84.0000,157.7500){\line(1,0){6.0000}}
\put(90.0000,157.7500){\line(1,0){6.0000}}
\put(96.0000,157.7500){\line(1,0){6.0000}}
\put(102.0000,178.0000){\line(1,0){6.0000}}
\put(108.0000,178.0000){\line(1,0){6.0000}}
\put(114.0000,178.0000){\line(1,0){6.0000}}
\put(120.0000,178.0000){\line(1,0){6.0000}}
\put(126.0000,157.7500){\line(1,0){6.0000}}
\put(132.0000,157.7500){\line(1,0){6.0000}}
\put(138.0000,157.7500){\line(1,0){6.0000}}
\put(144.0000,157.7500){\line(1,0){6.0000}}
\put(-1.0000,145.6250){\normalsize\makebox(0,0)[r]{r2}}
\put(12.0000,135.5000){\line(0,1){20.2500}}
\put(18.0000,155.7500){\line(0,-1){20.2500}}
\put(30.0000,135.5000){\line(0,1){20.2500}}
\put(48.0000,155.7500){\line(0,-1){20.2500}}
\put(84.0000,135.5000){\line(0,1){20.2500}}
\put(126.0000,155.7500){\line(0,-1){20.2500}}
\put(132.0000,135.5000){\line(0,1){20.2500}}
\put(0.0000,135.5000){\line(1,0){6.0000}}
\put(6.0000,135.5000){\line(1,0){6.0000}}
\put(12.0000,155.7500){\line(1,0){6.0000}}
\put(18.0000,135.5000){\line(1,0){6.0000}}
\put(24.0000,135.5000){\line(1,0){6.0000}}
\put(30.0000,155.7500){\line(1,0){6.0000}}
\put(36.0000,155.7500){\line(1,0){6.0000}}
\put(42.0000,155.7500){\line(1,0){6.0000}}
\put(48.0000,135.5000){\line(1,0){6.0000}}
\put(54.0000,135.5000){\line(1,0){6.0000}}
\put(60.0000,135.5000){\line(1,0){6.0000}}
\put(66.0000,135.5000){\line(1,0){6.0000}}
\put(72.0000,135.5000){\line(1,0){6.0000}}
\put(78.0000,135.5000){\line(1,0){6.0000}}
\put(84.0000,155.7500){\line(1,0){6.0000}}
\put(90.0000,155.7500){\line(1,0){6.0000}}
\put(96.0000,155.7500){\line(1,0){6.0000}}
\put(102.0000,155.7500){\line(1,0){6.0000}}
\put(108.0000,155.7500){\line(1,0){6.0000}}
\put(114.0000,155.7500){\line(1,0){6.0000}}
\put(120.0000,155.7500){\line(1,0){6.0000}}
\put(126.0000,135.5000){\line(1,0){6.0000}}
\put(132.0000,155.7500){\line(1,0){6.0000}}
\put(138.0000,155.7500){\line(1,0){6.0000}}
\put(144.0000,155.7500){\line(1,0){6.0000}}
\put(-1.0000,123.3750){\normalsize\makebox(0,0)[r]{r3}}
\put(18.0000,113.2500){\line(0,1){20.2500}}
\put(24.0000,133.5000){\line(0,-1){20.2500}}
\put(48.0000,113.2500){\line(0,1){20.2500}}
\put(66.0000,133.5000){\line(0,-1){20.2500}}
\put(84.0000,113.2500){\line(0,1){20.2500}}
\put(126.0000,133.5000){\line(0,-1){20.2500}}
\put(132.0000,113.2500){\line(0,1){20.2500}}
\put(0.0000,113.2500){\line(1,0){6.0000}}
\put(6.0000,113.2500){\line(1,0){6.0000}}
\put(12.0000,113.2500){\line(1,0){6.0000}}
\put(18.0000,133.5000){\line(1,0){6.0000}}
\put(24.0000,113.2500){\line(1,0){6.0000}}
\put(30.0000,113.2500){\line(1,0){6.0000}}
\put(36.0000,113.2500){\line(1,0){6.0000}}
\put(42.0000,113.2500){\line(1,0){6.0000}}
\put(48.0000,133.5000){\line(1,0){6.0000}}
\put(54.0000,133.5000){\line(1,0){6.0000}}
\put(60.0000,133.5000){\line(1,0){6.0000}}
\put(66.0000,113.2500){\line(1,0){6.0000}}
\put(72.0000,113.2500){\line(1,0){6.0000}}
\put(78.0000,113.2500){\line(1,0){6.0000}}
\put(84.0000,133.5000){\line(1,0){6.0000}}
\put(90.0000,133.5000){\line(1,0){6.0000}}
\put(96.0000,133.5000){\line(1,0){6.0000}}
\put(102.0000,133.5000){\line(1,0){6.0000}}
\put(108.0000,133.5000){\line(1,0){6.0000}}
\put(114.0000,133.5000){\line(1,0){6.0000}}
\put(120.0000,133.5000){\line(1,0){6.0000}}
\put(126.0000,113.2500){\line(1,0){6.0000}}
\put(132.0000,133.5000){\line(1,0){6.0000}}
\put(138.0000,133.5000){\line(1,0){6.0000}}
\put(144.0000,133.5000){\line(1,0){6.0000}}
\put(-1.0000,101.1250){\normalsize\makebox(0,0)[r]{r4}}
\put(24.0000,91.0000){\line(0,1){20.2500}}
\put(30.0000,111.2500){\line(0,-1){20.2500}}
\put(66.0000,91.0000){\line(0,1){20.2500}}
\put(126.0000,111.2500){\line(0,-1){20.2500}}
\put(0.0000,91.0000){\line(1,0){6.0000}}
\put(6.0000,91.0000){\line(1,0){6.0000}}
\put(12.0000,91.0000){\line(1,0){6.0000}}
\put(18.0000,91.0000){\line(1,0){6.0000}}
\put(24.0000,111.2500){\line(1,0){6.0000}}
\put(30.0000,91.0000){\line(1,0){6.0000}}
\put(36.0000,91.0000){\line(1,0){6.0000}}
\put(42.0000,91.0000){\line(1,0){6.0000}}
\put(48.0000,91.0000){\line(1,0){6.0000}}
\put(54.0000,91.0000){\line(1,0){6.0000}}
\put(60.0000,91.0000){\line(1,0){6.0000}}
\put(66.0000,111.2500){\line(1,0){6.0000}}
\put(72.0000,111.2500){\line(1,0){6.0000}}
\put(78.0000,111.2500){\line(1,0){6.0000}}
\put(84.0000,111.2500){\line(1,0){6.0000}}
\put(90.0000,111.2500){\line(1,0){6.0000}}
\put(96.0000,111.2500){\line(1,0){6.0000}}
\put(102.0000,111.2500){\line(1,0){6.0000}}
\put(108.0000,111.2500){\line(1,0){6.0000}}
\put(114.0000,111.2500){\line(1,0){6.0000}}
\put(120.0000,111.2500){\line(1,0){6.0000}}
\put(126.0000,91.0000){\line(1,0){6.0000}}
\put(132.0000,91.0000){\line(1,0){6.0000}}
\put(138.0000,91.0000){\line(1,0){6.0000}}
\put(144.0000,91.0000){\line(1,0){6.0000}}
\put(-1.0000,78.8750){\normalsize\makebox(0,0)[r]{a1}}
\put(6.0000,68.7500){\line(0,1){20.2500}}
\put(12.0000,89.0000){\line(0,-1){20.2500}}
\put(30.0000,68.7500){\line(0,1){20.2500}}
\put(42.0000,89.0000){\line(0,-1){20.2500}}
\put(48.0000,68.7500){\line(0,1){20.2500}}
\put(60.0000,89.0000){\line(0,-1){20.2500}}
\put(66.0000,68.7500){\line(0,1){20.2500}}
\put(78.0000,89.0000){\line(0,-1){20.2500}}
\put(96.0000,68.7500){\line(0,1){20.2500}}
\put(126.0000,89.0000){\line(0,-1){20.2500}}
\put(0.0000,68.7500){\line(1,0){6.0000}}
\put(6.0000,89.0000){\line(1,0){6.0000}}
\put(12.0000,68.7500){\line(1,0){6.0000}}
\put(18.0000,68.7500){\line(1,0){6.0000}}
\put(24.0000,68.7500){\line(1,0){6.0000}}
\put(30.0000,89.0000){\line(1,0){6.0000}}
\put(36.0000,89.0000){\line(1,0){6.0000}}
\put(42.0000,68.7500){\line(1,0){6.0000}}
\put(48.0000,89.0000){\line(1,0){6.0000}}
\put(54.0000,89.0000){\line(1,0){6.0000}}
\put(60.0000,68.7500){\line(1,0){6.0000}}
\put(66.0000,89.0000){\line(1,0){6.0000}}
\put(72.0000,89.0000){\line(1,0){6.0000}}
\put(78.0000,68.7500){\line(1,0){6.0000}}
\put(84.0000,68.7500){\line(1,0){6.0000}}
\put(90.0000,68.7500){\line(1,0){6.0000}}
\put(96.0000,89.0000){\line(1,0){6.0000}}
\put(102.0000,89.0000){\line(1,0){6.0000}}
\put(108.0000,89.0000){\line(1,0){6.0000}}
\put(114.0000,89.0000){\line(1,0){6.0000}}
\put(120.0000,89.0000){\line(1,0){6.0000}}
\put(126.0000,68.7500){\line(1,0){6.0000}}
\put(132.0000,68.7500){\line(1,0){6.0000}}
\put(138.0000,68.7500){\line(1,0){6.0000}}
\put(144.0000,68.7500){\line(1,0){6.0000}}
\put(-1.0000,56.6250){\normalsize\makebox(0,0)[r]{a2}}
\put(12.0000,46.5000){\line(0,1){20.2500}}
\put(18.0000,66.7500){\line(0,-1){20.2500}}
\put(42.0000,46.5000){\line(0,1){20.2500}}
\put(48.0000,66.7500){\line(0,-1){20.2500}}
\put(84.0000,46.5000){\line(0,1){20.2500}}
\put(102.0000,66.7500){\line(0,-1){20.2500}}
\put(114.0000,46.5000){\line(0,1){20.2500}}
\put(120.0000,66.7500){\line(0,-1){20.2500}}
\put(132.0000,46.5000){\line(0,1){20.2500}}
\put(144.0000,66.7500){\line(0,-1){20.2500}}
\put(0.0000,46.5000){\line(1,0){6.0000}}
\put(6.0000,46.5000){\line(1,0){6.0000}}
\put(12.0000,66.7500){\line(1,0){6.0000}}
\put(18.0000,46.5000){\line(1,0){6.0000}}
\put(24.0000,46.5000){\line(1,0){6.0000}}
\put(30.0000,46.5000){\line(1,0){6.0000}}
\put(36.0000,46.5000){\line(1,0){6.0000}}
\put(42.0000,66.7500){\line(1,0){6.0000}}
\put(48.0000,46.5000){\line(1,0){6.0000}}
\put(54.0000,46.5000){\line(1,0){6.0000}}
\put(60.0000,46.5000){\line(1,0){6.0000}}
\put(66.0000,46.5000){\line(1,0){6.0000}}
\put(72.0000,46.5000){\line(1,0){6.0000}}
\put(78.0000,46.5000){\line(1,0){6.0000}}
\put(84.0000,66.7500){\line(1,0){6.0000}}
\put(90.0000,66.7500){\line(1,0){6.0000}}
\put(96.0000,66.7500){\line(1,0){6.0000}}
\put(102.0000,46.5000){\line(1,0){6.0000}}
\put(108.0000,46.5000){\line(1,0){6.0000}}
\put(114.0000,66.7500){\line(1,0){6.0000}}
\put(120.0000,46.5000){\line(1,0){6.0000}}
\put(126.0000,46.5000){\line(1,0){6.0000}}
\put(132.0000,66.7500){\line(1,0){6.0000}}
\put(138.0000,66.7500){\line(1,0){6.0000}}
\put(144.0000,46.5000){\line(1,0){6.0000}}
\put(-1.0000,34.3750){\normalsize\makebox(0,0)[r]{a3}}
\put(18.0000,24.2500){\line(0,1){20.2500}}
\put(24.0000,44.5000){\line(0,-1){20.2500}}
\put(60.0000,24.2500){\line(0,1){20.2500}}
\put(66.0000,44.5000){\line(0,-1){20.2500}}
\put(96.0000,24.2500){\line(0,1){20.2500}}
\put(102.0000,44.5000){\line(0,-1){20.2500}}
\put(114.0000,24.2500){\line(0,1){20.2500}}
\put(120.0000,44.5000){\line(0,-1){20.2500}}
\put(144.0000,24.2500){\line(0,1){20.2500}}
\put(0.0000,24.2500){\line(1,0){6.0000}}
\put(6.0000,24.2500){\line(1,0){6.0000}}
\put(12.0000,24.2500){\line(1,0){6.0000}}
\put(18.0000,44.5000){\line(1,0){6.0000}}
\put(24.0000,24.2500){\line(1,0){6.0000}}
\put(30.0000,24.2500){\line(1,0){6.0000}}
\put(36.0000,24.2500){\line(1,0){6.0000}}
\put(42.0000,24.2500){\line(1,0){6.0000}}
\put(48.0000,24.2500){\line(1,0){6.0000}}
\put(54.0000,24.2500){\line(1,0){6.0000}}
\put(60.0000,44.5000){\line(1,0){6.0000}}
\put(66.0000,24.2500){\line(1,0){6.0000}}
\put(72.0000,24.2500){\line(1,0){6.0000}}
\put(78.0000,24.2500){\line(1,0){6.0000}}
\put(84.0000,24.2500){\line(1,0){6.0000}}
\put(90.0000,24.2500){\line(1,0){6.0000}}
\put(96.0000,44.5000){\line(1,0){6.0000}}
\put(102.0000,24.2500){\line(1,0){6.0000}}
\put(108.0000,24.2500){\line(1,0){6.0000}}
\put(114.0000,44.5000){\line(1,0){6.0000}}
\put(120.0000,24.2500){\line(1,0){6.0000}}
\put(126.0000,24.2500){\line(1,0){6.0000}}
\put(132.0000,24.2500){\line(1,0){6.0000}}
\put(138.0000,24.2500){\line(1,0){6.0000}}
\put(144.0000,44.5000){\line(1,0){6.0000}}
\put(-1.0000,12.1250){\normalsize\makebox(0,0)[r]{a4}}
\put(24.0000,2.0000){\line(0,1){20.2500}}
\put(30.0000,22.2500){\line(0,-1){20.2500}}
\put(78.0000,2.0000){\line(0,1){20.2500}}
\put(84.0000,22.2500){\line(0,-1){20.2500}}
\put(96.0000,2.0000){\line(0,1){20.2500}}
\put(102.0000,22.2500){\line(0,-1){20.2500}}
\put(114.0000,2.0000){\line(0,1){20.2500}}
\put(120.0000,22.2500){\line(0,-1){20.2500}}
\put(0.0000,2.0000){\line(1,0){6.0000}}
\put(6.0000,2.0000){\line(1,0){6.0000}}
\put(12.0000,2.0000){\line(1,0){6.0000}}
\put(18.0000,2.0000){\line(1,0){6.0000}}
\put(24.0000,22.2500){\line(1,0){6.0000}}
\put(30.0000,2.0000){\line(1,0){6.0000}}
\put(36.0000,2.0000){\line(1,0){6.0000}}
\put(42.0000,2.0000){\line(1,0){6.0000}}
\put(48.0000,2.0000){\line(1,0){6.0000}}
\put(54.0000,2.0000){\line(1,0){6.0000}}
\put(60.0000,2.0000){\line(1,0){6.0000}}
\put(66.0000,2.0000){\line(1,0){6.0000}}
\put(72.0000,2.0000){\line(1,0){6.0000}}
\put(78.0000,22.2500){\line(1,0){6.0000}}
\put(84.0000,2.0000){\line(1,0){6.0000}}
\put(90.0000,2.0000){\line(1,0){6.0000}}
\put(96.0000,22.2500){\line(1,0){6.0000}}
\put(102.0000,2.0000){\line(1,0){6.0000}}
\put(108.0000,2.0000){\line(1,0){6.0000}}
\put(114.0000,22.2500){\line(1,0){6.0000}}
\put(120.0000,2.0000){\line(1,0){6.0000}}
\put(126.0000,2.0000){\line(1,0){6.0000}}
\put(132.0000,2.0000){\line(1,0){6.0000}}
\put(138.0000,2.0000){\line(1,0){6.0000}}
\put(144.0000,2.0000){\line(1,0){6.0000}}
\label{fig:simDetGODCSKBLenBurstInt}
\end{picture}
}
}

{\color{black}
\subsection{Automata synthesized for 2 cell arbiter}
\label{section:2cellexample}
Fig.~\ref{fig:monitor2cell} gives the monitor automaton for 2-cell arbiter (See section \ref{sec:casestudy}) for n-cell arbiter specification.
 Each transition is labeled by 4 bit vector giving values of $req_1, req_2, ack_1, ack_2$.

\begin{figure}[!h]
\centering
\includegraphics[width=\textwidth, keepaspectratio]{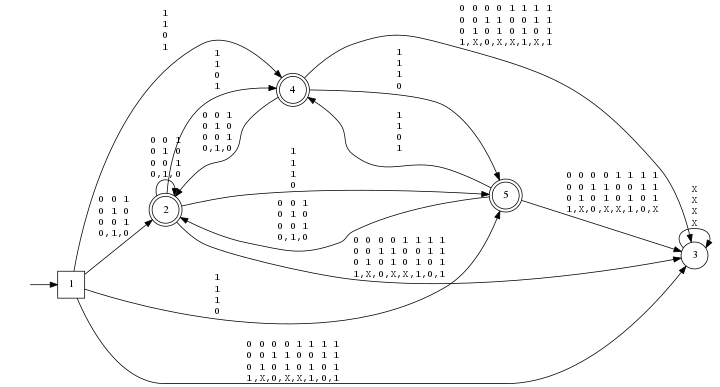}
\caption{Safety Monitor Automaton: 2 Cell Arbiter}
\label{fig:monitor2cell}
\end{figure}

\begin{figure}[h]
\begin{minipage}{2.5in}
\includegraphics[scale=.4,keepaspectratio]{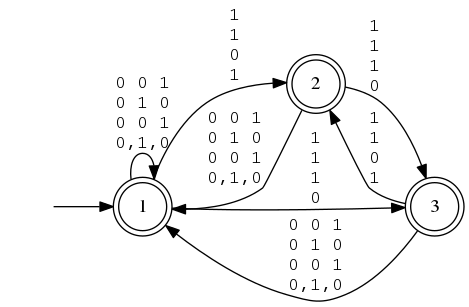}
\end{minipage} \ \ 
\begin{minipage}{2.5in}
\includegraphics[scale=.4,keepaspectratio]{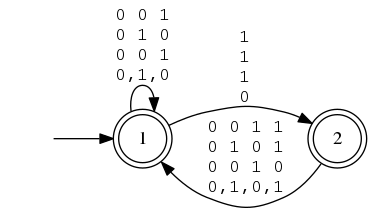}
\end{minipage}
\caption{Supervisors for $Arb^{hard}(2,2)$ 
(a): \MPNC  (b): \GODSC determinized}
\label{fig:mpnc}
\end{figure}

Fig.~\ref{fig:mpnc} gives the \MPNC automaton for the 2-cell arbiter computed from 
the safety monitor automaton of Fig.~\ref{fig:monitor2cell}. 
(There is an additional reject state. All missing transitions are
directed to it. These are omitted from the diagram for simplicity.) Note that this is a DFA 
whose transitions are labelled by 4-bit vectors representing alphabet $2^{\{req_1, req_2, ack_1, ack_2\}}$.
As defined in Definition \ref{def:nondetmm},
the DFA also denotes an output-nondeterministic Mealy machine with input variables $(req_1,req_2)$ and 
output variables $(ack_1,ack_2)$.
The automaton is nondeterministic in output 
as from state $1$, on input $(1,1)$ it can move to state $2$ with output $(1,0)$, 
or to state $3$ with output $(0,1)$. The reader can verify that the automaton is non-blocking and hence
a controller.

In 2-cell arbiter example, with soft requirements $\langle ack1, ack2 \rangle$ 
which give $ack1$ priority over $ack2$, 
we obtain the \GODSC controller automaton of Fig.~\ref{fig:mpnc}(b) from the \MPNC of Fig.~\ref{fig:mpnc}(a).
Note that we minimize the automaton at each step.
}
\clearpage
{\color{black}
\section{Mine pump \DCSYNTH Specification and Robust controllers}
\label{section:minepumpcasestudy}
In this section we illustrate the effect of {\em soft requirements} 
on the quality of the synthesized controllers with a case study of a {\em mine pump controller} specification \cite{Pan01a}. 
%
The controller has two input sensors: high water level sensor $\mathit{HH2O}$
and methane leakage sensor $\mathit{HCH4}$; 
and one output, 
$\mathit{PumpOn}$ to keep the pump on. 
The objective of the controller is to \emph{safely} operate 
the pump in such a way that 
the water level never remains high continuously for more
that $w$ cycles. 
\medskip

Thus, minepump controller has input and output variables  $(\{HH2O, HCH4\}, \{PumpOn\}) $.

\noindent We have following {\bf assumptions} on the mine and the pump.Their conjunction
is denote {\bf $MineAssume(\epsilon,\zeta,\kappa)$}.
\begin{itemize}

\item[-] \emph{Pump capacity:} $([]!(slen=\epsilon ~\&\& ~([[PUMPON ~\&\& ~HH2O]]$\verb|^|$\langle HH2O \rangle)))$.
If pump is continuously on for at least $\mathit{\epsilon+1}$ cycles, when water level is high, 
then water level will not be high at the end (i.~e.~ $\mathit{\epsilon+1}$ cycles). 
\item[-] \emph{Methane release:} 
$[](([HCH4]$\verb|^|$[!HCH4]$\verb|^|$\langle HCH4 \rangle ) \Rightarrow (slen> \zeta))$ and
$[]([[HCH4]] \Rightarrow slen<\kappa)$.
The minimum separation between the two leaks of methane is $\mathit{\zeta}$ cycles 
and the methane leak cannot persist for more than $\mathit{\kappa}$ cycles. 
\end{itemize}
The {\bf commitments} are:
\begin{itemize}
\item[-] \emph{Safety conditions:} $[[(HCH4 ~|| ~!HH2O) \Rightarrow ~!PumpOn)]]$ saying that
if there is a methane leakage or absence of high water, then pump should be off;
and $[]([[HH2O]] ~\&\& ~slen < w)$ stating that it is not possible that the water level continuously remains high for $w$ cycles.
\end{itemize}
The conjunction of commitments is denoted {\bf $MineCommit(w)$}.

It is to be noted that all the assumption and commitments formulas are made intermittent by restricting the scope of  each formula to the most recent interval of $n$ cycles i.e. formula is satisfied or violation is decided based on the recent past instead of whole past. 
This is done by using formula using $KBOUNDED(D,n)$ macro (See Figure \ref{sec:minepumpSpecSource}) defined as $((slen<K => (D)) ~\&\&~ (true \verb|^|
 slen=K => true \verb|^| (slen=K ~\&\&~ (D))))$.
This is important for robust specification, because if we consider the whole past then even a single violation of formula will keep it unsatisfiable always without allowing it to recover ever.
The Minepump specification {\bf $MinePump(w,\epsilon,\zeta,\kappa)$} is given by the assumption, commitment pair
$(MineAssume(\epsilon,\zeta,\kappa), MineCommit(w))$. 
Appendix \ref{sec:minepumpSpecSource} gives the textual source of $(MinePump(8,2,6,2))$ for $BeCurrentlyCorrect$
robust specification used by the \DCSYNTH\/ tool.

\subsection{Must Dominance betweeen various robust Controllers}
\subsubsection{Comparision based on Must Dominance}
We use "$==$" to indicate that the supervisors are \emph{identical}, whereas "$=$" indicates that they are \emph{must equivalent} but not identical.
For $Minepump(8,2,6,2)$ case study some of the supervisors become identical as given in Appendix \ref{section:minepumpcasestudy} Figure \ref{fig:mineMPNCRobustnessHierarchy}.

\begin{itemize}
\item $\GODSC$ supervisors does not have theoretical order. However, for $Minepump(8,2,6,2)$ the complete robustness order collapses and all supervisors becomes must equivalent. Although, they are not identical (See Appendix Figure \ref{fig:mineGODSCRobustnessHierarchy}).
The $\_nonInt$ represents that all the non intermittent specification i.e. ResCnt, LenCnt, ResBurst and LenBurst.

\end{itemize}

\begin{figure}
\begin{tikzpicture}

    \node[ellipse] (5) at (0,2) {BeCurrentlyCorrect};
        \node[ellipse] (7) at (0,1) {LenBurstInt};
        \node[ellipse] (6) at (0, 0)  {LenCntInt};
    \node[ellipse] (4) at (0, -1)  {ResCntInt==ResBurstInt};
    \node[ellipse] (3) at (0, -2) {ResCnt == LenCnt == ResBurst==LenBurst};
    \node[ellipse] (1) at (0, -3) {BeCorrect};
	\draw [->] (1) edge (3);
    \draw [->] (3) edge (4);
    \draw [->] (4) edge (6);  
    \draw [->] (6) edge (7);  
    \draw [->] (7) edge (5);    
\end{tikzpicture}
\caption{Robust synthesis criteria hierarchy for Must Dominance between \MPNC.}
\label{fig:mineMPNCRobustnessHierarchy}

\end{figure}

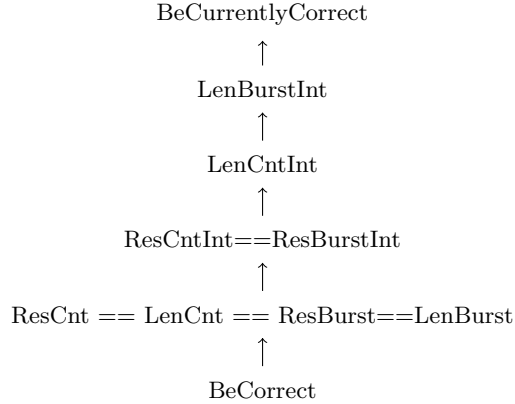
\begin{figure}
\oomit{
\centering
\begin{tikzpicture}
    \node[ellipse] (1) at (0, 0) {AssumeFalse=BeCorrect=
    (KB\_nonInt==)=(KBResCntInt==KBResBurstInt)=
    (KBLenCntInt=KBLenBurstInt)
    =BeCurrentlyCorrect};
\end{tikzpicture}
}
\begin{center}
 AssumeFalse=BeCorrect=
    (\_nonInt==)=(ResCntInt==ResBurstInt)=
    LenCntInt=LenBurstInt
    =BeCurrentlyCorrect
\end{center}
\caption{Robust synthesis criteria hierarchy for Must Dominance between \GODSC.}
\label{fig:mineGODSCRobustnessHierarchy}

\end{figure}

\subsection{Expected Case Performance}
Expected values of various robust controllers are given in the Table \ref{tab:mineExpectedValuesPumpOn} and \ref{tab:mineExpectedValuesPumpOff}. It is evident that expected values of meeting the commitment using hard robustness (given as \MPNC in column 2) as well as soft robustness (given as \GODSC in column 4) are useful and maximize the holding of commitments.

\begin{table}
\caption{Expected Values of meeting the commitments for various robust controllers of Minepump case study with default PumpOn}
\label{tab:mineExpectedValuesPumpOn}
\begin{tabular}{|c|c|c|c|c|}
\hline
Minepump(PumpOn) & \multicolumn{2}{c|}{\MPNC}& \multicolumn{2}{c|}{\GODSC}
\\
\hline
& $E(C)$ & $E(A)$ & $E(C)$ & $E(A)$\\
\hline
 AssumeFalse	& 0.00000 & 0.012257 & 0.997070	& 0.014441
\\
\hline

 BeCorrect	& 0.000000	& 0.0122567 & 0.997070	& 0.014441
\\
\hline

		
 ResCnt(2,8)	& \multirow{4}{*}{0.000000}	& \multirow{4}{*}{0.012257} & \multirow{4}{*}{0.997070}	& \multirow{4}{*}{0.014441}
\\
\cline{1-1}
ResBurst(2,8)	& & & & \oomit{0.000000	& 0.012257 & 0.997070	& 0.014441}
\\
\cline{1-1}

 LenCnt(2,8)	& & & & \oomit{0.000000	& 0.012257 & 0.997070	& 0.014441}
\\
\cline{1-1}

LenBurst(2,8)	& & & & \oomit{0.000000	& 0.012257 & 0.997070	& 0.014441}
\\
\hline

 ResCntInt(2,8)	& \multirow{2}{*}{0.000966}	& \multirow{2}{*}{0.013247} &	\multirow{2}{*}{0.997070}	& \multirow{2}{*}{0.014441}
\\
\cline{1-1}
		
 ResBurstInt(2,8)	& & & & \oomit{ 0.000966	& 0.013247 & 0.997070	& 0.014441}
\\
\hline

 LenCntInt(2,8)	& 0.0027342	& 0.013318 & 0.997070	& 0.014441
\\
\hline

 LenBurstInt(2,8)	& 0.004514	& 0.013321 & 0.997070	& 0.014441
\\
\hline

 BeCurrentlyCorrect &	0.997070	& 0.014441 &	0.997070	& 0.014441
\\
\hline

\end{tabular}
\end{table}

\begin{table}
\caption{Expected Values of meeting the commitments for various robust controllers of Minepump case study with default PumpOff}
\label{tab:mineExpectedValuesPumpOff}
\begin{tabular}{|c|c|c|c|c|}
\hline

Minepump (PumpOff) & \multicolumn{2}{c|}{\MPNC}& \multicolumn{2}{c|}{\GODSC}  
\\
\hline
AssumeFalse & 0.997070 & 0.024903	& 0.997070	& 0.024903
\\
\hline

 BeCorrect	& 0.997070	& 0.024904 & 0.997070	& 0.024904
\\
\hline




		
 ResCnt(2,8)	& \multirow{4}{*}{0.997070}	& \multirow{4}{*}{0.024904} & \multirow{4}{*}{0.997070}	& \multirow{4}{*}{0.024904}
\\
\cline{1-1}

ResBurst(2,8)	& & & & \oomit{ 0.997070	& 0.024904 & 0.997070	& 0.024904}
\\
\cline{1-1}
		
 LenCnt(2,8)	& & & & \oomit{0.997070	& 0.024904 & 0.997070	& 0.024904}
\\
\cline{1-1}

LenBurst(2,8)	& & & & \oomit{0.997070	& 0.024904 & 0.997070	& 0.024904}
\\
\hline
		
 ResCntInt(2,8)	& \multirow{2}{*}{0.997070}	& \multirow{2}{*}{0.024662} & \multirow{2}{*}{0.997070} & \multirow{2}{*}{0.024662}
\\
\cline{1-1}
 ResBurstInt(2,8)	& & & & \oomit{0.997070 &	0.024662 &	0.997070	& 0.024662}
\\
\hline

 LenCntInt(2,8)	& 0.997070 &	0.024206 & 0.997070 &	0.024206
\\
\hline
 LenBurstInt(2,8)	& 0.997070 &	0.023268 & 0.997070 &	0.023268
\\
\hline
  BeCurrentlyCorrect	& 0.997070	& 0.024293 & 0.997070 & 0.024293
\\
\hline


\end{tabular}
\end{table}

\subsection{Tool Performance}
Table \ref{tab:performanceRobustMinepumpSynthesis} provides the detailed statistics for synthesizing various robust controllers. It can be seen that monitor automaton computation and \GODSC computation takes the maximum time and  total for computing the controller is 3 seconds.

\begin{scriptsize}

\begin {table}[t]
\caption {
\DCSYNTH Statistics for synthesis of various Robust Minepump example $Minepump(8,2,6,2)$.
Parameters Used: Disc Factor = 0.900000, H = 50, Type = Avg-Max, Delta = 0.000100} 
\label{tab:performanceRobustMinepumpSynthesis}
\begin{center}
	\begin{tabular}	{|c|c||c|c|c|c|c|}
	\hline
	& 
	& \multicolumn{5}{c|}	 {\textbf{ Synthesis (States/Time)}}
	 
	\\
	\hline
		\textbf{Sr} &
		\textbf{Robustness}  & 
		\textbf{Monitor}  & 
		\textbf{\MPNC} & 
		\textbf{\GODSC} &
	 	\multicolumn{2}{c|}{\textbf{Controller Stats}} 
		\\
	\cline{6-7}
		\textbf{No} &
		 \textbf{Criterion} &
		\textbf{Stats} &
		\textbf{Stats} &
		\textbf{Stats} &
		\textbf{PumpOn} &
		\textbf{PumpOff}  
		\\
	     \hline
	     \multicolumn{7}{|c|}{\MPNC}\\
	     \hline
	    & AssumeFalse      & 5268/0.413 &    5267/0.051886 &      1/0.203156 &      2/0.000059  &     2/0.000042  
\\
	     \hline

&  BeCorrect     & 6299/0.460 &    5390/0.051482 &      61/0.203690 &      21/0.003270  &      47/0.002950  
\\
	     \hline

&  BeCurrentlyCorrect & 6102/0.416 &    905/0.025670 &      30/0.020076 &      2/0.000397  & 25/0.000496  
\\
	     \hline

& LenBurst(2,8)      & 7689/0.489 &    5404/0.055062 &      45/0.209456 &      22/0.000437  &   37/0.000461  

\\
	     \hline

& LenBurstInt(2,8)      & 12169/0.754 &    6067/0.086382 &      342/0.210206 &      219/0.009601  &     51/0.008175  
\\
	     \hline

& LenCnt(2,8)      & 7689/0.495 &    5404/0.055288 &      45/0.208896 &      22/0.000849  &     37/0.000442  

\\
	     \hline

&  LenCntInt(2,8)     & 11971/0.813 &    6671/0.088817 &      420/0.232484 &      257/0.012043  &      71/0.011698  
\\
	     \hline

& ResBurst(2,8)      & 7689/0.469 &    5404/0.057592 &      45/0.209711 &      22/0.000633  &      37/0.000431  
\\
	     \hline

& ResBurstInt(2,8)      & 13723/0.870 &    8430/0.098644 &      469/0.288897 &      306/0.013445  &       88/0.007452  
\\
	     \hline

& ResCnt(2,8)      & 7689/0.468 &    5404/0.055485 &      45/0.208627 &      22/0.000719  &       37/0.000672  
\\
	     \hline

&  ResCntInt(2,8)      & 13723/0.883 &    8430/0.100492 &      469/0.287158 &      306/0.008003  &       88/0.007359  
\\
	     \hline
\oomit{
& SepBurst(2,8)      & 9016 &    5729/0.060334 &      210/0.215072 &      101/0.017330  &       38/0.016691  
\\
	     \hline
}
	     	     
		\hline
		\multicolumn{7}{|c|}{\GODSC}\\
	     \hline

& AssumeFalse      & 5268/0.413 &    5267/0.051461 &      2/1.215522 &      2/0.000054  &       2/0.000068 

\\
	     \hline

 &  BeCorrect      & 6299/0.468 &    5390/0.053660 &      61/1.246768 &      2/0.000471  &       47/0.001097 

\\
	     \hline

 &  BeCurrentlyCorrect      & 6102/0.413 &    905/0.025486 &      30/0.133132 &      2/0.000500  &       25/0.000490  

\\
	     \hline
 
&  LenBurst(2,8)      & 7689/0.532 &    5404/0.053116 &      44/1.229908 &      2/0.000566  &      37/0.000478 

\\
	     \hline

 & LenBurstInt(2,8)      & 12169/0.764 &    6067/0.085852 &      76/1.237215 &      2/0.000996  &       51/0.001041

\\
	     \hline

& LenCnt(2,8)      & 7689/0.516 &    5404/0.053728 &      44/1.239951 &      2/0.000705  &       37/0.000455  

\\
	     \hline

& LenCntInt(2,8)      & 11971/0.786 &    6671/0.089755 &      116/1.374660 &      2/0.001376  &       71/0.001521  

\\
	     \hline

& ResBurst(2,8)      & 7689/0.523 &    5404/0.053781 &      44/1.322684 &      2/0.000595  &       37/0.000748 

\\
	     \hline

 & ResBurstInt(2,8)      & 13723/0.875 &    8430/0.100483 &      156/1.817073 &      2/0.002108  &       88/0.002476  

\\
	     \hline

& ResCnt(2,8)     & 7689/0.479 &    5404/0.053901 &      44/1.230389 &      2/0.000566  &      37/0.000452 

\\
	     \hline

 & ResCntInt(2,8)     & 13723/0.884 &    8430/0.104372 &      156/1.810811 &      2/0.002319  &       88/0.001283  

\\
	     \hline
\oomit{
& SepBurst(2,8)      & 9016 &    5729/0.058311 &      46/1.293318 &      2/0.000552  &       38/0.000703  
\\
	     \hline
}

	\end{tabular}
\end{center}
\end{table}

\end{scriptsize}


\oomit{
Let  $\mathbf{mpsr1}$ denote  
$!(true \textrm{\textasciicircum} (slen=2 ~\&\& \langle \rangle \langle HCH4 \rangle)$\\
$\textrm{\textasciicircum}\langle PumpOn \rangle ~)$,
we also define the indicator variables $as1$ and $c1$,
for $ASSUME$ and $COMMIT$ formulae defined above.
These indicator variables track the validity of assumption and commitment respectively at any given point
in an infinite run.
which states that it is not the case that there is a methane leakage somewhere in the last 3 cycles and the pump is still on.
\begin{itemize}
\item \emph{MPV1}: Soft requirement $\langle PumpOn:2 \rangle$ states that try to keep {\em pump  on} as much as possible.
\item \emph{MPV2}: Soft requirement $\langle \mathit{mpsr1}:4,~~ \mathit{PumpOn}:2 \rangle$ states that try to keep {\em pump off} if there is a methane leak in the last 3 cycles; 
otherwise try to keep pump on. 
\item \emph{MPV3}: Soft requirement $\langle !PumpOn:2 \rangle$ states that try to keep {\em pump  off} as much as possible.
\item \emph{MPV4}: Soft requirement $\langle c1:2 \rangle$ states that try to meet {\em commitment} as much as possible.
\end{itemize}
}

\oomit{
We have synthesized the controllers for the specification $MinePump(8,2,6,2)$  and $MinePump^{soft}(8,2,6,2)$. The performance of the deterministic controllers (obtained by setting the default value of $\mathit{PumpOn}$ and $\mathit{PumpOff}$ respectively)  are compared using the expected value
of meeting the requirement $\mathit{REQUIREMENT}$.
}

\oomit{
Appendix \ref{sec:minepumpInputAndSimulation} gives textual
input to the tool and simulations carried out using the synthesized controllers.

It can be argued that these controllers have different quality attributes. 
For example, 
$\mathit{MPV1}$ gives rise to a controller that 
aggressively gets rid of water by keeping pump on whenever possible.
$\mathit{MPV3}$ saves power by keeping pump off as much as possible. On the other hand, 
$\mathit{MPV2}$ aggressively keeps pump on but it opts for a safer policy of not keeping pump on for two cycles even after methane is gone. 
}

\oomit{
\begin {table}[!h]
\caption {\DCSYNTH\/ synthesis for the hard requirement $MinePump(8,2,10,1,1)$ with soft requirements $MPV1$, $MPV2$ and $MPV3$. }
\label{tab:comparisionopt}
        \begin{center}
        \begin{tabular}
        {|c|c|c|c|} 
                \hline
                \multirow{2}{*}{Soft Requirement}  &  \multicolumn{3}{|c|}{ Controller Synthesis}\\
                \cline{2-4}
                &  states & time (Sec) & Memory (MB) \\
                \hline      
%
                $MPV1$  & 31 &0.07  & 9.1 \\
                \hline  
                $MPV2$  & 34 &0.09 & 8.9 \\
                \hline
                $MPV3$ & 83 &0.04  & 9.1 \\
                \hline            
	\end{tabular}%
\end{center}
\vspace{-0.5cm}
\end{table}
\begin {table}[!h]
\caption {Worst Case Latency Analysis using CTLDC for MAXLEN computation}
\label{tab:perfMeasure2}
\begin{center}
	\begin{tabular}	{|c|c|c|}
	\hline
		Soft Requirement & Response Formula ($D^p$) & MAXLEN value \\
	 \hline
%
		$MPV1$ & $[[AssumptionOk\ \&\&\ HH2O]]$ & 4 \\
	 \hline
	    $MPV2$ & $[[AssumptionOk\ \&\&\ HH2O]]$ & 7 \\
	 \hline
	    $MPV3$ & $[[AssumptionOk\ \&\&\ HH2O]]$ & 8 \\
	 \hline
	\end{tabular}
\vspace{-1cm}
\end{center}
\end{table}

\subsection{Quantitative Latency Measurement}
\label{subsection:quantitative}

A model checking technique, implemented in a tool CTLDC \cite{Pan05}, can
measure the worst case longest/shortest span of a \qddc formula $D^p$ in a system $M$.
This involves symbolic search for longest/shortest paths.
Formally,
%
$MAXLEN(D^{p},M)=
\sup \{e-b \mid \rho[b,e] \models D^{p}, \rho \in Exec(M) \}$ which computes the length of
the longest interval satisfying $D^{p}$ within the executions of $M$. 
Similarly, for $MINLEN(D^{p},M)$. Thus, 
CTLDC allows measurement of {\em user defined latencies} in the worst case.
This can be used to evaluate the performance (latency) of synthesized controllers.
%
%
Table~\ref{tab:perfMeasure2} gives worst case latency measurements carried out 
using tool CTLDC for the mine pump controllers with soft requirements $MPV1$, $MPV2$ and $MPV3$.
Clearly, $MPV1$ gets rid of water the fastest (in atmost 4 cycles) 
whereas $MPV3$ is the slowest in this sense (requiring up to 8 cycles).

}
}
{\color{black}
\subsection{Mine pump Specification Source}
\label{sec:minepumpSpecSource}
\begin{figure}[!b]
\caption{Mine pump specification in \DCSYNTH}
\begin{small}
\framebox{\parbox[t][][t]{\columnwidth}{
$\begin{array}{l}
\mathrm{\textsf{\#qsf "minepump"}}\\
\mathrm{\textsf{interface\{}}\\
\quad\mathrm{\textsf{input HH2Op, HCH4p;}}\\
\quad\mathrm{\textsf{output PUMPONp monitor x, A monitor x, C monitor x;}}\\
\quad\mathrm{\textsf{constant w = 8, epsilon=2 , zeta=6, kappa=2;}}\\
\mathrm{\textsf{\}}}\\
\mathrm{\textsf{definitions\{}}\\
\mathrm{\textsf{//Methane release assumptions}}\\
\mathrm{\textsf{dc methane1(HCH4)\{}}\\
\quad\mathrm{\textsf{KBOUNDED([]([HCH4]\textasciicircum[!HCH4]\textasciicircum \textless HCH4 \textgreater =\textgreater slen\textgreater zeta ), n);}}\\
\mathrm{\textsf{\}}}\\
\mathrm{\textsf{dc methane2(HCH4)\{}}\\
\quad\mathrm{\textsf{KBOUNDED([]([[HCH4]] =\textgreater slen\textless kappa ), n);}}\\
\mathrm{\textsf{\}}}\\
\mathrm{\textsf{//Pump capacity assumption}}\\
\mathrm{\textsf{dc pumpcap1(HH2O, PUMPON)\{}}\\
\quad\mathrm{\textsf{KBOUNDED([]!(slen=epsilon \&\& ([[PUMPON \&\& HH2O]] \textasciicircum \textless HH2O \textgreater)),n);}}\\
\mathrm{\textsf{\}}}\\
\mathrm{\textsf{dc MineAssume\_2\_6\_2(HH2O, HCH4, PUMPON)\{}}\\
\quad\mathrm{\textsf{methane1(HH2O, HCH4, PUMPON) \&\& methane2(HH2O, HCH4, PUMPON) \&\&}}\\ 
\quad\mathrm{\textsf{pumpcap1(HH2O, HCH4, PUMPON);}}\\ 
\mathrm{\textsf{\}}}\\
\mathrm{\textsf{//safety condition}}\\ 
\mathrm{\textsf{dc req1(HH2O, HCH4, PUMPON)\{}}\\
\quad\mathrm{\textsf{KBOUNDED(([[( (HCH4 $||$ !HH2O) =\textgreater !PUMPON)]]),n);}}\\
\mathrm{\textsf{\}}}\\
\mathrm{\textsf{dc req2(HH2O, HCH4, PUMPON)\{}}\\
\quad\mathrm{\textsf{KBOUNDED((!([] ([[HH2O]] \&\& (slen = w)))),n)}}\\
\mathrm{\textsf{\}}}\\
\mathrm{\textsf{dc MineCommit\_8(HH2O, HCH4, PUMPON)\{}}\\
\quad\mathrm{\textsf{req1(HH2O, HCH4, PUMPON) \&\& req2(HH2O, HCH4, PUMPON);}}\\
\mathrm{\textsf{\}}}\\
\mathrm{\textsf{indefinitions\{}}\\
\quad\mathrm{\textsf{A : MineAssume\_8(HH2Op, HCH4p, PUMPONp);}}\\
\quad\mathrm{\textsf{C : MineCommit\_8(HH2Op, HCH4p, PUMPONp);}}\\
\mathrm{\textsf{\}}}\\
\mathrm{\textsf{hardreq\{}}\\
\quad\mathrm{\textsf{MineAssume\_2\_6\_2(HH2Op, HCH4p, PUMPONp) =\textgreater}} \\
\quad\quad\mathrm{\textsf{MineCommit\_8(HH2Op, HCH4p, PUMPONp);}}\\
\mathrm{\textsf{hardreq\{}}\\
\quad\mathrm{\textsf{useind A, C;}}\\
\quad\mathrm{\textsf{BeCurrentlyCorrect(A) =\textgreater EP(C);}}\\
\mathrm{\textsf{\}}}\\
\mathrm{\textsf{softreq\{}}\\
\quad\mathrm{\textsf{useind C;}}\\
\quad\mathrm{\textsf{(C);}}\\
\mathrm{\textsf{\}}}\\
\end{array}$}
}
\end{small}
\end{figure}

\oomit{
\clearpage
\textbf{Simulation of Synthesized Mine pump Controllers}
The controllers are encoded as Lustre specification and Lustre V4 tools are used 
for simulation. The example simulation for these three variants of mine pump
with soft requirements MPV1, MPV2 and MPV3 are shown in figures 
~\ref{fig:pumptrue}, ~\ref{fig:pumpsoft} and ~\ref{fig:pumpfalse} respectively.

\begin{figure*}[!h]
\centering
\includegraphics[width=\textwidth, keepaspectratio]{PumpTrue.png}
\caption{Simulation of mine pump controller with soft requirement MPV1}
\label{fig:pumptrue}
\centering
\includegraphics[width=\textwidth, keepaspectratio]{PumpSoft.png}
\caption{Simulation of mine pump controller with soft requirement MPV2}
\label{fig:pumpsoft}
%
\centering
\includegraphics[width=\textwidth, keepaspectratio]{PumpFalse.png}
\caption{Simulation of mine pump controller with soft requirement MPV3}
\label{fig:pumpfalse}
\end{figure*}
}

}

\end{document}